\documentclass[a4paper,11pt]{article}
\usepackage{microtype}
\usepackage{tikz}
\usetikzlibrary{arrows,chains,matrix,positioning,scopes,decorations,decorations.pathreplacing}
\makeatletter
\tikzset{join/.code=\tikzset{after node path={%
\ifx\tikzchainprevious\pgfutil@empty\else(\tikzchainprevious)%
edge[every join]#1(\tikzchaincurrent)\fi}}}
\makeatother
\tikzset{>=stealth',every on chain/.append style={join},
         every join/.style={->}}
\tikzstyle{labeled}=[execute at begin node=$\scriptstyle,
execute at end node=$]
\tikzstyle{vertex}=[circle, draw, fill=blue!20, inner sep=0pt, minimum width=15pt]
\tikzstyle{finalvertex}=[circle, draw, fill=red!30, inner sep=0pt, minimum width=15pt]
\tikzstyle{initialvertex}=[circle, draw, fill=green!30, inner sep=0pt, minimum width=15pt]
\tikzstyle{en}=[text=blue]
\author{Suryajith Chillara\thanks{Department of CSE, IIT Bombay. \texttt{suryajith@cse.iitb.ac.in}}
  \and
  Christian Engels\thanks{Department of CSE, IIT Bombay.  \texttt{christian@cse.iitb.ac.in}}
  \and
  Nutan Limaye\thanks{ Department of CSE, IIT Bombay. \texttt{nutan@cse.iitb.ac.in}}
  \and
  Srikanth Srinivasan\thanks{Department of Mathematics, IIT Bombay. \texttt{srikanth@math.iitb.ac.in}}}

\usepackage{amsmath,amssymb,amsthm,algorithm,algorithmic}
\usepackage{fullpage}
\usepackage[colorlinks]{hyperref}
\usepackage[capitalize]{cleveref}
\hypersetup{
  linkcolor=[rgb]{0.3,0.3,0.6},
  citecolor=[rgb]{0.2, 0.6, 0.2},
  urlcolor=[rgb]{0.6, 0.2, 0.2}
}

\newtheorem{theorem}{Theorem}
\newtheorem{corollary}[theorem]{Corollary}
\newtheorem{lemma}[theorem]{Lemma}

\newtheorem{proposition}[theorem]{Proposition}

\newtheorem{claim}[theorem]{Claim}

\newtheorem{remark}[theorem]{Remark}

\newcommand{\prob}[2]{\mathop{\mathrm{Pr}}_{#1}[#2]}
\newcommand{\avg}[2]{\mathop{\textbf{E}}_{#1}[#2]}

\newcommand{\abs}[1]{\left|#1\right|}

\newcommand{\notwhat}{0.5}
\newcommand{\what}{0.5}

\newcommand{\mc}[1]{\mathcal{#1}}
\newcommand{\rank}{\mathrm{rank}}

\newcommand{\unbias}{\mathrm{Unbias}}
\newcommand{\bias}{\mathrm{Bias}}

\newtheorem*{rep@theorem}{\rep@title}
\newcommand{\newreptheorem}[2]{%
\newenvironment{rep#1}[1]{%
 \def\rep@title{#2 \ref{##1}}%
 \begin{rep@theorem}}%
 {\end{rep@theorem}}}
\makeatother

\newreptheorem{theorem}{Theorem}
\newreptheorem{lemma}{Lemma}

\newcommand{\F}{\mathbb{F}}

\newcommand{\Vars}{\mathrm{Vars}}

\newcommand{\supp}{\mathrm{Supp}}
\newcommand{\M}{{M}}
\newcommand{\wt}{\mathrm{wt}}
\newcommand{\Img}{\mathrm{Img}}

\newcommand{\ABP}{\mathrm{ABP}}
\newcommand{\mlABP}{\mathrm{mlABP}}
\newcommand{\relrk}{\mathrm{relrk}}

\title{A Near-Optimal Depth-Hierarchy Theorem for Small-Depth Multilinear Circuits}
\begin{document}

\maketitle
\thispagestyle{empty}
\begin{abstract}
We study the size blow-up that is necessary to convert an algebraic circuit of product-depth $\Delta+1$ to one of product-depth $\Delta$ in the multilinear setting.

We show that for every positive $\Delta = \Delta(n) = o(\log n/\log \log n),$ there is an explicit multilinear polynomial $P^{(\Delta)}$ on $n$ variables that can be computed by a multilinear formula of product-depth $\Delta+1$ and size $O(n)$, but not by any multilinear circuit of product-depth $\Delta$ and size less than $\exp(n^{\Omega(1/\Delta)})$. This result is tight up to the constant implicit in the double exponent for all $\Delta = o(\log n/\log \log n).$

This strengthens a result of Raz and Yehudayoff (Computational Complexity 2009) who prove a quasipolynomial separation for constant-depth multilinear circuits, and a result of Kayal, Nair and Saha (STACS 2016) who give an exponential separation in the case $\Delta = 1.$

Our separating examples may be viewed as algebraic analogues of variants of the Graph Reachability problem studied by Chen, Oliveira, Servedio and Tan (STOC 2016), who used them to prove lower bounds for constant-depth \emph{Boolean} circuits.
\end{abstract}

\section{Introduction}

This paper deals with a question in the area of \emph{Algebraic Complexity,} which studies the Computational Complexity of any algorithmic task that can be cast as the problem of computing a fixed multivariate polynomial (or polynomials) $f\in \F[x_1,\ldots,x_N]$ on a given $a\in \F^N.$  Many fundamental problems such as the Determinant, Permanent, the Fast Fourier Transform and Matrix Multiplication can be captured in this general paradigm. The natural computational model for solving such problems are Algebraic Circuits (and their close relations Algebraic Formulas and Algebraic Branching Programs) which use the algebraic operations of the polynomial ring $\F[x_1,\ldots,x_N]$ to compute the given polynomial. 

Our main focus in this paper is on \emph{Small-depth} Algebraic circuits, which are easily defined as follows.\footnote{We actually define algebraic \emph{formulas} here, but the distinction is not too important for our results in this paper.} Recall that any multivariate polynomial $f\in \F[x_1,\ldots,x_N]$ can be written as a linear combination of terms which are products of variables (i.e.\ monomials); we call such an expression a $\Sigma\Pi$ circuit for $f$. In general, such an expression for $f$ could be prohibitively large. More compact representations for $f$ may be obtained if we consider representations as linear combinations of products of linear functions, which we call $\Sigma\Pi\Sigma$ circuits, and subsequent generalizations such as $\Sigma\Pi\Sigma\Pi$ circuits and so on. We consider representations of the form $\Sigma\Pi\cdots O_d$ for some $d = d(N)$ where $d$ is a slow growing function of the number of variables $N$. We consider such a representation of a polynomial $f$ as a computation of $f$. 

The efficiency of such a computation is captured by the following two complexity measures: the \emph{size} of the corresponding expression, which captures the number of operations used in computing the polynomial; and the \emph{product-depth} of the circuit, which is the number of $\Pi$s in the expression for the circuit class (e.g. $1$ for $\Sigma\Pi\Sigma$ circuits, $2$ for $\Sigma\Pi\Sigma\Pi$ circuits etc.) and which measures in some sense the inductive complexity of the corresponding computation.\footnote{It is also standard in the literature to consider the \emph{depth} of the circuit, which is the number of $\Sigma$ and $\Pi$ terms in the defining expression for the circuit class, but the product-depth is frequently nicer (invariant under simple operations such as linear transformations etc.) and is essentially $\lfloor \text{depth}/2\rfloor$. So we mostly use product-depth. We also state our results in terms of depth.}

This paper is motivated by questions in a general body of results in the area that go by the name of \emph{Depth-reduction.} Informally, the question is: what is the worst case blow-up in size required to convert a circuit of product-depth $\Delta$ to one of product-depth $\Delta' < \Delta$? This is an important question, since circuits of smaller depth are frequently easier to analyze and understand, and such results allow us to transfer this understanding to more complex (i.e. higher depth) classes of circuits. Consequently, there have been many results addressing this general question for various models of Algebraic (and also Boolean) computation including Algebraic and Boolean formulas~\cite{brent, spira}, Boolean circuits~\cite{HopcroftPaulValiant}, and Algebraic circuits~\cite{VSBR,AV,Koiran,Tav13,GKKSdepth3}. Recently, Tavenas~\cite{Tav13} and Gupta, Kamath, Kayal and Saptharishi~\cite{GKKSdepth3} (building on~\cite{VSBR,AV,Koiran}) showed that a strong enough exponential lower bound for $\Sigma\Pi\Sigma$ circuits would imply a superpolynomial lower bound for general algebraic circuits. Interesting impossibility results are also known in this direction: for instance, it is known that the result of Tavenas~\cite{Tav13}, which converts a general circuit to a $\Sigma\Pi\Sigma\Pi$ circuit of subexponential size, cannot be improved in the restricted model of \emph{homogeneous}\footnote{I.e. a circuit where every intermediate expression computes a homogeneous polynomial.} $\Sigma\Pi\Sigma\Pi$ circuits~\cite{GKKSdepth4,FLMS,KLSS,KShom}.

Here, we study a more fine-grained version of the question of depth-reduction. We ask: what is the size blow-up in converting a circuit of product-depth $\Delta+1$ to one of product-depth $\Delta$? A natural strategy to carry out such a depth-reduction for a given $(\Sigma\Pi)^{\Delta+1}\Sigma$\footnote{I.e. a $\Sigma\Pi\cdots \Sigma$ expression with exactly $(\Delta+1)$ many $\Pi$s.} expression $F$ is to take some product terms in the expression and interchange them with the inner sum terms via the distributive law. This creates a blow-up in the size of the expression that is exponential in the number of sum terms. It is not hard to show that by choosing the sum terms carefully, one can limit this blow-up to $\exp(s^{1/\Delta+o(1)})$ for constant $\Delta,$\footnote{In fact, the upper bound is of the form $\exp(s^{1/\Delta}\log s)$, and is the $\exp(s^{1/\Delta+o(1)})$ as long as $\Delta = o(\log s/\log \log s).$} where $s$ is the size of the expression of depth $\Delta+1.$

Is this exponential\footnote{Strictly speaking, we get an exponential blow-up only for \emph{constant} $\Delta$. The careful reader should read ``exponential'' as ``exponential in $s^{1/\Delta}$'' for general $\Delta.$} blow-up unavoidable for any $\Delta$? The evidence we do have seems to suggest the answer is yes. For example, in the Boolean setting (where the $\Sigma$ and $\Pi$ are replaced by their Boolean counterparts $\bigvee$ and $\bigwedge$), such an exponential \emph{Depth-hierarchy theorem} is a classical result of H\r{a}stad~\cite{Hastadthesis}, with recent improvements by Rossman, Servedio, Tan and H\r{a}stad~\cite{RST15,Has16,HRSTjournal}. Even in the algebraic setting, there have been partial results in this direction~\cite{nw1997,ry09,KNS16,KumarSap16}. For \emph{homogeneous} circuits such a result was proved for $\Delta=1$ in  the work of Nisan and Wigderson~\cite{nw1997}. 

We study this question in the \emph{multilinear} setting, where the circuits are restricted to computing \emph{multilinear} polynomials at each stage of computation (multilinear polynomials are polynomials where the degree of each variable is at most $1$). This is a fairly natural model of computation for multilinear polynomials, and has been extensively studied in the literature \cite{Raz,Raznc2nc1,ry08,ry09,RSY08,DMPY12,KNS16}. 

Raz and Yehudayoff~\cite{ry09} considered the problem of separating product-depth $\Delta+1$ circuits from product-depth $\Delta$ circuits in the multilinear setting and proved a \emph{superpolynomial} separation implying a superpolynomial depth-hierarchy theorem. More precisely, they showed that there are circuits of size $s$ and product-depth $\Delta+1$ such that any product-depth $\Delta$ circuit computing the same polynomial must have size at least $s^{(\log s)^{\Omega(1/\Delta)}}.$ While this result shows that some blow-up is unavoidable in the multilinear setting, this is still only a quasipolynomial separation and does not completely resolve our original question. More recently, Kayal, Nair and Saha~\cite{KNS16} resolved the question completely in the case when $\Delta = 1$ by giving an optimal exponential separation between product-depth $2$ and product-depth $1$ multilinear circuits. 

We extend both these results to prove a strong depth-hierarchy theorem for all small depths. The following is implied by Corollary~\ref{cor:abstract-pdepth}. 

\begin{theorem}
For each $\Delta = \Delta(n) = o(\log n/\log \log n),$ there is an explicit multilinear polynomial $P^{(\Delta)}$ on $n$ variables that can be computed by a multilinear formula of product-depth $(\Delta+1)$ and linear size, but not by any multilinear circuit of product-depth $\Delta$ and size less than $\exp(n^{\Omega(1/\Delta)}).$ 
\end{theorem}
We also prove an analogous result for depth instead of product-depth (Corollary~\ref{cor:abstract-depth}). Note that, from the above discussion, the above results are tight up to the constant implicit in the $\Omega(1/\Delta)$.

\subsection{Related Work}

We survey here some of the work on depth-hierarchy theorems in the Algebraic and Boolean settings. 

As mentioned above, this is a well-known question in the Boolean setting, with a near-optimal separation between different depths due to H\r{a}stad~\cite{Hastadthesis} (building on~\cite{Ajtai,FSS,Yao}); the separating examples in H\r{a}stad's result are the \emph{Sipser functions}, which are computed by linear-sized Boolean circuits of depth $\Delta+1$ but cannot be computed by subexponential-sized depth $\Delta$ circuits. A recent variant of this lower bound was proved by Chen et al.~\cite{COST} for \emph{Skew-Sipser functions}, which in turn was used to prove near-optimal lower bounds for constant-depth Boolean circuits solving variants of the Boolean Graph Reachability problem. Their construction motivates our hard polynomials, as we describe below. 

In the algebraic setting, we have separations between fixed constant-depth circuits under the restriction of homogeneity. Nisan and Wigderson~\cite{nw1997} show that converting a homogeneous $\Sigma\Pi\Sigma\Pi$ circuit to a homogeneous $\Sigma\Pi\Sigma$ circuit requires an exponential blow-up. A quasipolynomial separation between homogeneous $\Sigma\Pi\Sigma\Pi\Sigma$ and $\Sigma\Pi\Sigma\Pi$ circuits was shown by Kumar and Saptharishi~\cite{KumarSap16}. However, as far as we know, nothing is known for larger depths. 

When the algebraic circuits are instead restricted to be multilinear, more is known. Raz and Yehudayoff~\cite{ry09} were the first to study this question, and showed a quasipolynomial depth-hierarchy theorem for all constant product-depths. In the case of product-depth $1$ vs. product-depth $2$, this was strengthened to an exponential separation by Kayal, Nair and Saha~\cite{KNS16}. 

One can ask if the methods of~\cite{ry09,KNS16} can be used to prove our result. Kayal et al.\cite{KNS16} prove their result by defining a suitable complexity measure for any polynomial $f\in \F[x_1,\ldots,x_N]$, and show that this measure is small for any subexponential-sized $\Sigma\Pi\Sigma$ circuit but large for some linear-sized $\Sigma\Pi\Sigma\Pi$ circuit. In particular, this means that this measure needs to be changed to prove lower bounds for larger depth circuits. However, it is not clear how to modify this measure to take into account the product-depth of the circuit class.

This is not a problem with the technique of Raz and Yehudayoff~\cite{ry09}, which can indeed be used to prove exponential lower bounds on the sizes of small-depth circuits for computing certain polynomials. However, the polynomials used to witness the superpolynomial separation cannot give an exponential depth-hierarchy, for the reasons we now explain. 

The depth-hierarchy theorem of~\cite{ry09} is obtained via an exponential lower bound $s(n,\Delta) \approx \exp(n^{1/\Delta})$ against product-depth $\Delta$ circuits computing polynomials from a certain ``hard'' class $\mc{P}$ of polynomials on $n$ variables. Importantly, these lower bounds are tight in the sense that there also exist product-depth $\Delta$ circuits of size roughly $s(n,\Delta)$ computing polynomials from $\mc{P}$. Since $s(n,\Delta-1)$ is superpolynomially larger than $s(n,\Delta),$ we obtain a superpolynomial separation between circuits of product-depth $\Delta$ and product-depth $\Delta-1$. However, since we cannot improve on either the upper bound of $s(n,\Delta)$ or the lower bound of $s(n,\Delta-1)$ for this class of polynomials, we cannot hope to improve this separation. 

There is a striking parallel between this line of work and the setting of Boolean circuits, where also a similar quasipolynomial depth-hierarchy theorem can be obtained by appealing to easier lower bounds for explicit functions such as the ``Parity'' function~\cite{Hastadthesis}. It is known that the depth-$\Delta$ Boolean complexity of the Parity function is $\exp(\Theta(n^{1/\Delta-1}))$ and this yields a quasipolynomial depth-hierarchy theorem for constant-depth Boolean circuits as in the work of Raz and Yehudayoff described above. However, H\r{a}stad~\cite{Hastadthesis} was able to improve this to an exponential depth-hierarchy theorem by changing the candidate hard function to the Sipser functions  and then proving a lower bound for these functions by a related, but more involved, technique~\cite{Hastadthesis, COST}.

\subsection{Proof Outline}

As mentioned in the Related Work section above, the polynomials considered by Raz and Yehudyaoff~\cite{ry09} cannot be used to prove better than a quasipolynomial depth-hierarchy theorem. This parallels a quasipolynomial depth-hierarchy theorem in the Boolean setting. However, in the Boolean setting, there is a different family of explicit functions that can be used to prove an exponential depth-hierarchy theorem.

Our aim is to do something similar in the setting of multilinear small-depth circuits. While the methods for proving Boolean circuit lower bounds do not seem to apply in the algebraic setting, we can take inspiration in the matter of choosing the candidate hard polynomial. For this, we look to the recent result of Chen et al.~\cite{COST}, who observe that the Sipser functions (and also their ``skew'' variants) can be interpreted as special cases of the Boolean Graph Reachability problem. This is quite appealing for us, since Graph Reachability has a natural polynomial analogue, the Iterated Matrix Multiplication polynomial, which has been a source of many lower bounds in algebraic circuit complexity~\cite{nw1997,FLMS,KShom,BeraChakrabarti15,KNS16,KST16,CLS}. We therefore choose our lower bound candidate to be a restriction of the Iterated Matrix Multiplication polynomial, which we now describe. 

\subsubsection{The Hard Polynomials}

All the polynomials we consider will be naturally defined in terms of Directed Acyclic graphs (DAGs) with a unique source and sink. Given such a graph $G$ with source $s$ and sink $t$, we define a corresponding polynomial $P_G$ as follows. Label each edge $e$ of $G$ with a distinct variable $x_e$. The polynomial $P_G$ is defined to be the sum, over all paths $\pi$ from $s$ to $t$, of the monomial which is the product of edge labels along that path. (See Figure~\ref{fig:eg1intro} for a simple example.) This polynomial $P_G$ is the algebraic analogue of the Boolean Computational problem of checking $s$-$t$ reachability  on subgraphs of $G$. (Informally, we think of each $x_e$ as a Boolean variable that determines if $e$ remains in the subgraph or not. Then the polynomial $P_G$ on the Boolean input corresponding to a subgraph $H$ of $G$ counts the number of $s$-$t$ paths in $H$.)

 \begin{figure}
    \centering
    \begin{tikzpicture}[scale=0.75]
      \tikzstyle{vertex}=[circle, draw, fill=blue!20, inner sep=0pt, minimum width=5pt]
      
      \node[vertex] (a) at (0,0) {};
      \node[vertex] (b) at (1,1) {};
      \node[vertex] (c) at (2,0) {};
      \node[vertex] (d) at (3,1) {};
      \node[vertex] (e) at (5,0) {};

      \draw[->] (a) -- (b) node[midway,left] {$x_{1}$};
      \draw[->] (b) -- (d) node[midway,above] {$x_{2}$};
      \draw[->] (d) -- (e) node[midway,above] {$x_{3}$};
      \draw[->] (a) -- (c) node[midway,below] {$x_{4}$};
      \draw[->] (c) -- (e) node[midway,below] {$x_{6}$};
      \draw[->] (b) -- (c) node[midway,right] {$x_{5}$};
    \end{tikzpicture}
    \caption{A simple example of a graph $G$ (all edges go from left to right). In this case, $P_G$ is $x_1x_2x_3+x_1x_5x_6+x_4x_6.$}
    \label{fig:eg1intro}
  \end{figure}
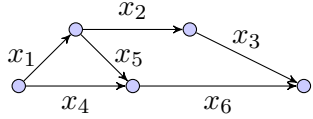

When $G$ is a layered graph with $d$ layers and all possible edges between consecutive layers, $P_G$ is known as the Iterated Matrix Multiplication polynomial (for the connection to matrix product, see, e.g.,~\cite{FLMS}). This polynomial has been studied in many contexts in Algebraic circuit complexity. Unfortunately, it is not useful in our setting, since it does not have a subexponential-size constant-depth multilinear circuit, as was shown recently by some of the authors~\cite{CLS}. In particular, this means that it cannot be used to obtain the claimed separation between depths $\Delta+1$ and $\Delta.$

However, changing $G$ can drastically reduce the complexity of the polynomial $P_G$. To obtain a $G$ such that $P_G$ has an efficient depth $\Delta+1$ circuit, we use \emph{series-parallel} graphs, as in the result of Chen et al.~\cite{COST}. 

Recall that given DAGs $G_1,\ldots,G_k$ with sources $s_1,\ldots,s_k$ and sinks $t_1,\ldots,t_k$, we can construct larger DAGs by composing these graphs in parallel (by identifying all the sources and all the sinks) or in series (by identifying $t_1$ with $s_2$, $t_2$ with $s_3$ and so on until $t_{k-1}$ with $s_k$) to get larger DAGs $G_{par}$ or $G_{ser}$ respectively. Note that the corresponding polynomials $P_{par}$ and $P_{ser}$ are the sums and products of the polynomials $P_1,\ldots,P_k,$ respectively. In particular, if $P_1,\ldots,P_k$ have efficient circuits of product-depth at most $\Delta$, then $P_{par}$ (resp.\ $P_{ser}$) has an efficient circuit of product-depth at most $\Delta$ (resp.\ $\Delta + 1$).

In this way, we can inductively construct polynomials which, by their very definition, have efficient circuits of small depth. In particular, to keep the product-depth of the circuit bounded by $\Delta+1,$ it is sufficient to ensure that the number of  series compositions used in the construction of the graph is at most $\Delta+1.$

\subsubsection{The $\Sigma\Pi\Sigma$ lower bound}
We motivate our proof with the solution to the simpler problem of separating product-depth $2$ and product-depth $1$ circuits. In fact, we will separate the power of $\Pi\Sigma\Pi$ and $\Sigma\Pi\Sigma$ circuits. While an exponential separation was already known in this case thanks to the work of Kayal et al.~\cite{KNS16}, we will outline a different proof that extends to larger depths.

Given the above discussion, a natural polynomial to witness the separation between product-depth $1$ and product-depth $2$ circuits is a polynomial corresponding to a series parallel graph obtained with exactly $2$ series compositions. Unfortunately, our proof technique is not able to prove a lower bound for such a graph. However, we are able to prove such a separation with the slightly more complicated graph $G^{(1)}$ in Figure~\ref{fig:g1intro}, which is made of a composition of $m$ copies of the basic graph $H^{(1)}$ (see Figure~\ref{fig:h1intro}). 
Though the graph $G^{(1)}$ is constructed using three series compositions rather than two, the corresponding polynomial $P^{(1)}(X)$ nevertheless has a product-depth $2$ circuit of small size (this is because the polynomial corresponding to $H^{(1)}$ depends only on a constant-number of variables and hence has a ``brute-force'' $\Sigma\Pi$ circuit).

\begin{figure}
\centering
  \begin{tikzpicture}[scale=0.75]
    \tikzstyle{vertex}=[circle, draw, fill=blue!20, inner sep=0pt, minimum width=5pt]
      \node[vertex] (a) at (0,0) {};

      \node[vertex] (b) at (1.5,0.5) {};
      \node[vertex] (c) at (0.5,1.5) {};
      \node[vertex] (d) at (2,2) {};
      \draw[->] (a) edge (b) (b) edge (d);
      \draw[->] (a) edge (c) (c) edge (d);

      \begin{scope}[rotate=-90,shift={(0,0)}]
        \node[vertex] (bl) at (1.5,0.5) {};
        \node[vertex] (cl) at (0.5,1.5) {};
        \node[vertex] (dl) at (2,2) {};
        \draw[->] (a) edge (bl) (bl) edge (dl);
        \draw[->] (a) edge (cl) (cl) edge (dl);
      \end{scope}

      \begin{scope}[rotate=-90,shift={(-2,2)}]
        \node[vertex] (bur) at (1.5,0.5) {};
        \node[vertex] (cur) at (0.5,1.5) {};
        \node[vertex] (dur) at (2,2) {};
        \draw[->] (d) edge (bur) (bur) edge (dur);
        \draw[->] (d) edge (cur) (cur) edge (dur);
      \end{scope}
      
      \begin{scope}[rotate=0,shift={(2,-2)}]
        \node[vertex] (blr) at (1.5,0.5) {};
        \node[vertex] (clr) at (0.5,1.5) {};
        \draw[->] (dl) edge (blr) (blr) edge (dur);
        \draw[->] (dl) edge (clr) (clr) edge (dur);
      \end{scope}
    \end{tikzpicture}
    \caption{$H^{(1)}$.}
    \label{fig:h1intro}
  \end{figure}

\begin{figure}
\centering
    \begin{tikzpicture}[scale=0.75,level/.style={},decoration={brace,mirror,amplitude=7}]
      \tikzstyle{vertex}=[circle, draw, fill=blue!20, inner sep=0pt, minimum width=5pt]

      \node[vertex] (a) at (0,0) {};

      \node[vertex] (b) at (1.5,0.5) {};
      \node[vertex] (c) at (0.5,1.5) {};
      \node[vertex] (d) at (2,2) {};
      \draw[->] (a) edge (b) (b) edge (d);
      \draw[->] (a) edge (c) (c) edge (d);

      \begin{scope}[rotate=-90,shift={(0,0)}]
        \node[vertex] (bl) at (1.5,0.5) {};
        \node[vertex] (cl) at (0.5,1.5) {};
        \node[vertex] (dl) at (2,2) {};
        \draw[->] (a) edge (bl) (bl) edge (dl);
        \draw[->] (a) edge (cl) (cl) edge (dl);
      \end{scope}

      \begin{scope}[rotate=-90,shift={(-2,2)}]
        \node[vertex] (bur) at (1.5,0.5) {};
        \node[vertex] (cur) at (0.5,1.5) {};
        \node[vertex] (dur) at (2,2) {};
        \draw[->] (d) edge (bur) (bur) edge (dur);
        \draw[->] (d) edge (cur) (cur) edge (dur);
      \end{scope}
      
      \begin{scope}[rotate=0,shift={(2,-2)}]
        \node[vertex] (blr) at (1.5,0.5) {};
        \node[vertex] (clr) at (0.5,1.5) {};
        \draw[->] (dl) edge (blr) (blr) edge (dur);
        \draw[->] (dl) edge (clr) (clr) edge (dur);
      \end{scope}

      \begin{scope}[shift={(4,0)}]
        \node[vertex] (a) at (0,0) {};
        
        \node[vertex] (b) at (1.5,0.5) {};
        \node[vertex] (c) at (0.5,1.5) {};
        \node[vertex] (d) at (2,2) {};
        \draw[->] (a) edge (b) (b) edge (d);
        \draw[->] (a) edge (c) (c) edge (d);

        \begin{scope}[rotate=-90,shift={(0,0)}]
          \node[vertex] (bl) at (1.5,0.5) {};
          \node[vertex] (cl) at (0.5,1.5) {};
          \node[vertex] (dl) at (2,2) {};
          \draw[->] (a) edge (bl) (bl) edge (dl);
          \draw[->] (a) edge (cl) (cl) edge (dl);
        \end{scope}

        \begin{scope}[rotate=-90,shift={(-2,2)}]
          \node[vertex] (bur) at (1.5,0.5) {};
          \node[vertex] (cur) at (0.5,1.5) {};
          \node[vertex] (dur) at (2,2) {};
          \draw[->] (d) edge (bur) (bur) edge (dur);
          \draw[->] (d) edge (cur) (cur) edge (dur);
        \end{scope}
        
        \begin{scope}[rotate=0,shift={(2,-2)}]
          \node[vertex] (blr) at (1.5,0.5) {};
          \node[vertex] (clr) at (0.5,1.5) {};
          \draw[->] (dl) edge (blr) (blr) edge (dur);
          \draw[->] (dl) edge (clr) (clr) edge (dur);
        \end{scope}
      \end{scope}

      \node at (9.5,0) {\Large $\dots$};
      \begin{scope}[shift={(11,0)}]
        \node[vertex] (a) at (0,0) {};

        \node[vertex] (b) at (1.5,0.5) {};
        \node[vertex] (c) at (0.5,1.5) {};
        \node[vertex] (d) at (2,2) {};
        \draw[->] (a) edge (b) (b) edge (d);
        \draw[->] (a) edge (c) (c) edge (d);

        \begin{scope}[rotate=-90,shift={(0,0)}]
          \node[vertex] (bl) at (1.5,0.5) {};
          \node[vertex] (cl) at (0.5,1.5) {};
          \node[vertex] (dl) at (2,2) {};
          \draw[->] (a) edge (bl) (bl) edge (dl);
          \draw[->] (a) edge (cl) (cl) edge (dl);
        \end{scope}

        \begin{scope}[rotate=-90,shift={(-2,2)}]
          \node[vertex] (bur) at (1.5,0.5) {};
          \node[vertex] (cur) at (0.5,1.5) {};
          \node[vertex] (dur) at (2,2) {};
          \draw[->] (d) edge (bur) (bur) edge (dur);
          \draw[->] (d) edge (cur) (cur) edge (dur);
        \end{scope}
        
        \begin{scope}[rotate=0,shift={(2,-2)}]
          \node[vertex] (blr) at (1.5,0.5) {};
          \node[vertex] (clr) at (0.5,1.5) {};
          \draw[->] (dl) edge (blr) (blr) edge (dur);
          \draw[->] (dl) edge (clr) (clr) edge (dur);
        \end{scope}
      \end{scope}
      \node (start) at (1,0) {};
      \node (end) at (15,0) {};
      \draw [decorate] ([yshift=-25mm]start.west) --node[below=3mm]{$m$ copies} ([yshift=-25mm]end.east);
    \end{tikzpicture}
    \caption{$G^{(1)}$}

\label{fig:g1intro}
  \end{figure}

We now describe how to prove that any $\Sigma\Pi\Sigma$ circuit for $P^{(1)}$ must have large size (the lower bound we obtain is $2^{\Omega(m)}$). Let $F$ be a $\Sigma\Pi\Sigma$ circuit of size $s$ computing $P^{(1)}$ and assume $s$ is small. We can write $F$ as $T_1+\cdots + T_s$, where each $T_i$ is a product of linear polynomials. 

The lower bound is via a rank argument first used by Raz~\cite{Raz}. The idea is to associate a rank-based complexity measure with any polynomial $f$, and show that while the polynomial $P^{(1)}$ has rank as large as possible, the rank of $F$ must be small. 

This rank measure is defined as follows. We partition the variables $X$ in our polynomial into two sets $Y$ and $Z$ and consider any polynomial $f(X)$ as a polynomial in the variables in $Y$ with coefficients from $\F[Z].$ The rank of the space of coefficients (as vectors over the base field $\F$) is considered a measure of the complexity of $f$. 

It is easy to come up with a partition of the underlying variable set $X$ into $Y,Z$ so that the complexity of our polynomial $P^{(1)}$ is as large as it can be. Unfortunately, it is also easy to find $\Sigma\Pi\Sigma$ formulas that have maximum rank w.r.t. this partition. Hence, this notion of complexity is not by itself sufficient to prove a lower bound. At this point, we follow an idea of Raz~\cite{Raz} and show something stronger for $P^{(1)}$: we show that its complexity is \emph{robust} in the sense that it is full rank w.r.t. many different partitions.

More precisely, we carefully design a large space of \emph{restrictions} $\rho^{(1)}:X\rightarrow Y\cup Z\cup \{0,1\}$ such that for any such restriction $\rho^{(1)},$ the resulting substitution of $P^{(1)}$, which we will call a \emph{restriction of $P^{(1)}$}, continues to have full rank, but the rank of the restriction of $F$ is small under many of them. We define these restrictions now. 

The definition of the space of restrictions is motivated by the following easily verified observation (also used in many previous results~\cite{Raznc2nc1,ry08,ry09,DMPY12}): any polynomial of the form 
\begin{equation}
\label{eq:fullrkpoly}
\prod_{i\in [t]}(y_{i} + z_{i})
\end{equation}
is full-rank w.r.t. the natural partition $Y = \{y_1,\ldots,y_t\}$ and $Z=\{z_1,\ldots,z_t\}$. Given this, a possible option is to have the restriction $\rho^{(1)}$ set variables in each copy of $H^{(1)}$ in ways that ensure that after substitution, the polynomial is of the form in (\ref{eq:fullrkpoly}). We choose one among the three different (up to isomorphism) restrictions shown in Figure~\ref{fig:g0set}. It can be checked that each such restriction results in a polynomial of the form in (\ref{eq:fullrkpoly}) and hence is full rank. Since $P^{(1)}$ is a product of many disjoint copies of this basic polynomial, it remains full rank also.  

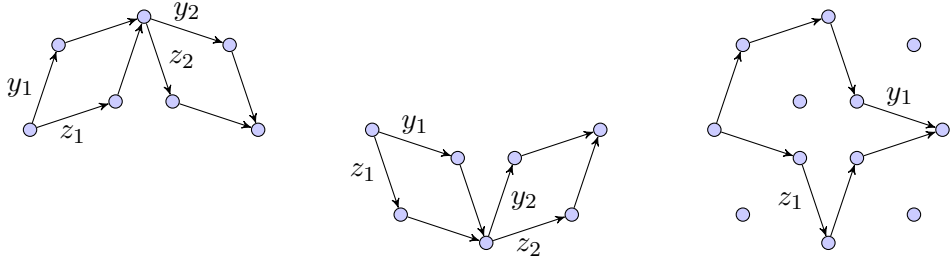
\begin{figure}
\centering
\begin{tikzpicture}[scale=0.75]
  \tikzstyle{vertex}=[circle, draw, fill=blue!20, inner sep=0pt, minimum width=5pt]
  
  \node[vertex] (a) at (0,0) {};
  
  \node[vertex] (b) at (1.5,0.5) {};
  \node[vertex] (c) at (0.5,1.5) {};
  \node[vertex] (d) at (2,2) {};
  \draw[->] (a) -- (c) node[midway,left] {$y_{1}$};
  \draw[->] (c) -- (d);
  \draw[->] (a) -- (b) node[midway,below] {$z_1$};
  \draw[->] (b) edge (d);

  \begin{scope}[rotate=-90,shift={(0,0)}]
    \node at (0,0) {};
    \node at (2,2) {};
  \end{scope}

  \begin{scope}[rotate=-90,shift={(-2,2)}]
    \node[vertex] (bur) at (1.5,0.5) {};
    \node[vertex] (cur) at (0.5,1.5) {};
    \node[vertex] (dur) at (2,2) {};
    \draw[->] (d) -- (cur) node[midway,above] {$y_{2}$};
    \draw[->] (cur) -- (dur);
    \draw[->] (d) -- (bur) node[midway,right] {$z_2$}; 
    \draw[->] (bur) edge (dur);
  \end{scope}
\begin{scope}[shift={(6,0)}]
  \node[vertex] (a) at (0,0) {};

  \begin{scope}[rotate=-90,shift={(0,0)}]
    \node[vertex] (bl) at (1.5,0.5) {};
    \node[vertex] (cl) at (0.5,1.5) {};
    \node[vertex] (dl) at (2,2) {};
    \draw[->] (a) -- (cl) node[midway,above] {$y_{1}$};
    \draw[->] (cl) -- (dl);
    \draw[->] (a) -- (bl) node[midway,left] {$z_1$};
    \draw[->] (bl) edge (dl);
  \end{scope}
  
  \begin{scope}[rotate=0,shift={(2,-2)}]
    \node[vertex] (blr) at (1.5,0.5) {};
    \node[vertex] (clr) at (0.5,1.5) {};
    \node[vertex] (dur) at (2,2) {};
    \draw[->] (dl) -- (clr) node[midway,right] {$y_{2}$};
    \draw[->] (clr) -- (dur);
    \draw[->] (dl) -- (blr) node[midway,below] {$z_2$};
    \draw[->] (blr) edge (dur);
  \end{scope}
\end{scope}
\begin{scope}[shift={(12,0)}]
  \node[vertex] (a) at (0,0) {};

  \node[vertex] (b) at (1.5,0.5) {};
  \node[vertex] (c) at (0.5,1.5) {};
  \node[vertex] (d) at (2,2) {};
  \draw[->] (a) -- (c);
  \draw[->] (c) -- (d);
  
  \begin{scope}[rotate=-90,shift={(0,0)}]
    \node[vertex] (bl) at (1.5,0.5) {};
    \node[vertex] (cl) at (0.5,1.5) {};
    \node[vertex] (dl) at (2,2) {};
  \end{scope}

  \begin{scope}[rotate=-90,shift={(-2,2)}]
    \node[vertex] (bur) at (1.5,0.5) {};
    \node[vertex] (cur) at (0.5,1.5) {};
    \node[vertex] (dur) at (2,2) {};
  \end{scope}
  
  \begin{scope}[rotate=0,shift={(2,-2)}]
    \node[vertex] (blr) at (1.5,0.5) {};
    \node[vertex] (clr) at (0.5,1.5) {};
  \end{scope}
  
  \draw[->] (d) -- (bur);
  \draw[->] (bur) -- (dur) node[midway,above] {$y_{1}$};

  \draw[->] (a) -- (cl);
  \draw[->] (cl) -- (dl) node[midway,left] {$z_{1}$};
  \draw[->] (dl) -- (clr);
  \draw[->] (clr) -- (dur);
\end{scope}
\end{tikzpicture}
\caption{Restrictions applied to each copy of $H^{(1)}$. Edges that are not labelled have their variables set to $1$. Edges that are not present have their variables set to $0$.}
\label{fig:g0set}
\end{figure}

Now, to show that $F$ has small rank under some such restriction, we apply a \emph{uniformly random} copy of such a restriction $\rho^{(1)}$ to $F$ and show that with high probability each of its constituent terms $T_1,\ldots, T_s$ is very small in rank. Since rank is subadditive and $s$ is assumed to be small, this implies that $F$ cannot be full rank after the restriction and hence cannot have been computing the polynomial $P^{(1)}$. 

Fix a term $T_i$ which is a product of linear functions as mentioned above. We argue that it restricts to a small-rank term with high probability as follows.

\begin{itemize}
\item  A standard argument~\cite{Raz} shows that, by multilinearity,\footnote{This essentially means that the linear functions that participate in $T_i$ do not share any variables.} the rank of $T_i$ is small if many of its constituent linear functions have smaller than ``full-rank'' after applying $\rho^{(1)}$. Here,  a linear polynomial $L(Y',Z')$ ($Y'\subseteq Y$, $Z'\subseteq Z$) is said to be full-rank if its complexity (as defined above) is $2^{(|Y'|+|Z'|)/2}.$ 
\item By a simple argument we can argue that, upon applying a random restriction $\rho^{(1)}$, any constituent linear function $L$ of $T_i$ becomes rank-deficient (i.e. short of full-rank) with good (constant) probability. This is by a case analysis on the set of variables $X'$ that appear in $L$ and proceeds roughly as follows. 
\begin{itemize}
\item Say $X'$ is \emph{structured} in that it is the union of the variable sets in some copies of $H^{(1)}.$ In this case, we observe that with good probability, the total number of variables in $L$ after restriction is larger than $2$ and hence, the target full-rank is a number larger than $2$. On the other hand, any linear function cannot have rank more than $2$. 
\item On the other hand if $X'$ is not structured in the above sense, then it can be shown that with good probability, the number of restricted variables in $L$ is odd and then it cannot be full-rank, since the target full-rank is $2^{t+1/2}$ for some integer $t$, while the rank of $L$ can be at most $2^t$. 
\end{itemize}
\item Given the above, we can easily argue using a concentration bound\footnote{By multilinearity, the events that the various linear functions are small rank become (roughly) independent.} that with high probability, $T_i$ has \emph{many} rank-deficient linear functions and consequently has small rank. This yields the proof of the $\Sigma\Pi\Sigma$ lower bound.
\end{itemize} 

\subsubsection{Separating $\Delta$ from $\Delta+1$}

To prove the lower bound for product-depth $\Delta$ circuits, we proceed by induction on $\Delta.$ Each step in the inductive proof is analogous to a step in the $\Sigma\Pi\Sigma$ lower bound so we will be brief.

The graph $H^{(\Delta)}$ is constructed from two copies of $G^{(\Delta-1)}$ by parallel composition and the graph $G^{(\Delta)}$ is constructed from many copies of $H^{(\Delta)}$ by series composition as shown in Figure~\ref{fig:gdelta}. By construction $P^{(\Delta)}$ will have a circuit of product-depth $\Delta+1.$

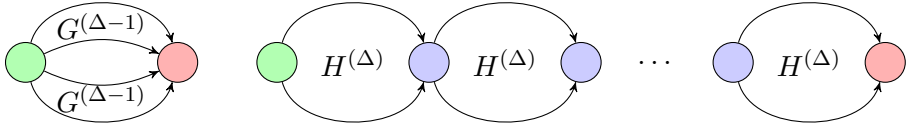
\begin{figure}[H]
\centering
    \begin{tikzpicture}[scale=0.4]
      \node[initialvertex] (a) at (0,0) {};
      \node[finalvertex]   (b) at (5,0) {};
      \draw[->] (a) to [bend left=75] (b);
      \node at (2.5,1.2) {$G^{(\Delta-1)}$};
      \draw[->] (a) to [bend left=25] (b);
      \draw[->] (a) to [bend left=-25] (b);
      \node at (2.5,-1.2) {$G^{(\Delta-1)}$};
      \draw[->] (a) to [bend left=-75] (b);
    \end{tikzpicture}
	\hspace*{0.2in}
    \begin{tikzpicture}
      \node[initialvertex] (a) at (0,0) {};
      \node[vertex] (b) at (2,0) {};
      \draw[->] (a) to [bend left=75] (b);
      \draw[->] (a) to [bend left=-75] (b);
      \node at (1,0) {$H^{(\Delta)}$};
      
      \node[vertex] (c) at (4,0) {};
      \draw[->] (b) to [bend left=75] (c);
      \draw[->] (b) to [bend left=-75] (c);
      \node at (3,0) {$H^{(\Delta)}$};
      
      \node at (5,0) {$\dots$};
      
      \node[vertex] (d) at (6,0) {};
      \node[finalvertex] (e) at (8,0) {};
      \draw[->] (d) to [bend left=75] (e);
      \draw[->] (d) to [bend left=-75] (e);
      \node at (7,0) {$H^{({\Delta})}$};
    \end{tikzpicture}
\caption{$H^{(\Delta)}$ (left) and $G^{(\Delta)}$}
\label{fig:gdelta}
\end{figure}

The lower bound is again proved via a random restriction argument. We define the random restriction $\rho^{(\Delta)}$ by restricting each copy of $H^{(\Delta)}$ independently by choosing one of these options at random.
\begin{itemize}
\item set it to a polynomial of the form $(y+z)$ by choosing a random path in each copy of $G^{(\Delta-1)}$ and setting a single variable in each to $y$ or $z$ (other variables in the path are set to $1$ and off-path variables are set to $0$),
\item inductively applying the restriction $\rho^{(\Delta-1)}$ to one of the copies of $G^{(\Delta-1)}$ (setting the other one to $0$).
\end{itemize}
As before, the polynomial $P^{(\Delta)}$ remains full-rank after the restriction with probability $1$, since it always transforms into a polynomial of the form (\ref{eq:fullrkpoly}) after the restrictions. 

Let $F$ be a small product-depth $\Delta$ circuit. To argue that it is low-rank w.h.p. after applying $\rho^{(\Delta)}$, we prove a variant of a decomposition lemma from the literature that allows us to write $F = T_1 + \ldots + T_s$ where each $T_i$ now is a product of circuits each of which has  small size and product depth at most $\Delta-1.$ 

As before, we argue that each $T_i$ is low-rank with high probability. Say $T_i = \prod_{j\in [t]}Q_{i,j}$. To show this, it suffices to show that each $Q_{i,j}$ is somewhat short of full-rank with good probability. This does indeed turn out to be true in one of two ways.
\begin{itemize}
\item If the variable set $X'$ of $Q_{i,j}$ is structured (as defined above), then we can use induction and show that it is low-rank if we happen to apply the restriction $\rho^{(\Delta-1)}$ to the variables of $Q_{i,j}.$ This happens with constant probability.
\item On the other hand, if $Q_{i,j}$ is unstructured, then with good probability, $Q_{i,j}$ is left with an odd number of variables, which implies that it is rank-deficient exactly as before. 
\end{itemize}

At an intuitive level, the above cases correspond to two different ways in which a small circuit can try to compute the hard polynomial $P^{(\Delta)}.$ In the first case, it uses the fact that the polynomial $P^{(\Delta)}$ is a product of many polynomials and tries to compute each of them with a smaller circuit of depth $\Delta-1$; in this case, we use the inductive hypothesis to show that this cannot happen. In the second case, the circuit tries to do some non-trivial (i.e. unexpected) computation at the top level and in this case, we argue directly that the circuit fails. 

This finishes the sketch of an idealized version of the argument. While the above strategy can be carried out as stated, it would yield a sub-optimal but still exponential bound of the form $\exp(n^{\alpha_\Delta})$ for some $\alpha_\Delta > 0$ that depends on $\Delta.$ To obtain the near-optimal lower bound of $\exp(n^{\Omega(1/\Delta)})$, we need to make the following changes to the above idealized proof strategy. 

\begin{itemize}
\item We expand the set of ``structured'' polynomials to allow for polynomials $Q_{i,j}$ (and inductively formulas computed at smaller depths) that compute polynomials over a non-trivial fraction, say $\varepsilon$, of the variables in many different copies, say $M$, of $H^{(\Delta-1)}$. The exact relationship between $\varepsilon$ and $M$ is somewhat delicate (we choose $\varepsilon \geq 1/M^{0.25}$) and must be chosen carefully for the proof to yield an optimal lower bound.
\item Given this, we also need to handle (for the general inductive statement) the case that $F$ does not depend on all the variables in $G^{(\Delta)}$ but rather at least $\delta$-fraction of the variables in $K$ many copies of $H^{(\Delta)}$ for some suitable $\delta$ and $K$. Given a single term $T_i = \prod_{j}Q_{i,j}$ in $F$ as above, one of the following two cases occurs.
\begin{itemize}
\item One of the $Q_{i,j}$s depends on at least a $\delta'$-fraction of the variables in $K'$ many copies of $H^{(\Delta)}$, where $\delta'$ and $K'$ are not much smaller than $\delta$ and $K$ respectively. This corresponds to the structured case above and in this case, we can actually use the induction hypothesis to show that we are done with high probability. 

To argue this, we have to use the fact that while $\delta$ has decreased to $\delta'$, the $K'$ many copies of $H^{(\Delta)}$ contain many more copies of $H^{(\Delta-1)}$. Hence the number of copies of $H^{(\Delta-1)}$, say $K''$, and $\delta'$ still have the ``correct'' relationship so that the inductive statement is applicable. 

\item Otherwise, in many copies of $H^{(\Delta)},$ the $\delta$-fraction that are variables of $F$ are further partitioned into parts of relative size $< \delta'$ by the variable sets of the $Q_{i,j}$. In this case, we have to argue that many of the $Q_{i,j}$ are rank-deficient with high probability. This is the most technical part of the proof and is done by a careful re-imagining of the sampling process that defines the random restriction. 
\end{itemize}
\end{itemize}

The more general inductive statement (and some additional technicalities that force us to carry around auxiliary sets of $Y$ and $Z$ variables) results in a technically complicated theorem statement (see Theorem~\ref{thm:technical}), which yields the near-optimal depth-hierarchy theorem. 

\paragraph{Organization.} We start with some Preliminaries in Section~\ref{sec:prelims}. We define the hard polynomials that we use in Section~\ref{sec:polysandrests}. The general inductive statement (Theorem~\ref{thm:technical}) is given in Section~\ref{sec:mainthm}, which also contains the proof of the main depth-hierarchy theorem (Corollary~\ref{cor:abstract-pdepth}) assuming Theorem~\ref{thm:technical}. Finally, we prove Theorem~\ref{thm:technical} in Section~\ref{sec:technical}.

\section{Preliminaries}
\label{sec:prelims}

\subsection{Polynomials and restrictions}
\label{sec:polyrest}

Throughout, let $\F$ be an arbitrary field. A polynomial $P\in \F[X]$ is called \emph{multilinear} if the degree of $P$ in each variable $x\in X$ is at most $1$. 

Let $X,Y,Z$ be disjoint sets of variables. An \emph{$(X,Y,Z)$-restriction} is a function $\rho:X\rightarrow Y\cup Z\cup \{0,1\}.$ The sets $X,Y$ and $Z$ will sometimes be omitted when clear from context. We say that the restriction $\rho$ is \emph{multilinear} if no two variables in $X$ are mapped to the same variable in $Y\cup Z$ by $\rho$. A \emph{random} $(X,Y,Z)$-restriction is simply a random function $\rho:X\rightarrow Y\cup Z\cup \{0,1\}$ (chosen according to some distribution).

Let $X,Y,Z$ be as above and say we have for each $i\in [k]$, an $(X_i,Y_i,Z_i)$-restriction $\rho_i$ where $X_i\subseteq X, Y_i\subseteq Y$ and $Z_i\subseteq Z.$ If $X_1,\ldots,X_k$ form a (pairwise disjoint) partition of $X$, we define their \emph{composition} --- denoted $\rho_1\circ\rho_2\cdots\circ \rho_k$ --- to be the $(X,Y,Z)$-restriction $\rho$ such that $\rho(x)$ agrees with $\rho_i(x)$ for each $x\in X_i.$ Further, if restrictions $\rho_i$ are multilinear restrictions and the sets $Y_i$ ($i\in [k]$) and $Z_j$ ($j\in [k]$) are pairwise disjoint, then $\rho_1\circ\cdots\circ \rho_k$ is also multilinear. 

Let $\rho$ be an $(X,Y,Z)$-restriction. Given a polynomial $f\in \F[X]$, the restriction $\rho$ yields a natural polynomial in $\F[Y\cup Z]$ obtained from $f$ by substitution; we denote this polynomial $f|_\rho$. Note, moreover, that if $f$ is multilinear and $\rho$ is multilinear, then so is $f|_\rho$.

\subsection{Partial derivative matrices and relative rank}
\label{sec:pdrelrk}

Let $Y$ and $Z$ be two disjoint sets of variables and let $g\in \F[Y\cup Z]$ be a multilinear polynomial. Define the $2^{|Y|}\times 2^{|Z|}$ matrix $\M_{(Y,Z)}(g)$ whose rows and columns are labelled by distinct multilinear monomials in $Y$ and $Z$ respectively and the $(m_1,m_2)$th entry of $\M_{(Y,Z)}(g)$ is the coefficient of the monomial $m_1\cdot m_2$ in $g$. We will use the rank of this matrix as a measure of the complexity of $g$. 

We define the \emph{relative-rank} of $g$ w.r.t. $(Y,Z)$ (denoted by $\relrk_{(Y,Z)}(g)$) by
\[
\relrk_{(Y,Z)}(g) = \frac{\rank(M_{(Y,Z)}(g))}{2^{(|Y|+|Z|)/2}}.
\]

We note the following properties of relative rank~\cite{ry09}.

\begin{proposition}
\label{prop:relrk}
Let $g,g_1,g_2\in \F[Y\cup Z]$ be multilinear polynomials.
\begin{enumerate}
\item $\relrk_{(Y,Z)}(g) \leq 1.$ Further if $|Y|+|Z|$ is odd, $\relrk_{(Y,Z)}(g) \leq 1/\sqrt{2}.$
\item $\relrk_{(Y,Z)}(g_1+g_2)\leq \relrk_{(Y,Z)}(g_1) + \relrk_{(Y,Z)}(g_2).$
\item \label{relrk:mult} If $Y$ is partitioned into $Y_1,Y_2$ and $Z$ into $Z_1,Z_2$ with $g_i\in \F[Y_i\cup Z_i]$ ($i\in [2]$), then $\rank(M_{(Y,Z)}(g)) = \rank(M_{(Y_1,Z_1)})(g_1)\cdot \rank(M_{(Y_2,Z_2)})(g_2).$ In particular, $\relrk_{(Y,Z)}(g_1\cdot g_2) = \relrk_{(Y_1,Z_1)}(g_1)\cdot \relrk_{(Y_2,Z_2)}(g_2).$
\end{enumerate}
\end{proposition}

\subsection{Multilinear models of computation}
\label{sec:mlmodels}

We refer the reader to the standard resources (e.g.~\cite{sy,github}) for basic definitions related to algebraic circuits and formulas. Having said that, we make a few remarks.

\begin{itemize}
\item All the gates in our formulas and circuits will be allowed to have \emph{unbounded} fan-in.
\item  The size of a formula or circuit will refer to the number of gates (including input gates) in it, and the depth of the formula or circuit will refer to the maximum number of gates on a path from an input gate to the output gate.
\item  Further, the \emph{product-depth} of the formula or circuit (as in \cite{ry08}) will refer to the maximum number of product gates on a path from the input gate to output gate. Note that if a formula or circuit has depth $\Delta$, we can assume without loss of generality that its product depth is between $\lfloor \Delta/2\rfloor$ and $\lceil \Delta/2\rceil$ (by collapsing sum and product gates if necessary).
\end{itemize}

 An algebraic formula $F$ (resp.\ circuit $C$) computing a polynomial from $\F[X]$ is said to be \emph{multilinear} if each gate in the formula (resp.\ circuit) computes a multilinear polynomial. Moreover, a formula $F$ is said to be \emph{syntactic multilinear} if for each $\times$-gate $\Phi$ of $F$ with children $\Psi_1,\ldots,\Psi_t$, we have 
$\supp(\Psi_i)\cap \supp(\Psi_j) = \emptyset \text{ for each $i\neq j$},$ 
where $\supp(\Phi)$ denotes the set of variables that appear in the subformula rooted at $\Phi.$ Finally, for $\Delta \geq 1$, we say that a multilinear formula (resp.\ circuit) is a $(\Sigma\Pi)^{\Delta}\Sigma$ formula (resp.\ circuit) if the output gate is a sum gate and along any path from an input gate to the output gate, the sum and product gates alternate, with product gates appearing exactly $\Delta$ times and the bottom gate being a sum gate. We can define $(\Sigma\Pi)^{\Delta}, \Sigma\Pi\Sigma, \Sigma\Pi\Sigma\Pi$ formulas and circuits similarly.

An \emph{Algebraic Branching Program} ($\ABP$), over the set of variables $X$ and field $\mathbb{F}$ is a layered (i.e., the edges are only between two consecutive layers) directed acyclic graph $G$ with two special vertices called the \emph{source} and the \emph{sink}. The label of an edge is a linear polynomial in $\mathbb{F}[X]$. The weight of a path is the product of the labels of its edges. The polynomial computed by $G$ is the sum of the weights of all the paths from $s$ to $t$ in $G$. 

  A \emph{Multilinear ABP} ($\mlABP$) is an algebraic branching program such that for any path from the source to sink, the labels of edges on that path are linear polynomials over pairwise-disjoint sets of variables.

\paragraph*{Composing ABPs in series and in parallel.}  Let $G_1,\ldots,G_k$ be ABPs (on disjoint sets of vertices) with sources $s_1,\ldots,s_k$ and sinks $t_1,\ldots, t_k$. We say that $G$ is obtained by composing $G_1,\ldots,G_k$ \emph{in parallel} if $G$ is obtained by identifying all sources $s_1,\ldots, s_k$ to obtain a single source $s$ and all the sinks $t_1,\ldots,t_k$ to obtain a single sink $t$. Note that the polynomial computed by $G$ is the sum of the polynomials computed by $G_1,\ldots,G_k.$

We say that $G$ is obtained by composing $G_1,\ldots,G_k$ \emph{in series} if $G$ is obtained by identifying $t_i$ with $s_{i+1}$ for each $i\in [k-1]$ to obtain an ABP with source $s_1$ and sink $t_k$. Note that the polynomial computed by $G$ is the product of the polynomials computed by $G_1,\ldots,G_k.$ Further if $G_1,\ldots,G_k$ are $\mlABP$s over disjoint sets of variables, then $G$ is also an $\mlABP.$
\begin{figure}
    
    \centering
    \begin{tikzpicture}[scale=1]
      \node[initialvertex] (v0) at (0,0) {$s_1$};
      \node[vertex] (v11) at (1,1) {};
      \node[vertex] (v12) at (1,-1) {};
      \node (mid1) at (2,1) {$\dots$};
      \node (mid2) at (2,0) {$\dots$};
      \node (mid3) at (2,-1) {$\dots$};
      \node[vertex] (v21) at (3,1) {};
      \node[vertex] (v22) at (3,-1) {};
      \node[finalvertex] (vf) at (4,0) {$t_1$};
      \node[initialvertex] (u0) at (5,0) {$s_2$};
      \node[vertex] (u11) at (6,1) {};
      \node[vertex] (u12) at (6,-1) {};
      \node (mid11) at (7,1) {$\dots$};
      \node (mid21) at (7,0) {$\dots$};
      \node (mid31) at (7,-1) {$\dots$};
      \node[vertex] (u21) at (8,1) {};
      \node[vertex] (u22) at (8,-1) {};
      \node[finalvertex] (uf) at (9,0) {$t_2$};
      \path[->]
      (v0) edge (v11)
      (v0) edge (v12)
      (v21) edge (vf)
      (v22) edge (vf)
      (u0) edge (u11)
      (u0) edge (u12)
      (u21) edge (uf)
      (u22) edge (uf);
      \draw[red, thick, dashed] (3.5,-0.5) rectangle node {$u$} (5.5,0.5);
    \end{tikzpicture}
    \hfill
    \begin{tikzpicture}[scale=1]
      \node[initialvertex] (v0) at (0,0) {$s_1$};
      \node[vertex] (v11) at (1,1) {};
      \node[vertex] (v12) at (1,-1) {};
      \node (mid1) at (2,1) {$\dots$};
      \node (mid2) at (2,0) {$\dots$};
      \node (mid3) at (2,-1) {$\dots$};
      \node[vertex] (v21) at (3,1) {};
      \node[vertex] (v22) at (3,-1) {};
      \node[finalvertex] (vf) at (4,0) {$t_1$};
      \node[initialvertex] (u0) at (0,-3) {$s_2$};
      \node[vertex] (u11) at (1,-2) {};
      \node[vertex] (u12) at (1,-4) {};
      \node (mid11) at (2,-2) {$\dots$};
      \node (mid21) at (2,-3) {$\dots$};
      \node (mid31) at (2,-4) {$\dots$};
      \node[vertex] (u21) at (3,-2) {};
      \node[vertex] (u22) at (3,-4) {};
      \node[finalvertex] (uf) at (4,-3) {$t_2$};
      \draw[red, thick, dashed] (-0.5,-3.5) rectangle node {$s$} (0.5,0.5);
      \draw[red, thick, dashed] (3.5,-3.5) rectangle node {$t$} (4.5,0.5);
      \path[->]
      (v0) edge (v11)
      (v0) edge (v12)
      (v21) edge (vf)
      (v22) edge (vf)
      (u0) edge (u11)
      (u0) edge (u12)
      (u21) edge (uf)
      (u22) edge (uf);     
    \end{tikzpicture}
    \caption{Compostion of ABPs in series and parallel.}
    \label{fig:abp-series}
  \end{figure}
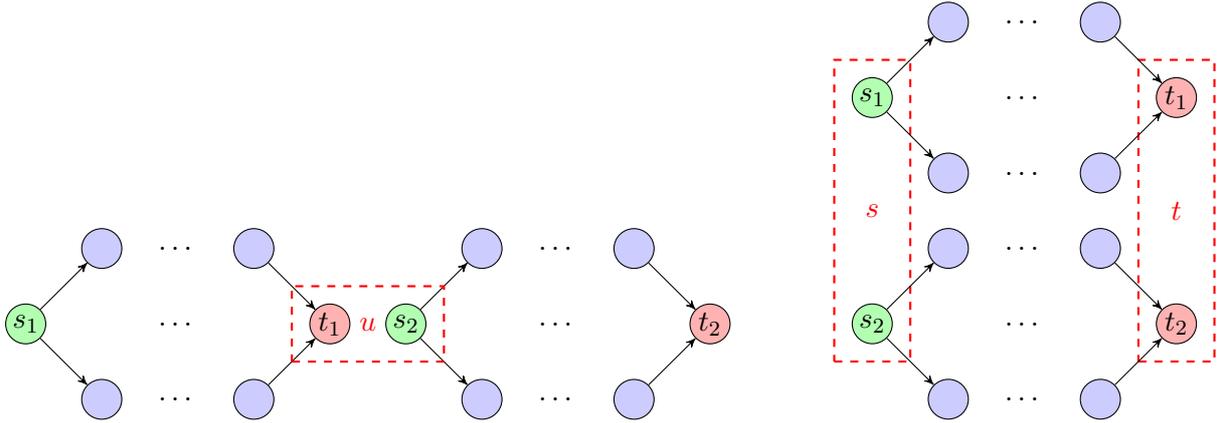

\subsection{Some structural lemmas}
\label{sec:struct}

Given a syntactically multilinear formula $F$ computing a polynomial $P\in \F[X],$ we define a \emph{variable-set labelling} of $F$ to be a labelling function that assigns to each gate $\Phi$ of $F$ a set $\Vars(\Phi)\subseteq X$ with the following properties.
\begin{enumerate}
\item For any gate $\Phi$ in $F$, $\supp(\Phi) \subseteq \Vars(\Phi)$. 
\item If $\Phi$ is an sum gate, with children $\Psi_1, \Psi_2, \ldots, \Psi_k$, then $\forall i \in [k]$, $\Vars(\Psi_i) = \Vars(\Phi)$. 
\item If $\Phi$ is a product gate, with children $\Psi_1, \Psi_2, \ldots, \Psi_k$, then $\Vars(\Phi) = \cup_{i=1}^k \Vars(\Psi_i)$ and the sets $\Vars(\Psi_i)$ ($i\in [k]$) are pairwise disjoint.
\end{enumerate}
We call a syntactically multilinear formula $F$ \emph{variable-labelled} if it is equipped with a labelling function as above. For such an $F$, we define $\Vars(F)$ to be $\Vars(\Phi_0)$, where $\Phi_0$ is the output gate of $F$.

The following lemma shows that we may always assume that a syntactically multilinear formula $F$ is variable-labelled. 

\begin{proposition}
\label{prop:vars}
For each syntactically multilinear formula $F$ computing $P\in \F[X],$ there is a variable-labelled syntactically multilinear formula $F'$ of the same size as $F$ computing $P$ such that $\Vars(\Phi)\subseteq X$ for each gate $\Phi$ of $F'$ and $\Vars(F) = X.$
\end{proposition} 

\begin{proof}
The formula $F'$ is the same as the formula $F$ along with variable-set labels $\Vars(\Phi)$ for each gate $\Phi$ of $F'$. These labels are defined by downward induction on the structure of $F$ as follows. 

For the output gate $\Phi_0$ of $F$, we define $\Vars(\Phi_0) = X$ (this ensures that $\Vars(F) = X$). If $\Phi$ is a sum gate with children $\Psi_1, \ldots, \Psi_k$ and $\Vars(\Phi) =S \subseteq X$, then for each $1\leq i \leq k$, we define $\Vars(\Psi_i) = S$. If $\Phi$ is a product gate with children $\Psi_1, \ldots \Psi_k$ and $\Vars(\Phi) = S \subseteq X$, then we define $\Vars(\Psi_i) = \supp(\Psi_i)$ for $1 \leq i \leq k-1$ and $\Vars(\Psi_k) = S \setminus \left(\cup_{i=1}^{k-1} \Vars(\Psi_i)\right)$.

It is easy to check that the above labelling is indeed a valid variable-set labelling. 
\end{proof}

We will use the following structural result that converts any small-depth multilinear circuit to a small-depth syntactic multilinear formula without a significant blowup in size.

\begin{lemma}[Raz and Yehudayoff~\cite{ry09}, Lemma 2.1]
\label{lem:RY-nf-ckts}
For any multilinear circuit $C$ of product-depth at most $\Delta$ and size at most $s$, there is a syntactic multilinear $(\Sigma\Pi)^{\Delta}\Sigma$ formula $F$ of size at most $(\Delta+1)^2\cdot s^{2\Delta+1}$ computing the same polynomial as $C$. 
\end{lemma}

\subsection{Useful Probabilistic Estimates}
\label{sec:misc}

We will need the Chernoff bound~\cite{chernoff, hoeffding} for sums of independent Boolean random variables. We use the version from the book of Dubhashi and Panconesi~\cite[Theorem 1.1]{DubhashiPanconesi}.

\begin{theorem}[Chernoff bound]
\label{thm:Chernoff}
Let $W_1,\ldots,W_n$ be independent $\{0,1\}$-valued random variables and let $W = \sum_{i\in [n]} W_i.$ Then we have the following.
\begin{enumerate}
\item For any $\varepsilon > 0,$ 
\[
\prob{}{W > (1+\varepsilon) \avg{}{W}} \leq \exp(-\frac{\varepsilon^2}{3}\avg{}{W})\ \text{ and}\  \prob{}{W < (1-\varepsilon) \avg{}{W}} \leq \exp(-\frac{\varepsilon^2}{2}\avg{}{W}).
\] 
\item For any $t > 2e\avg{}{W}$,
\[
\prob{}{W > t}\leq 2^{-t}.
\]
\end{enumerate}
\end{theorem}

We will need the following simple fact that reduces the problem of proving a large deviation bound to the problem of bounding the probability of a conjunction of events. 

\begin{proposition}
\label{prop:ANDthr}
Let $W_1,\ldots,W_n$ be $n$ $\{0,1\}$-valued random variables (not necessarily independent) and let $W = (W_1,\ldots,W_n)$. Assume that for some $p\in [0,1]$ and for any $\beta\in \{0,1\}^n$ we have $\prob{}{W = \beta}\leq p.$ Then, we have $\prob{}{\sum_i W_i \leq r} \leq n^r\cdot p.$
\end{proposition}

\begin{proof}
  For any $\beta \in \{0,1\}^n$, let weight of $\beta$ (denoted $\wt(\beta)$) be $\sum_{i=1}^n\beta(i)$. It is now easy to see that
  \begin{align*}
    \prob{}{\sum_i W_i \leq r} = \sum_{k=0}^r\left(\sum_{\substack{\beta\in\{0,1\}^n\\\wt(\beta)=k}}\prob{}{W=\beta}\right) \leq \sum_{k=0}^r{n\choose k}\cdot p \leq n^r \cdot p.
  \end{align*}
  The first inequality comes from the fact that there are ${n\choose k}$ many bit vectors $\beta$ in $\{0,1\}^n$ of weight exactly $k$ and $\prob{}{W=\beta} \leq p$ for each of them, and the last inequality comes from an upper bound of $n^r$ on $\sum_{k=0}^r{n\choose k}$.

\end{proof}

For a Boolean random variable $W,$ define the \emph{unbias} of $W$ by 
\[
\mathrm{Unbias}(W) = \min\{\prob{}{W = 0}, \prob{}{W=1}\}
\]
and \emph{bias} of $W$ by
\[
  \mathrm{Bias}(W) = \abs{\prob{}{W=0}-\prob{}{W=1}} = \abs{\mathbb{E}[(-1)^{W}]}.
\]
Note that $\mathrm{Unbias}(W)\in [0,1/2].$ The following fact follows directly from the aforementioned definitions.

\begin{proposition}
  \label{prop:bias-unbias}
  For a Boolean random variable $W,$ $\mathrm{Bias}(W) + 2\cdot\mathrm{Unbias}(W) = 1$.
\end{proposition}

The following fact is folklore, but we state it here in the exact form we need and prove it for completeness.

\begin{proposition}
\label{prop:xorunbias}
Let $W_1,\ldots,W_n$ be independent $\{0,1\}$-valued random variables and assume $W = \bigoplus_{j\in [n]} W_j$. Then $\mathrm{Unbias}(W) \geq \min\{(1/2)\cdot \sum_j \mathrm{Unbias}(W_j),1/10\}.$ 
\end{proposition}

\begin{proof}
  From Proposition~\ref{prop:bias-unbias}, we get that $\unbias(W) = \frac{1 - \bias(W)}{2} = \frac{1-\abs{\mathbb{E}[(-1)^W]}}{2}$. We also get the following.
  \begin{align*}
    \abs{\mathbb{E}[(-1)^W]} = \abs{\mathbb{E}[(-1)^{\bigoplus_{j\in [n]} W_j}]} = \abs{\mathbb{E}[(-1)^{\sum_{j\in[n]}W_j}]} = \abs{\prod_{j\in[n]}{\mathbb{E}[(-1)^{W_j}]}} = \prod_{j\in[n]}\bias(W_j).
  \end{align*}
  Thus,
  \begin{align*}
    \unbias(W) = \frac{1-\prod_{j\in[n]}(1-2\cdot\unbias(W_j))}{2}
    &\geq \frac{1-\prod_{j\in[n]}e^{-2\cdot\unbias(W_{j})}}{2}\\ &=\frac{1 - e^{-2(\sum_{j\in[n]}\unbias(W_j))}}{2}.
  \end{align*}
  If $\sum_{j\in[n]}\unbias(W_j) \geq \frac{1}{2}$, then $\unbias(W) \geq \frac{1-e^{-1}}{2}\geq \frac{1}{10}$. Else, $\unbias(W)\geq \frac{\sum_{j\in[n]}\unbias(W_j)}{2}.$ The latter follows from the inequality $e^{x}\leq 1 - \frac{x}{2}$ for all $x\in[0,1/2]$. This proves the proposition.
\end{proof}

Let $A_1,\ldots,A_N$ be any $N$ independent random variables taking values over any finite set. Let $B_1,\ldots,B_M$ be $M$ \emph{Boolean} random variables with $B_i = f_i(A_j : j\in S_i)$ for some function $f_i$ and some $S_i\subseteq [N].$ We say $B_1,\ldots,B_M$ are read-$k$ for $k\geq 1$ if each $j\in [N]$ belongs to at most $k$ many sets $S_i.$ 

A result of Janson~\cite{Janson} yields concentration bounds for sums of read-$k$ Boolean random variables. We use the following form of the bound that appears in the result of Gavinsky, Lovett, Saks and Srinivasan~\cite{GLSS}.

\begin{theorem}[\cite{GLSS}]
\label{thm:GLSS}
If $B_1,\ldots, B_M$ are read-$k$ Boolean random variables and $B = \sum_i {B_i}$, then 
\[
\prob{}{B\leq \frac{1}{2}\avg{}{B}} \leq \exp(-\Omega(\avg{}{B}/k)).
\]
\end{theorem}

\section{The Hard polynomials and their restrictions}
\label{sec:polysandrests}
\subsection{The Hard polynomials}
\label{sec:polynomial}
The hard polynomial for circuits of product-depth $\Delta$ will be defined to be the polynomial computed by a suitable ABP. The definition is by induction on the product depth $\Delta$. We also define some auxiliary variable sets that will be useful later when we restrict these polynomials. 

Let $m\in \mathbb{N}$ be a growing parameter. We will define for each $\Delta\geq 1$ an ABP $G^{(\Delta)}$ on a variable set $X^{(\Delta)}$ of size  $n_\Delta$, along with some auxiliary variable sets $Y^{(\Delta)}$ and $Z^{(\Delta)}$. The parameter $n_\Delta$ itself is defined inductively as follows. Let $n_0 = 8$ and for each $\Delta\geq 1,$ define
\begin{equation}
\label{eq:defn_im_i}
n_\Delta = 2 \cdot m\cdot n_{\Delta-1}.
\end{equation}
Observe that $n_{\Delta} = 8\cdot (2m)^{\Delta}$ for $\Delta \geq 1.$

\paragraph{The Variable sets.} For each $\Delta \geq 0,$ we will define three disjoint variable sets $X^{(\Delta)},Y^{(\Delta)}$ and $Z^{(\Delta)}.$ We will have $|X^{(\Delta)}| = n_\Delta$ for each $\Delta.$

Let $X^{(0)}$ be the set $\{x_{i,j}^{(0)}\ |\ i\in [4],j\in [2]\}$, $Y^{(0)} = \{y_1^{(0)},y_2^{(0)}\}$ and $Z^{(0)} = \{z_1^{(0)},z_2^{(0)}\}.$ Note that $|X^{(0)}| = n_0 = 8.$

Given $X^{(\Delta)}, Y^{(\Delta)}$ and $Z^{(\Delta)},$ we now define $X^{(\Delta+1)}, Y^{(\Delta+1)}$ and $Z^{(\Delta + 1)}$ as follows.

\paragraph{Clones, segments and half-segments.} For each $i\in [m]$ and $j\in [2]$, let $X^{(\Delta+1)}_{i,j}$, $Y^{(\Delta+1)}_{i,j}$ and $Z^{(\Delta+1)}_{i,j}$ be pairwise disjoint copies of $X^{(\Delta)}, Y^{(\Delta)}$ and $Z^{(\Delta)}$ respectively. Define $X^{(\Delta+1)} = \bigcup_{i,j} X^{(\Delta+1)}_{i,j}$; note that $|X^{(\Delta+1)}| = 2m |X^{(\Delta)}| = n_{\Delta+1}$ as desired. Define $Y^{(\Delta+1)}$ to be $\bigcup_{i} Y^{(\Delta+1)}_{i,1} \cup Y^{(\Delta+1)}_{i,2}  \cup \{y^{(\Delta+1)}_i\}$ and $Z^{(\Delta+1)}$ to be $\bigcup_{i} Z^{(\Delta+1)}_{i,1} \cup Z^{(\Delta+1)}_{i,2}  \cup \{z^{(\Delta+1)}_i\}$, where $y_i^{(\Delta+1)}$ and $z_i^{(\Delta+1)}$ are fresh variables. 

By the inductive definition of the sets $X^{(\Delta+1)}$ above, we see that $X^{(\Delta+1)}$ is made up of $2m$ copies of $X^{(\Delta)}$, which we denote by $X^{(\Delta+1)}_{i,j}$ for $i\in [m]$ and $j\in [2].$ Using this fact inductively, we see that for any $t\in \{0,\ldots,\Delta+1\},$ $X^{(\Delta+1)}$ contains $(2m)^t$ copies of $X^{(\Delta+1-t)}$. Each such copy is uniquely labelled by a tuple $\omega = ((i_1,j_1),\ldots,(i_t,j_t))\in ([m]\times [2])^t$; we denote this copy by $X^{(\Delta+1)}_\omega$ and call this an \emph{$X^{(\Delta+1-t)}$-clone}. In a similar way, we see that for each $\omega\in ([m]\times [2])^t$, $Y^{(\Delta+1)}$ and $Z^{(\Delta+1)}$ contain copies of $Y^{(\Delta+1-t)}$ and $Z^{(\Delta+1-t)}$ respectively, which we denote as $Y^{(\Delta+1)}_\omega$ and $Z^{(\Delta+1)}_\omega$ respectively, and call $Y^{(\Delta+1-t)}$-clones and $Z^{(\Delta+1-t)}$-clones respectively.

Say $X^{(\Delta')}_\omega$ is an $X^{(\Delta)}$-clone for some $\Delta' \geq \Delta \geq 1.$ Then we refer to $X^{(\Delta')}_{(\omega,(i,j))}$ (which is an $X^{(\Delta-1)}$-clone) as the $(i,j)$th \emph{half-segment} of $X^{(\Delta')}_\omega$. Further, we refer to $X^{(\Delta')}_{(\omega,(i,1))}\cup X^{(\Delta')}_{(\omega,(i,2))}$ as the $i$th \emph{segment} of $X^{(\Delta')}_\omega,$ which we will denote $X^{(\Delta')}_{(\omega,i)}$; note that $X^{(\Delta')}_\omega$ is a union of its $m$ segments. Similarly a $Y^{(\Delta)}$-clone $Y^{(\Delta')}_\omega$ (resp.\ a $Z^{(\Delta)}$-clone $Z^{(\Delta')}_\omega$) is a disjoint union of $m$ segments $Y^{(\Delta')}_{(\omega,i)}$ (resp.\ $Z^{(\Delta')}_{(\omega,i)}$) for $i\in [m]$, the $i$th of which is made up of two half-segments $Y^{(\Delta')}_{(\omega,(i,1))}$ and $Y^{(\Delta')}_{(\omega,(i,2))}$ (resp.\ $Z^{(\Delta')}_{(\omega,(i,1))}$ and $Z^{(\Delta')}_{(\omega,(i,2))}$) and an additional fresh variable that we denote $y^{(\Delta')}_{(\omega,i)}$ (resp.\ $z^{(\Delta')}_{(\omega,i)}$). We refer to a set of the form $X^{(\Delta')}_{(\omega,i)}$ as an $X^{(\Delta)}$-segment\footnote{The reader may want to read ``$X^{(\Delta)}$-segment'' as ``a segment of (a clone of) $X^{(\Delta)}$'' and similarly for other segments and half-segments.} and a set of the form $X^{(\Delta')}_{(\omega,(i,j))}$ as an $X^{(\Delta)}$-half-segment (similarly $Y^{(\Delta)}$-segment, $Y^{(\Delta)}$-half-segment, $Z^{(\Delta)}$-segment and $Z^{(\Delta)}$-half-segment).

To summarize, given any $\Delta'\geq \Delta\geq 0$ the set $X^{(\Delta')}$ contains $(2m)^{\Delta'-\Delta}$ many $X^{(\Delta)}$-clones which are indexed by elements of the set $([m]\times [2])^{\Delta'-\Delta}$. If $\Delta \geq 1,$ each such $X^{(\Delta)}$-clone contains $m$ many $X^{(\Delta)}$-segments, which are indexed by the set $([m]\times [2])^{\Delta'-\Delta} \times [m];$ and each segment contains two $X^{(\Delta-1)}$-clones, also referred to as $X^{(\Delta)}$-half-segments. Finally, these definitions extend naturally to $Y^{(\Delta')}$- and $Z^{(\Delta')}$-clones.

\paragraph{The Hard polynomial.} The hard polynomial for product-depth $\Delta$ is defined by induction on $\Delta.$

Define $G^{(0)}$ to be an ABP as follows. Assume that $G^{(0)}$ has source vertex $s_0$, sink vertex $t_0$ and an intermediate vertex $u_0$. The graph of the ABP consists of two disjoint paths $\pi_1$ and $\pi_2$ of length $2$ each from $s_0$ to $u_0$ and two disjoint paths $\pi_3$ and $\pi_4$ of length $2$ each from $u_0$ to $t_0$ (see Figure~\ref{fig:g0}). The edges of $\pi_i$ are labelled by the distinct variables $x_{i,1}^{(0)}$ and $x_{i,2}^{(0)}$. Let $X^{(0)}$ be the set $\{x_{i,j}^{(0)}\ |\ i\in [4],j\in [2]\}$ of variables that appear in $G^{(0)}.$ Note that $G^{(0)}$ computes a polynomial over the variable set $X^{(0)}$ defined above.

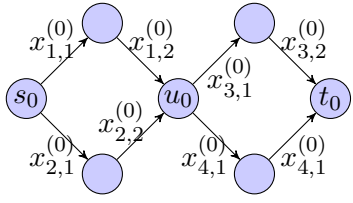
\begin{figure}
\centering
 \begin{tikzpicture}
    \node[vertex] (a) at (0,0) {$s_{0}$};
    \node[vertex] (b) at (1,1) {};
    \node[vertex] (c) at (1,-1) {};
    \node[vertex] (d) at (2,0) {$u_{0}$};

    \node[vertex] (f) at (3,1) {};
    \node[vertex] (g) at (3,-1) {};
    \node[vertex] (h) at (4,0) {$t_{0}$};

    \draw[->] (a) -- node[above,pos=0.25] {$x_{1,1}^{(0)}$} (b);
    \draw[->] (b) -- node[above,pos=0.75] {$x_{1,2}^{(0)}$} (d);
    \draw[->] (d) -- node[below,pos=0.8] {$x_{3,1}^{(0)}$} (f);
    \draw[->] (f) -- node[above,pos=0.75] {$x_{3,2}^{(0)}$} (h);

    \draw[->] (a) -- node[below,pos=0.25] {$x_{2,1}^{(0)}$} (c);
    \draw[->] (c) -- node[above,pos=0.1] {$x_{2,2}^{(0)}$} (d);
    \draw[->] (d) -- node[below,pos=0.25] {$x_{4,1}^{(0)}$} (g);
    \draw[->] (g) -- node[below,pos=0.75] {$x_{4,1}^{(0)}$} (h);
  \end{tikzpicture}
  \caption{The ABP $G^{(0)}$. All edges go from left to right.}
  \label{fig:g0}
\end{figure}

Now fix any $\Delta \geq 0$. Given $G^{(\Delta)}$, an ABP on variable set $X^{(\Delta)}$, we inductively define the ABP $G^{(\Delta +1)}$ an ABP on variable set $X^{(\Delta+1)}$ as follows. 

\begin{itemize}
\item For each $i\in[m]$ and $j\in[2]$, let $G^{(\Delta+1)}_{i,j}$ be a copy of $G^{(\Delta)}$  on the variable set $X^{(\Delta+1)}_{i,j}$. 

\item Let $H^{(\Delta+1)}_{i}$ be the ABP obtained by composing $G^{(\Delta+1)}_{i,1}$ and $G^{(\Delta+1)}_{i,2}$ in parallel. This ABP is defined over variable set $X^{(\Delta+1)}_i.$ 

\item Finally, let $G^{(\Delta+1)}$ be the ABP obtained by composing $H^{(\Delta+1)}_{1}, \ldots, H^{(\Delta+1)}_{m}$ in series, in that order. 
\end{itemize}

From the definition of $G^{(\Delta)}$ it is easy to observe the following properties.
\begin{proposition}
\label{prop:graph}
For any $\Delta \geq 0$, $G^{(\Delta)}$ has a unique source, say $s_\Delta$, and a unique sink, say $t_\Delta$. Also, each edge in $G^{(\Delta)}$ appears on some source to sink path. 
\end{proposition}

We define $P^{(\Delta)}$ to be the multilinear polynomial in ${\F}[X^{(\Delta)}]$ computed by the ABP $G^{(\Delta)}$. We can also note the following properties of the polynomial computed by $G^{(\Delta)}$. 
\begin{proposition}[Properties of $P^{(\Delta)}$]
\label{prop:poly}
\begin{enumerate}
\item $P^{(0)}$ is the polynomial $\sum_{i_1 \in \{1,2\}, i_2 \in \{3,4\}} x^{(0)}_{i_1,1} \cdot  x^{(0)}_{i_1,2} \cdot  x^{(0)}_{i_2,1} \cdot  x^{(0)}_{i_2,2}.$

\item For each $\Delta \geq 0$, $$P^{(\Delta+1)}(X^{(\Delta+1)}) = \prod_{i\in [m]} \left(P^{(\Delta)}(X^{(\Delta+1)}_{i,1}) + P^{(\Delta)}(X^{(\Delta+1)}_{i,2})\right)  $$
\item The polynomial $P^{(0)}$ is computed by a $\Sigma \Pi$ multilinear formula of size $O(1).$ For each $\Delta \geq 1,$ $P^{(\Delta)}$ can be computed by a syntactic multilinear $(\Pi\Sigma)^{\Delta} \Pi$ formula of size $O(n_\Delta).$
\end{enumerate}
\end{proposition}

As with the variable sets, we see that for any $\Delta' \geq \Delta\geq 1,$ and for each $\omega\in ([m]\times [2])^{\Delta'-\Delta}$, the ABP $G^{(\Delta')}$ contains a corresponding copy of $G^{(\Delta)}$ on the variable set $X^{(\Delta')}_\omega.$ We denote this copy of $G^{(\Delta)}$ by $G^{(\Delta')}_\omega.$ The ABP $G^{(\Delta')}_\omega$ is obtained by composing in series the ABPs $H^{(\Delta')}_{(\omega,1)},\ldots,H^{(\Delta')}_{(\omega,m)}$ (defined on the segments $X^{(\Delta')}_{(\omega,1)},\ldots,X^{(\Delta')}_{(\omega,m)}$ respectively) in that order.

\subsection{Restrictions}
\label{sec:restriction}

We define a random multilinear\footnote{See Section~\ref{sec:polyrest} for the definition.} $(X^{(\Delta+1)}, Y^{(\Delta+1)},Z^{(\Delta+1)})$-restriction $\rho^{(\Delta+1)}$ by an inductively defined sampling process. 

\paragraph{Base case.} Let $\rho^{(0)}$ be defined as follows: $\rho^{(0)}(x^{(0)}_{1,1}) = y^{(0)}_1, \rho^{(0)}(x^{(0)}_{2,1}) = z^{(0)}_1$, $\rho^{(0)}(x^{(0)}_{3,1}) = y^{(0)}_2, \rho^{(0)}(x^{(0)}_{4,1}) = z^{(0)}_2$ (see Figure~\ref{fig:base-restriction}). That is, we set the first variable in each of the paths $\pi_1,\pi_2,\pi_3,\pi_4$ to a distinct variable in $Y^{(0)}\cup Z^{(0)}$. Also, we set $\rho^{(0)}(x^{(0)}_{i,2}) = 1$ for each $i \in [2]$. That is, we set all the remaining variables to the constant $1$. (Note that $\rho^{(0)}$ is in fact a \emph{deterministic} $(X^{(0)},Y^{(0)},Z^{(0)})$-restriction.) It can be checked that $\rho^{(0)}$ is multilinear. 

\paragraph{Inductive case.} Let us now assume that we have defined the process for sampling the  random multilinear $(X^{(\Delta)},Y^{(\Delta)},Z^{(\Delta)})$-restriction $\rho^{(\Delta)}$. We define $\rho^{(\Delta+1)}$ by the following sampling process.

For each $i \in [m]$, we sample a random multilinear $(X^{(\Delta+1)}_i,Y^{(\Delta+1)}_i,Z^{(\Delta+1)}_i)$-restriction $\rho^{(\Delta+1)}_i$ independently using the general sampling process described below. We then define $\rho^{(\Delta+1)}$ as a composition of these restrictions, i.e. $\rho^{(\Delta+1)} = \rho^{(\Delta+1)}_1\circ \cdots \circ \rho^{(\Delta+1)}_{m}$ (as defined in Section~\ref{sec:polyrest}). Clearly, $\rho^{(\Delta+1)}$ is multilinear since each $\rho^{(\Delta+1)}_i$ is multilinear.

We now give a general sampling process to sample a random multilinear $(X,Y,Z)$-restriction where $X$ is an $X^{(\Delta+1)}$-segment and $Y$ and $Z$ are the corresponding $Y^{(\Delta+1)}$- and $Z^{(\Delta+1)}$-segments respectively.

For the remainder of this section, let $\Delta'\geq \Delta+1 \geq 1$ be arbitrary and for some $\omega\in ([m]\times [2])^{\Delta'-\Delta-1}$ and $i\in [m]$, let $X$ denote the $i$th segment $X^{(\Delta')}_{(\omega,i)}$ of the $X^{(\Delta+1)}$-clone $X^{(\Delta')}_\omega$ and let $Y = Y^{(\Delta')}_{(\omega,i)}$, $Z = Z^{(\Delta')}_{(\omega,i)}.$ For any $j\in [2]$, let $X_j,Y_j,$ and $Z_j$ denote the $X^{(\Delta)}$-clones $X^{(\Delta')}_{(\omega,(i,j))}, Y^{(\Delta')}_{(\omega,(i,j))}$ and $Z^{(\Delta')}_{(\omega,(i,j))}$ respectively, and let $G_j$ denote the ABP $G^{(\Delta')}_{(\omega,(i,j))}.$ 

Let $y,z$ denote the variables $y^{(\Delta')}_{(\omega,i)}$ and $z^{(\Delta')}_{(\omega,i)}$ in the sets $Y$ and $Z$ respectively (recall that we have $Y = Y_1\cup Y_2 \cup \{y\}$ and $Z = Z_1\cup Z_2 \cup \{z\}$). Let $H$ denote the ABP $H^{(\Delta')}_{(\omega,i)}$ (recall that $H$ is the parallel composition of $G_1$ and $G_2$).

We show now how to sample for any such $X,Y,Z$ a random multilinear $(X,Y,Z)$-restriction $\rho$.

\paragraph{Sampling Algorithm $\mc{A}$ for $(X,Y,Z)$-restriction $\rho$.}

\begin{itemize}
\item[$E_1$:] Set all the variables from the set $X_2$ to $0$. For the variables in $X_1$, sample a random  $(X_1,Y_1,Z_1)$-restriction $\rho_1$ using the sampling procedure for $\rho^{(\Delta)}$. (Recall that $X_1,Y_1,Z_1$ are  $X^{(\Delta)}$-, $Y^{(\Delta)}$-, and $Z^{(\Delta)}$-clones respectively.) Set variables in $X_{1}$ according to $\rho_1$.\label{i:rest1}
\item[$E_2$:] \label{i:rest2} This is the same as $E_1$ except that the roles of $X_1$ and $X_2$ are exchanged. Formally, we set all the variables from the set $X_1$ to $0$ and apply a random $(X_2,Y_2,Z_2)$-restriction $\rho_2$, sampled using the sampling procedure for $\rho^{(\Delta)}$, to $X_2$ . 
\item[$E_3$:] Choose variables $x_1$ and $x_2$ independently and uniformly at random from $X_1$ and $X_2$ respectively. The variable $x_j$ ($j\in [2]$) labels a unique edge, say $e_j$, in ABP $G_j$; let $\pi_{e_j}$ denote the lexicographically smallest source to sink path in $G^{(\Delta+1)}_{i,j}$ containing $e_j$. For any variable $x\in X$ which does not label an edge on either of the paths $\pi_{e_1}$ or $\pi_{e_2}$, set $x$ to the constant $0$. 

Set $x_1$ to $y$ and $x_2$ to $z$. 

For any edge $e$ other than $e_1,e_2$ which lies on either $\pi_{e_1}$ or $\pi_{e_2},$ set the variable labelling $e$ to $1$.

\end{itemize}

\subsection{Properties of $\rho^{(\Delta)}$}
\label{sec:rest-props}

Fix some $\Delta \geq 1.$ For a given random choice of $\rho^{(\Delta)}$, we use $\tilde{Y}$ and $\tilde{Z}$ to denote $\Img(\rho^{(\Delta)})\cap Y^{(\Delta)}$  and $\Img(\rho^{(\Delta)})\cap Z^{(\Delta)}$ respectively. Note that these are \emph{random} sets.

Given any $f\in \F[X^{(\Delta)}],$ the polynomial $f|_{\rho^{(\Delta)}}$ belongs to the set $\F[\tilde{Y}\cup \tilde{Z}].$ By the multilinearity of $\rho^{(\Delta)},$ if $f$ is multilinear, then so is $f|_{\rho^{(\Delta)}}.$

\begin{lemma}
\label{lem:fullrank}
With probability $1$, $\relrk_{(\tilde{Y},\tilde{Z})}(P^{(\Delta)}|_{\rho^{(\Delta)}}) = 1.$
\end{lemma}
\begin{proof}

We prove the lemma by induction on $\Delta.$

The base case corresponds to $\Delta = 0,$ in which case $\tilde{Y} = Y^{(0)} = \{y^{(0)}_1,y^{(0)}_2\}$ and $\tilde{Z} = Z^{(0)} = \{z^{(0)}_1,z^{(0)}_2\}$ deterministically. From the definition of $P^{(0)}$ and $\rho^{(0)}$, we know that $P^{(0)}|_{\rho^{(0)}} = (y^{(0)}_1+ z^{(0)}_1)(y^{(0)}_2 + z^{(0)}_2)$, which can easily be seen to have relative rank $1$ w.r.t. the partition $(\tilde{Y},\tilde{Z})$.

\begin{figure}[H]
\centering
  \begin{tikzpicture}
    \node[vertex] (a) at (0,0) {$s_0$};
    \node[vertex] (b) at (1,1) {};
    \node[vertex] (c) at (1,-1) {};
    \node[vertex] (d) at (2,0) {$u_0$};

    \node[vertex] (f) at (3,1) {};
    \node[vertex] (g) at (3,-1) {};
    \node[vertex] (h) at (4,0) {$t_0$};

    \draw[->] (a) -- node[above,pos=0.25] {$y_{1}^{(0)}$} (b);
    \draw[->] (b) -- node[above,pos=0.75] {$1$} (d);
    \draw[->] (d) -- node[above,pos=0.25] {$y_{2}^{(0)}$} (f);
    \draw[->] (f) -- node[above,pos=0.75] {$1$} (h);

    \draw[->] (a) -- node[below,pos=0.25] {$z_{1}^{(0)}$} (c);
    \draw[->] (c) -- node[below,pos=0.75] {$1$} (d);
    \draw[->] (d) -- node[below,pos=0.25] {$z_{2}^{(0)}$} (g);
    \draw[->] (g) -- node[below,pos=0.75] {$1$} (h);
  \end{tikzpicture}
  \label{fig:base-restriction}
  \caption{The effect of $\rho^{(0)}$ on $G^{(0)}$}
\end{figure}
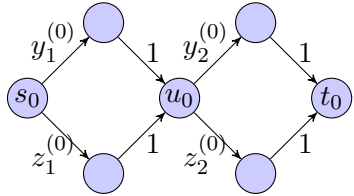

Recall (Proposition~\ref{prop:poly}) that for each $\Delta \geq 1$, $$P^{(\Delta)}(X^{(\Delta)}) = \prod_{i\in [m]} \left(P^{(\Delta-1)}(X^{(\Delta)}_{i,1}) + P^{(\Delta-1)}(X^{(\Delta)}_{i,2})\right)  $$
Let $Q^{(\Delta)}_i$ be equal to $P^{(\Delta-1)}(X^{(\Delta)}_{i,1}) + P^{(\Delta-1)}(X^{(\Delta)}_{i,2})$. The random restriction $\rho^{(\Delta)}$ is defined to be the composition $\rho^{(\Delta)}_1\circ \cdots \circ \rho^{(\Delta)}_m$ where each $\rho^{(\Delta)}_i$ is a random multilinear $(X^{(\Delta)}_i,Y^{(\Delta)}_i,Z^{(\Delta)}_i)$-restriction sampled according to the algorithm $\mc{A}$ described in Section~\ref{sec:restriction}. Thus, we have for any fixing of $\rho^{(\Delta)}_i$ ($i\in [m]$), 
\[
P^{(\Delta)}|_{\rho^{(\Delta)}} = \prod_{i\in [m]} Q^{(\Delta)}|_{\rho^{(\Delta)}_i}
\]
and hence by Proposition~\ref{prop:relrk}, we have
\[
\relrk_{(\tilde{Y},\tilde{Z})}(P^{(\Delta)}|_{\rho^{(\Delta)}}) = \prod_{i\in [m]} \relrk_{(\tilde{Y}_i,\tilde{Z}_i)}(Q^{(\Delta)}|_{\rho^{(\Delta)}_i})
\]
where $\tilde{Y}_i := \Img(\rho^{(\Delta)}_i)\cap Y^{(\Delta)}_i$ and $\tilde{Z}_i := \Img(\rho^{(\Delta)}_i)\cap Z^{(\Delta)}_i$. So it suffices to argue that each term in the above product is $1$ with probability $1$. 

Fix an $i\in [m]$ and consider $\relrk_{(\tilde{Y}_i,\tilde{Z}_i)}(Q^{(\Delta)}|_{\rho^{(\Delta)}_i})$. There are three possibilities for the sampling algorithm $\mc{A}$ in choosing $\rho^{(\Delta)}_i$.

 Say $\mc{A}$ picks option $E_1$. Then, all variables in $X^{(\Delta)}\setminus X^{(\Delta)}_{i,1} = X^{(\Delta)}_{i,2}$ are set to $0$ and $\rho^{(\Delta)}_i$ is simply a copy  of $\rho^{(\Delta-1)}$ defined w.r.t. the sets $(X^{(\Delta)}_{i,1},Y^{(\Delta)}_{i,1},Z^{(\Delta)}_{i,1})$. Thus, by induction $\relrk_{(\tilde{Y}_i,\tilde{Z}_i)}(Q^{(\Delta)}|_{\rho^{(\Delta)}_i})=\relrk_{(\tilde{Y}_i,\tilde{Z}_i)}(P^{(\Delta-1)}|_{\rho^{(\Delta)}_i}) = 1$. A similar reasoning works when $\mc{A}$ picks option $E_2.$
 
In case $\mc{A}$ picks $E_3$ then $Q^{(\Delta)}_i = (y_i^{(\Delta)} +z_i^{(\Delta)})$ and $\tilde{Y}_i = \{y_i^{(\Delta)}\}, \tilde{Z}_i = \{z_i^{(\Delta)}\}.$ It is then easily checked that $\relrk_{(\tilde{Y}_i,\tilde{Z}_i)}(Q^{(\Delta)}|_{\rho^{(\Delta)}_i})=1$. This completes the induction and proves the lemma.
\end{proof}

\section{The Main Result}
\label{sec:mainthm}

The main result of this paper is the following.

\begin{theorem}
\label{thm:PDeltalbd}
 Let $m,\Delta\in \mathbb{N}$ be growing parameters with $\Delta = m^{o(1)}.$\footnote{Since $n_\Delta = O(m)^{\Delta}$, this is equivalent to requiring that $\Delta = o(\log n_\Delta/\log \log n_\Delta).$} Assume that the polynomials $P^{(\Delta)}(X^{(\Delta)})$ are as defined in Section~\ref{sec:polynomial}. Then any multilinear circuit $C$ of product-depth $\Delta$ computing $P^{(\Delta)}$ must have a size of at least $\exp(m^{\Omega(1)}) = \exp(n_{\Delta}^{\Omega(1/\Delta)}).$
\end{theorem}

The following corollary is immediate from the lower bound in Theorem~\ref{thm:PDeltalbd} and the formula upper bound in Proposition~\ref{prop:poly}.

\begin{corollary}
\label{cor:abstract-pdepth}
Assume $\Delta = \Delta(n) = o(\log n/\log \log n)$. For all large enough $n\in \mathbb{N}$, there is an explicit multilinear polynomial on $n$ variables that has a multilinear formula of size $O(n)$ and product-depth $\Delta(n)+1$ but no multilinear circuit of size at most $\exp(n^{\Omega(1/\Delta)})$ and product-depth at most $\Delta(n).$ 
\end{corollary}

\begin{proof}
We can find $m = m(n)$ so that that $\Delta = m^{o(1)}$ and the corresponding $n_\Delta \in [\sqrt{n}, n]$ for large enough $n$. We now apply Proposition~\ref{prop:poly} and Theorem~\ref{thm:PDeltalbd} to obtain the result. 
\end{proof}

We also have a similar result for depth instead of product-depth.

\begin{corollary}
\label{cor:abstract-depth}
Assume $\Delta = \Delta(n) = o(\log n/\log \log n)$. For all large enough $n\in \mathbb{N}$, there is an explicit multilinear polynomial on $n$ variables that has a multilinear formula of size $n$ and depth $\Delta+1$ but no multilinear circuit of size $\exp(n^{\Omega(1/\Delta)})$ and depth at most $\Delta.$
\end{corollary}

\begin{proof}
Let $\Delta' = \lfloor (\Delta+1)/2\rfloor.$ 

Fix $m = m(n)$ so that that $\Delta = m^{o(1)}$ and the corresponding $n_{\Delta'} \in [\sqrt{n}, n/2]$ for large enough $n$. We define the explicit polynomial to be either $P^{(\Delta')}$ or the sum of two copies of $P^{(\Delta'-1)}$ depending on whether $\Delta$ is even or odd.

Assume that $\Delta$ is even. Then, $\Delta = 2\Delta'.$ In this case, the explicit polynomial is $P^{(\Delta')}$ which has a $(\Pi\Sigma)^{\Delta'}\Pi$ formula of size $O(n_{\Delta'}) = O(n)$ by Proposition~\ref{prop:poly}. Note that a $(\Pi\Sigma)^{\Delta'}\Pi$ formula is of depth $2\Delta'+1 = \Delta+1.$ This gives the upper bound. 

For the lower bound, we use Theorem~\ref{thm:PDeltalbd}. Any circuit $C$ of size at most $s$ and depth at most $\Delta = 2\Delta'$ can be converted to one of size at most $s$ and product-depth at most $\Delta'$ as follows. If $C$ contains two $\times$-gates $\Psi_1$ and $\Psi_2$ where $\Psi_1$ feeds into $\Psi_2$, we merge $\Psi_1$ and $\Psi_2$. Repeated applications of this procedure yields a circuit of depth at most $2\Delta'$ in which no input-to-output path can contain consecutive $\times$-gates. This circuit must have product-depth at most $\Delta^{\prime}.$ Clearly, this process can only reduce the size of the circuit. By Theorem~\ref{thm:PDeltalbd}, we see that if $C$ computes $P^{(\Delta')}$ it must have a size of at least $\exp(n_{\Delta'}^{\Omega(1/\Delta^{\prime})}) = \exp(n^{\Omega(1/\Delta)}).$

Now consider the case when $\Delta$ is odd. Then $\Delta = 2\Delta'-1.$ In this case, we define  the explicit polynomial to be $Q = P^{(\Delta'-1)}(X_1) + P^{(\Delta'-1)}(X_2)$ where $X_1$ and $X_2$ are disjoint copies of $X^{(\Delta'-1)}.$ Note that the number of variables in $Q$ is $2n_{\Delta'-1} \leq n_{\Delta^{\prime}}\leq n$ and by Proposition~\ref{prop:poly}, $Q$ has a $\Sigma(\Pi\Sigma)^{\Delta'-1}\Pi$ formula of size $O(n_{\Delta'-1}) = O(n)$. Note that a $\Sigma(\Pi\Sigma)^{\Delta'-1}\Pi$ formula is of depth $2(\Delta'-1)+2 = \Delta+1.$ This gives the upper bound.

For the lower bound, we proceed as was done in the case where $\Delta$ is even. Like before, we can assume that the most efficient circuit $C$ of depth $\Delta$ for $Q$ has the property that no path in $C$ contains consecutive $\times$-gates. Further, since $Q$ is easily seen to be an irreducible polynomial, we can also assume that the output gate is a $+$-gate. This implies that the product-depth of $C$ is at most $\lfloor \Delta/2\rfloor \leq \Delta'-1.$ Applying Theorem~\ref{thm:PDeltalbd} now yields the lower bound as in the even case.
\end{proof}

Theorem~\ref{thm:PDeltalbd} is proved by induction on the parameter $\Delta.$ For the purposes of induction, we need to prove a more technical statement from which we can easily infer Theorem~\ref{thm:PDeltalbd}. We give this technical statement below.

Recall that for a random mutlilinear $(X,Y,Z)$-restriction $\rho$, $\tilde{Y} = \Img(\rho)\cap Y$ and $\tilde{Z} = \Img(\rho)\cap Z$. Note that these are \emph{random} sets. Now, given a multilinear polynomial $f\in \F[X'\cup Y'\cup Z']$ where $X'\subseteq X$ and $Y',Z'$ are disjoint sets of variables, the polynomial $f|_\rho$ is a multilinear polynomial in $\F[\tilde{Y}\cup \tilde{Z}\cup Y'\cup Z'].$

\begin{theorem}
\label{thm:technical}
Let $m,k,\Delta, M\in \mathbb{N}$ be growing positive integer parameters with $k = m^{0.05}$ and $M \geq m/10.$ Let $\varepsilon > 0$  be such that $\varepsilon \geq 1/M^{0.25}.$ Let $\Delta' \geq \Delta$ be such that $\Delta'\leq m^{0.001}.$

Let $X^{(\Delta')}_{(\omega_1,j_1)},\ldots, X^{(\Delta')}_{(\omega_M,j_M)}$ be arbitrary distinct $X^{(\Delta)}$-segments. Let $X = \bigcup_{i\in [M]} X^{(\Delta')}_{(\omega_i,j_i)}, Y = \bigcup_{i\in [M]} Y^{(\Delta')}_{(\omega_i,j_i)},$ and $Z = \bigcup_{i\in [M]} Z^{(\Delta')}_{(\omega_i,j_i)}$.

Assume that $X' = \bigcup_{i\in [M]} X'_i$ where for each $i\in [M]$, $X'_i \subseteq X^{(\Delta')}_{(\omega_i,j_i)}$ satisfying $|X'_i| \geq \varepsilon\cdot |X^{(\Delta')}_{(\omega_i,j_i)}|.$ Let $F$ be any $(\Sigma\Pi)^{\Delta}\Sigma$ syntactially multilinear variable-labelled\footnote{See Section~\ref{sec:struct} for the definition of variable-labelled formulas.} formula with $\Vars(F) = X'\dot\cup Y'\dot\cup Z'$ where $Y'$ and $Z'$ are arbitrary sets of variables that are disjoint from $X'$. Assume that the size of $F$ is $s\leq \exp(k^{0.1}/\Delta^2).$

For each $i\in [M]$, let $\rho_i$ be an \emph{independent} multilinear random restriction obtained by using the sampling algorithm $\mc{A}$ described in Section~\ref{sec:restriction} to sample a $(X^{(\Delta')}_{(\omega_i,j_i)},Y^{(\Delta')}_{(\omega_i,j_i)},Z^{(\Delta')}_{(\omega_i,j_i)})$-restriction. Let $\rho = \rho_1\circ \cdots \circ \rho_M$ be the resulting $(X,Y,Z)$-random restriction. Let $\tilde{Y}' = \rho(X')\cap Y$ and $\tilde{Z}' = \rho(X')\cap Z$.\footnote{ Note that these are \emph{random} sets. Also note that for each fixing of $\rho$, the restricted formula $F|_\rho$ computes a multilinear polynomial in $\F[\tilde{Y}'\cup \tilde{Z}'\cup Y'\cup Z']$}

Then, we have
\[
\prob{\rho}{\relrk_{(Y'\cup \tilde{Y}',Z'\cup \tilde{Z}')}(F|_\rho) \geq \exp(-k^{0.1}/\Delta)} \leq \Delta \exp(-k^{0.1}/\Delta).
\]
\end{theorem}

\begin{remark}
\label{rem:technical}
Note that the $X^{(\Delta)}$-segments that appear in the statement of Theorem~\ref{thm:technical} are $X^{(\Delta)}$-segments contained in $X^{(\Delta')}$ for some $\Delta'\geq \Delta.$ This general form of the theorem is needed when we are proving lower bounds for multilinear formulas of product-depth $\Delta'.$ During the course of the inductive proof, the depth of the circuit reduces iteratively and at some intermediate point in the proof when the depth is down to $\Delta$, we end up in the situation described in Theorem~\ref{thm:technical}.
\end{remark}

We postpone the proof of the above theorem to Section~\ref{sec:technical}, and deduce Theorem~\ref{thm:PDeltalbd} from it below.

\begin{proof}[Proof of Theorem~\ref{thm:PDeltalbd}]
Let $C$ be any multilinear circuit of product-depth at most $\Delta$ computing $P^{(\Delta)}.$ Let $s$ denote the size of $C$. By Lemma~\ref{lem:RY-nf-ckts}, there is a $(\Sigma\Pi)^{\Delta}\Sigma$ syntactic multilinear formula $F$ of size $s' = s^{O(\Delta)}$ computing $P^{(\Delta)}.$ We can  assume that $F$ does not use any variables outside $X^{(\Delta)}$ since such variables may be safely set to $0$ without affecting the polynomial being computed.

Recall that $X^{(\Delta)}$ is the disjoint union of its segments $X^{(\Delta)}_1,\ldots,X^{(\Delta)}_m$. The random restriction $\rho^{(\Delta)}$ from Section~\ref{sec:restriction} is defined to be $\rho^{(\Delta)} = \rho_1^{(\Delta)} \circ \cdots\circ \rho_m^{(\Delta)}$ where each $\rho_i^{(\Delta)}$ is an independent $(X^{(\Delta)}_i,Y^{(\Delta)}_i,Z^{(\Delta)}_i)$-restriction sampled using the algorithm $\mc{A}.$ Let $\tilde{Y}$ and $\tilde{Z}$ denote $\Img(\rho^{(\Delta)})\cap (\bigcup_i Y_i^{(\Delta)})$ and $\Img(\rho^{(\Delta)})\cap (\bigcup_i Z_i^{(\Delta)})$ respectively.

Consider the restricted polynomials $P' = P^{(\Delta)}|_{\rho}$ and $F' = F|_{\rho}$. By Lemma~\ref{lem:fullrank}, we know that $\relrk_{(\tilde{Y},\tilde{Z})}(P') = 1$ with probability $1$. On the other hand, applying Theorem~\ref{thm:technical} with $M=m$, $\varepsilon = 1$, $\Delta'=\Delta = m^{o(1)}$ and $Y' = Z' = \emptyset,$ we get
\[
\prob{\rho}{\relrk_{(\tilde{Y},\tilde{Z})}(F') \geq \exp(-k^{0.1}/\Delta)}\leq \Delta \exp(-k^{-0.1}/\Delta) = o(1)
\]
if $s'\leq \exp(k^{0.1}/\Delta^2)$. The final inequality above follows from the fact that $\Delta \leq m^{o(1)}$ and hence $\Delta \exp(-k^{-0.1}/\Delta) = \Delta \exp(-m^{\Omega(1)}/\Delta) = o(1).$

Since the formula $F'$ computes $P'$, we must therefore have $s' > \exp(k^{0.1}/\Delta^2) = \exp(m^{\Omega(1)}/\Delta^2)$ which yields $s = \exp(m^{\Omega(1)}/\Delta^3)$. Since $\Delta = m^{o(1)},$ this means that $s = \exp(m^{\Omega(1)}) = \exp((n_\Delta)^{\Omega(1/\Delta)}).$
\end{proof}

\section{Proof of Theorem~\ref{thm:technical}}
\label{sec:technical}

The proof of Theorem~\ref{thm:technical} will be by induction on the product-depth $\Delta.$ The base case $\Delta = 1$ is handled in Section~\ref{sec:base-technical}.

For the induction case $\Delta > 1,$ we proceed as follows. Let $F$ be any formula as in the statement of Theorem~\ref{thm:technical}. We can write 
\begin{equation}
\label{eq:F-decompose}
 F = \sum_{i\in [s]} F_i
\end{equation}
where each $F_i$ is a $\Pi (\Sigma\Pi)^{\Delta - 1}\Sigma$ variable-labelled syntactic multilinear formula of size at most $s$. Further, note (see Section~\ref{sec:struct}) that $\Vars(F_i) = \Vars(F)$ for each $i$ by the property of the variable-labelling $\Vars(\cdot).$

We claim that it suffices to show the following. For any $\Pi (\Sigma\Pi)^{\Delta - 1}\Sigma$ variable-labelled syntactic multilinear formula $F'$ with $\Vars(F') = X'\cup Y'\cup Z'$ of size at most $s\leq \exp(k^{0.1}/\Delta^2),$ we have
\begin{equation}
\label{eq:relrkPigate}
\prob{\rho}{\relrk_{(\tilde{Y}'\cup Y',\tilde{Z}'\cup Z')}(F'|_\rho) \geq \exp(-k^{0.1}/(\Delta-1))}\leq \Delta \exp(-k^{0.1}/(\Delta-1)).
\end{equation}

Assuming (\ref{eq:relrkPigate}) for now, we proceed as follows. Applying (\ref{eq:relrkPigate}) to the formulas $F_1,\ldots,F_s$ and using a union bound, we have
\[
\prob{\rho}{\exists i\in [s],\ \relrk_{(\tilde{Y}'\cup Y',\tilde{Z}'\cup Z')}(F_{i}|_\rho) \geq \exp(-k^{0.1}/(\Delta-1))} \leq \Delta \exp(-k^{0.1}/(\Delta-1))\cdot s \leq \Delta \exp(-k^{0.1}/\Delta)
\]
where for the last inequality, we have used the fact that $s\leq \exp(k^{0.1}/\Delta^2).$ Note that if $\relrk_{(\tilde{Y}'\cup Y',\tilde{Z}'\cup Z')}(F_i|_\rho) < \exp(-k^{0.1}/(\Delta-1))$ for each $i\in [s]$, the subadditivity of relative rank (Propostion~\ref{prop:relrk} item 2) implies that
\[
\relrk_{(\tilde{Y}'\cup Y',\tilde{Z}'\cup Z')}(F|_\rho) \leq s\cdot \exp(-k^{0.1}/(\Delta-1)) \leq \exp(-k^{0.1}/\Delta),
\]
where for the last inequality we have again used $s\leq \exp(k^{0.1}/\Delta^2).$ We have thus shown that 
\[
\prob{\rho}{\relrk_{(\tilde{Y}'\cup Y',\tilde{Z}'\cup Z')}(F|_\rho) \geq \exp(-k^{0.1}/\Delta)} \leq \Delta \exp(-k^{0.1}/\Delta)
\]
which proves the theorem. 

It remains to prove (\ref{eq:relrkPigate}), which is the main technical part of the proof. Fix an $F'$ as in (\ref{eq:relrkPigate}) and assume that 
\begin{equation}
\label{eq:F'-decompose}
F' = \prod_{i\in [t]}F'_i
\end{equation}
where the $F'_i$ are the constituent $(\Sigma\Pi)^{(\Delta-1)}\Sigma$ formulas of $F'$. Recall that the $\Vars(F'_i)$ ($i\in [t]$) partition $\Vars(F') = X'\dot\cup Y'\dot\cup Z'.$ 

For $\delta > 0,$ $i\in [t]$ and $j\in [M]$, we say that $F'_i$ is \emph{$\delta$-heavy} in segment $X_j$ if $|\Vars(F'_i) \cap X_j|\geq \delta \cdot |X_j|.$ We say that $X_j$ is \emph{$\delta$-shattered} if there is no $i\in [t]$ such that $F'_i$ is $\delta$-heavy w.r.t. $X_j.$

The proof of (\ref{eq:relrkPigate}) breaks down into the following three cases. Roughly, in case $1$, we can prove (\ref{eq:relrkPigate}) directly, whereas in cases $2$ and $3$, we appeal to the inductive hypothesis.

\begin{itemize}
\item {\bf Case 1:} At least $M/2$ segments are $(1/4)$-shattered and further, no $F'_i$ is $(\varepsilon/k)$-heavy in at least $M/k$ segments $X_j$.
\item {\bf Case 2:} At least $M/2$ segments are \emph{not} $(1/4)$-shattered.
\item {\bf Case 3:} There is some $F'_i$ ($i\in [t]$) that is $(\varepsilon/k)$-heavy in at least $M/k$ many segments $X_j.$
\end{itemize}
We show (\ref{eq:relrkPigate}) in each of the above cases in Section~\ref{sec:cases1-4} below. This will conclude the proof of the theorem.

\paragraph{Notation.} We now define some notation that will be useful in the remainder of the theorem.
\begin{itemize}
\item For brevity, we use $X_i,Y_i,$ and $Z_i$ ($i\in [M]$) to denote the $X^{(\Delta)}$-segment $X^{(\Delta')}_{(\omega_i,j_i)}$, the $Y^{(\Delta)}$-segment $Y^{(\Delta')}_{(\omega_i,j_i)}$ and the $Z^{(\Delta)}$-segment $Z^{(\Delta')}_{(\omega_i,j_i)}$ respectively. The corresponding half-segments of $X_i,Y_i$ and $Z_i$ are denoted $X_{i,j}, Y_{i,j}$ and $Z_{i,j}$ respectively (for $j\in [2]$). Further $X^{(\Delta-1)}$-segments of $X_{i,j}$ are denoted $X_{i,j,p}$ ($p\in [m]$) and similarly we also have $Y_{i,j,p}$ and $Z_{i,j,p}.$
\item For each $i\in [M]$, let $y_i,z_i$ denote the variables $y^{(\Delta')}_{(\omega_i,j_i)}$ and $z^{(\Delta')}_{(\omega_i,j_i)}$ in the sets $Y_i$ and $Z_i$ respectively. Recall (see Section~\ref{sec:polynomial}) that we have $Y_i = Y_{i,1}\cup Y_{i,2} \cup \{y_i\}$ and $Z = Z_{i,1}\cup Z_{i,2} \cup \{z_i\}$.
\item For each $i\in [M]$, let $\tilde{Y}_i$ and $\tilde{Z}_i$ denote the random sets $\rho_i(X_i)\cap Y_i$ and $\rho_i(X_i)\cap Z_i$ respectively. Let $\tilde{Y}$ and $\tilde{Z}$ denote $\bigcup_{i\in [M]} \tilde{Y}_i$ and $\tilde{Z} = \bigcup_{i\in [M]}\tilde{Z}_i.$ Also, let $\tilde{Y}'_i$ and $\tilde{Z}'_i$ denote the random sets $\rho_i(X_i')\cap Y_i$ and $\rho_i(X_i')\cap Z_i$ respectively. 
\end{itemize}

\subsection{The base case of Theorem~\ref{thm:technical}}
\label{sec:base-technical}

Let $F$ be a $\Sigma \Pi \Sigma$ syntactic multilinear formula over the variable set $X^{'}\cup Y'\cup Z'$. We can write $F = T_1+\cdots + T_s$ where each $T_i$ is a $\Pi\Sigma$ syntactic multilinear formula. We claim that for any $\Pi\Sigma$ syntactic multilinear formula $T$, we have 
\begin{equation}
\label{eq:prodLi}
\prob{\rho}{\relrk_{(\tilde{Y}'\cup Y',\tilde{Z}'\cup Z')}(T|_{\rho})\geq \exp(-2k^{0.1})}\leq \exp(-2k^{0.1}).
\end{equation}

Assuming (\ref{eq:prodLi}) we are done, since for any fixing of the random restriction $\rho$ (which is a copy of the restriction $\rho^{(1)}$ from Section~\ref{sec:restriction}), we have by Proposition~\ref{prop:relrk} Item 2 that
\begin{align*}
\relrk_{(\tilde{Y}'\cup Y',\tilde{Z}'\cup Z')}(F|_{\rho}) &\leq \sum_{i=1}^s \relrk_{(\tilde{Y}'\cup Y',\tilde{Z}'\cup Z')}(T_i|_{\rho})\leq s\cdot \max_{i}\left(\relrk_{(\tilde{Y}'\cup Y',\tilde{Z}'\cup Z')}(T_i|_{\rho})\right)\\ 
&\leq \exp(k^{0.1})\cdot \max_{i}\left(\relrk_{(\tilde{Y}'\cup Y',\tilde{Z}'\cup Z')}(T_i|_{\rho})\right)
\end{align*}
and hence 
\begin{align*}
\prob{\rho}{\relrk_{(\tilde{Y}'\cup Y',\tilde{Z}'\cup Z')}(F|_{\rho}) \geq \exp(-k^{0.1})} & \leq  \prob{\rho}{\exists i\ \relrk_{(\tilde{Y}'\cup Y',\tilde{Z}'\cup Z')}(T_i|_{\rho})  \geq \exp(-2k^{0.1})} \\ & \leq \frac{s}{\exp(2k^{0.1})} \leq \exp(-k^{0.1}) 
\end{align*}
where we have used (\ref{eq:prodLi}) and a union bound for the second inequality.

It remains to prove (\ref{eq:prodLi}). To see this, we proceed as follows. Assume that the output $\times$-gate of $T$ is $\Phi$  and let $\Psi_1,\ldots,\Psi_r$  be the $+$-gates feeding into it and let $L_1,\ldots,L_r$ be the linear functions computed by these gates. Recall from Proposition~\ref{prop:vars} that we can assign to each $\Psi_i$ a variable set $\Vars(\Psi_i)\subseteq X^{'}\cup Y'\cup Z'$ such that $\{\Vars(\Psi_i)\ |\ i\in [r]\}$ induces a partition of $X^{'}\cup Y'\cup Z'$ and further $\Vars(\Psi_i)$ contains $\supp(\Psi_i)$.\footnote{$\supp(\Psi_i)$ is the set of variables that actually appear in the subformula rooted at $\Psi_i$ (cf. Section~\ref{sec:struct}).} 
Henceforth, we use $\Vars(L_i)$ to denote $\Vars(\Psi_i)$ for each $i\in [r]$. 

We divide the gates $\Psi_1,\ldots,\Psi_r$ into two classes: those gates such that $\Vars(L_i)\cap X^{'} \neq \emptyset$ and those for which $\Vars(L_i)\cap X^{'} = \emptyset$; let $t$ denote the number of gates of the former kind. Without loss of generality, we can assume that $\Psi_1,\ldots,\Psi_t$ are the gates such that $\Vars(L_i)\cap X^{'} \neq \emptyset$ and $\Psi_{t+1},\ldots,\Psi_r$  the rest. 

We write the polynomial $T$ as $T = L_1\cdots L_t\cdot Q'$ where $Q' = \prod_{i > t}L_i.$  
For each $i\in [t]$ and $j\in [M]$, let $\hat{Y}_{i,j}$ and $\hat{Z}_{i,j}$ denote the random sets $\rho(\Vars(L_i)\cap X'_j)\cap Y$ and $\rho(\Vars(L_i)\cap X'_j)\cap Z$ respectively. We also denote by $\hat{Y}_{i,0}$ and $\hat{Z}_{i,0}$ the (non-random) sets $\Vars(L_i)\cap Y'$ and $\Vars(L_i)\cap Z'.$  Let us define $\hat{Y}_{i} = \bigcup_{j\in [M]\cup \{0\}}\hat{Y}_{i,j}$ and $\hat{Z}_{i} = \bigcup_{j\in [M]\cup \{0\}}\hat{Z}_{i,j}$. Finally, let $\hat{Y}_{t+1}$ and $\hat{Z}_{t+1}$ denote the (non-random) sets $\bigcup_{i > t} \Vars(L_i)\cap Y'$ and $\bigcup_{i > t} \Vars(L_i)\cap Z'.$

Note that the sets $\tilde{Y}'\cup Y'$ can be partitioned as $Y'\cup \bigcup_{j\in [M]} \tilde{Y}'_j$ and also as $\bigcup_{i\in [t+1]}\hat{Y}_i = \left(\bigcup_{i\in [t], j\in [M] \cup \{0\}} \hat{Y}_{i,j}\right)\cup \hat{Y}_{t+1}.$ A similar statement is true for $\tilde{Z'}\cup Z'$ as well.

Upon the application of the random restriction $\rho$ (as described in Section~\ref{sec:restriction}), each $L_i$ ($i\in [t]$) restricts to a linear function $L_i|_{\rho} = \tilde{L}_i\in \F[\hat{Y}_i\cup \hat{Z}_i]$, while the polynomial $Q'$ is unaffected by the restriction $\rho$. 
Recall that $\rho$ is the composition of \emph{independent} restrictions $\rho_1,\ldots,\rho_{m}$ where each $\rho_j$ only 
affects variables in $X_j,$ for all $j\in[m]$ and the $\rho_j$'s are independent copies of the random restriction 
$\rho$ (defined in Section~\ref{sec:restriction}).

By Proposition~\ref{prop:relrk} items 3 and 1, we know that for any choice of $\rho$
\begin{align}
\relrk_{(Y'\cup \tilde{Y}',Z'\cup \tilde{Z}')}(T|_\rho) &= \left(\prod_{i\in [t]}\relrk_{(\hat{Y}_i, \hat{Z}_i )}(L_i|_{\rho})\right)\cdot \relrk_{(\hat{Y}_{t+1},\hat{Z}_{t+1})}(Q')\notag\\
& \leq \prod_{i\in [t]}\relrk_{(\hat{Y}_i, \hat{Z}_i )}(L_i|_{\rho}).\label{eq:relrkT}
\end{align}

To bound $\relrk_{(Y'\cup \tilde{Y}',Z'\cup \tilde{Z}')}(T|_\rho),$ we use (\ref{eq:relrkT}) along with a case analysis depending on the value of $t$. 

Before we do the case analysis, we will present two technical lemmas which will be useful. 

\begin{lemma}
\label{lem:size-tildey'}
For $\tilde{Y}', \tilde{Z}'$ defined as in the statement of Theorem~\ref{thm:technical}, $\prob{\rho}{|\tilde{Y}'|+ |\tilde{Z}'| \leq M/200} \leq \exp(-\Omega(M)).$
\end{lemma}

\begin{proof}
For $j\in [M]$, let $v_j$ be a $\{0,1\}$-valued random variable which takes value $1$ if and only if  $\tilde{Y}_j' \cup \tilde{Z}_j' \neq \emptyset$. This is a random variable which only depends on $\rho_j$ for each $j \in [M]$, i.e., $v_j$s are independent random variables. Let $v = \sum_j v_j$. Note that $v \leq |\tilde{Y}'| + |\tilde{Z}'|$. 

We first claim that for each $j$, $\prob{\rho_j}{\tilde{Y}_j' \cup \tilde{Z}_j' \neq \emptyset}$ is at least $1/48$. To see this, first note that $|X_j'|\geq \varepsilon |X_j| > 0$ (the non-emptiness of $X_j'\subseteq X_j$ is all we will need). Fix any variable $x \in X'_j$. From the description of the sampling algorithm $\mc{A}$ for the random restriction $\rho^{}$, it easily follows that the probability with which $\rho_j$ sets $x$ to a variable in $Y_j\cup Z_j$ is at least $1/48$.\footnote{W.l.o.g., say $x\in X_{j,1}$. Then $\mc{A}$ 
sets $x$ to a variable whenever option $E_3$ is chosen (this happens with probability $1/3$) and further $x$ is the variable 
chosen in $X_{j,1}$ to set to a variable (this happens with probability $1/16$).} Hence, $\avg{\rho_j}{v_j} \geq 1/48$. By linearity of expectation, we get that $\avg{\rho}{v} \geq M/48$. 

Using the Chernoff bound (Theorem~\ref{thm:Chernoff}) we get that $\prob{\rho}{v \leq M/200} \leq \exp(-\Omega(M))$. As $v \leq |\tilde{Y}'| + |\tilde{Z}'|$, we get that $\prob{\rho^{}}{|\tilde{Y}'|+ |\tilde{Z}'| \leq M/200} \leq \exp(-\Omega(M)).$
\end{proof}

\begin{lemma}
\label{lem:rank-loss}
For each $i \in [t]$, $\prob{\rho}{\relrk_{(\hat{Y}_i, \hat{Z}_i)}(L_i|_\rho) \leq 1/\sqrt{2}} \geq 1/1000$.
\end{lemma} 
\begin{proof}
Fix an $L_i$ for $i\in [t]$. Let us consider the smallest $j_0$ such that $X'_{j_0} \cap \Vars(L_i) \neq \emptyset$. We consider two possibilities for the relationship between $X_{j_0}$ and $\Vars(L_i)$: either $\Vars(L_i) \cap X'_{j_0} \subsetneq X_{j_0}$ or $X_{j_0} \subseteq \Vars(L_i)$. 

We fix $\rho_j$ for all $j \in [M]$ such that $j \neq j_0$. 
Let $\hat{Y}'_i = \cup_{j \neq j_0} \hat{Y}_{i,j}$ and $\hat{Z}'_i = \cup_{j \neq j_0} \hat{Z}_{i,j}$. 

\paragraph{\bf Case I:  $X_{j_0} \subseteq \Vars(L_i)$.} In this case, if the sampling step for $\rho_{j_0}$ chooses options $E_1$ or $E_2$, then from the definition of $\rho^{(0)}$ we can see that $|\hat{Y}_{i,j}| = |\rho(X_{j_0})\cap Y| = 2$ and similarly $|\hat{Z}_{i,j}|=2$ as well. This implies that $|\hat{Y}_i|,|\hat{Z}_i|\geq 2.$ 

On the other hand, we also claim that $\rank(M_{(\hat{Y}_i, \hat{Z}_i)})(L_i|_{\rho})) \leq 2.$ To see this, note that each $L_i|_{\rho}$ can be written as $L_i'(\hat{Y}_i) + L_i''(\hat{Z}_i)$ and for any polynomial $Q$ that depends only on the variables in either $\hat{Y}_i$ or $\hat{Z}_i$, we have $\rank(M_{(\hat{Y}_i,\hat{Z}_i)}(Q))\leq 1.$ The subadditivity of matrix rank now implies that $\rank(M_{(\hat{Y}_i,\hat{Z}_i)}(L_i|_\rho)) \leq 2.$

In particular, we see that whenever option $E_1$ or $E_2$ is chosen, we have
\[
\relrk_{(\hat{Y}_i,\hat{Z}_i)}(L_i|_{\rho}) \leq \frac{2}{2^{\frac{1}{2}(|\hat{Y}_i| + |\hat{Z}_i|)}}\leq \frac{1}{2}.
\]
Since one of $E_1$ or $E_2$ is chosen with probability $2/3$, we get the desired bound in this case.

\paragraph{\bf Case II: $\Vars(L_i) \cap X'_{j_0} \subsetneq X_{j_0}$.} 

As mentioned earlier, we already fixed the restrictions $\rho_j$ for $j \neq j_0$. Thus, the sets $\hat{Y}'_i$ and $ \hat{Z}'_i$ have been fixed. Let $b \in \{0,1\}$. Suppose $|\hat{Y}_i'\cup \hat{Z}'_i| \equiv b \mod{2}$. We will show that
\begin{equation}
\label{eq:odd-even}
 \prob{\rho_{j_0}}{|\hat{Y}_{i,j_0} \cup \hat{Z}_{i,j_0}| = (1-b) \mod 2} \geq 1/1000.
\end{equation}
Assuming (\ref{eq:odd-even}) for now, we get that with probability at least $1/1000$, $| \hat{Y}_i \cup  \hat{Z}_i| = |\hat{Y}_{i,j_0} \cup \hat{Z}_{i,j_0}| + |\hat{Y}_i'\cup \hat{Z}'_i|$ is odd. Therefore using Item 1 of Proposition~\ref{prop:relrk} we are done. This concludes the proof of Lemma~\ref{lem:rank-loss} assuming (\ref{eq:odd-even}). 

To see (\ref{eq:odd-even}), let us assume that $b=1$. (We only present the details for this case. The case of $b=0$ is similar.) If $b=1$, we need to show that $|\hat{Y}_{i,j_0} \cup \hat{Z}_{i,j_0}| = |\rho_{j_0}(\Vars(L_i)\cap X'_{j_0}) \cap (Y_{j_0}\cup Z_{j_0})|$ is even with probability $1/1000$.
The following three possibilities arise:

\begin{enumerate}
\item[(a)] $\Vars(L_i)\cap X_{j_0,1} \neq \emptyset$ and $\Vars(L_i)\cap X_{j_0,2} \neq \emptyset$, 
\item[(b)] $X_{j_0,1} \setminus \Vars(L_i) \neq \emptyset$ and $X_{j_0,2} \setminus \Vars(L_i) \neq \emptyset$,  or
\item[(c)] $\Vars(L_i)\cap X_{j_0} = X_{j_0,1}$ or $\Vars(L_i)\cap X_{j_0} = X_{j_0,2}$. 
\end{enumerate}

In each case we will show that $|\hat{Y}_{i,j_0}\cup \hat{Z}_{i,j_0}|$ is even with probability at least $1/1000$. For (a), if $\rho_{j_0}$ chooses option $E_3$, then with probability at least $1/8$ one of the variables in $\Vars(L_i)\cap X_{j_0,1}$ (resp.\ $\Vars(L_i)\cap X_{j_0,2}$) will be set to $y_{j_0}$ (resp.\ $z_{j_0}$). The option $E_3$ is chosen with probability $1/3$. Therefore, with probability at least $1/(8\cdot 8 \cdot 3)\geq 1/1000$,  $|\hat{Y}_{i,j_0} \cup \hat{Z}_{i,j_0}| = 2$ and  is thus even. 

For (b), observe that if $\rho_{j_0}$ chooses option $E_3$, then with probability at least $1/8$ one of the variables in $X_{j_0,1} \setminus \Vars(L_i)$ (resp.\ $X_{j_0,2} \setminus \Vars(L_i)$) will be set to $y_{j_0}$ (resp.\ $z_{j_0}$) and all other variables in $X_{j_0}$ are set to constants. As this implies that all variables in $\Vars(L_i)$ are set to constants, again with probability at least $1/(8\cdot 8 \cdot 3)$,  $|\hat{Y}_{i,j_0} \cup \hat{Z}_{i,j_0}|=0$ and is thus even. 

For (c), let us assume that  $\Vars(L_i)\cap X_{j_0} = X_{j_0,1}$. In this case, if $\rho_{j_0}$ chooses option $E_1$, then with probability $1$, $|\hat{Y}_{i,j_0} \cup \hat{Z}_{i,j_0}|=4$. The option $E_1$ is chosen with probability $1/3$. Similarly, if $\Vars(L_i)\cap X_{j_0} = X_{j_0,2}$ then if $\rho_{j_0}$ chooses option $E_2$, we have $|\hat{Y}_{i,j_0} \cup \hat{Z}_{i,j_0}|=4$. Again, option $E_2$ is chosen with probability $1/3$. 

This finishes the proof of (\ref{eq:odd-even}) in the case that $b=1$. For the case $b=0$, we employ a similar case analysis with the following cases. 
\begin{enumerate}
\item[(a)] $\Vars(L_i)\cap X_{j_0,1} \neq \emptyset$ and $X_{j_0,2} \setminus \Vars(L_i) \neq \emptyset$,  
\item[(b)] $X_{j_0,1} \setminus \Vars(L_i) \neq \emptyset$ and $\Vars(L_i)\cap X_{j_0,2} \neq \emptyset$.
\end{enumerate}
The proof is similar and left to the reader.
\end{proof}

We now proceed to the proof of (\ref{eq:prodLi}), which will finish the proof of the base case of Theorem~\ref{thm:technical}.

\vspace*{5pt}
\noindent
Let us recall that $t$ is the number of $+$-gates such that for all $i\in[t]$, $\Vars(L_{i})\cap X' \neq \emptyset.$
\vspace*{5pt}
\noindent
\underline{$t\leq M/500$ (Case 1):} This case is quite straightforward. As noted in the proof of Case I of Lemma~\ref{lem:rank-loss}, $\rank(M_{(\hat{Y}_i, \hat{Z}_i)}(L_i|_{\rho})) \leq 2.$

In particular, the above yields 
\[
\relrk_{(\hat{Y}_i, \hat{Z}_i)}(L_i|_{\rho}) \leq \frac{2}{2^{\frac{1}{2}(|\hat{Y}_i|+ |\hat{Z}_i|)}}
\]
and hence using (\ref{eq:relrkT}), we have
\[
\relrk_{(Y'\cup \tilde{Y}',Z'\cup \tilde{Z}')}(T|_\rho) \leq \frac{2^t}{2^{\frac{1}{2}(|\tilde{Y}'\cup Y'| + |\tilde{Z}'\cup Z'|)}}\leq \frac{2^t}{2^{\frac{1}{2}(|\tilde{Y}'| + |\tilde{Z}'|)}}. 
\]
Using Lemma~\ref{lem:size-tildey'}, we know that $|\tilde{Y}'|+ |\tilde{Z}'| \leq M/200$ with probability at most $\exp(-\Omega(M))$ . Therefore, with all but $\exp(-\Omega(M))$ probability, $$\relrk_{(Y'\cup \tilde{Y}',Z'\cup \tilde{Z}')}(T) \leq \frac{2^{M/500}}{2^{M/400}} \leq \exp(-2k^{0.1}).$$  Moreover, for our setting of parameters of $M$ and $k$, $\exp(-\Omega(M)) \leq \exp(-2k^{0.1})$. 
This proves (\ref{eq:prodLi}) in the case where $t\leq M/500.$

\vspace*{5pt}
\noindent
\underline{$t > M/500$ (Case 2):} 
For each $i\in [t]$,  we define a Bernoulli random variable $V_i$ as follows: $V_i=1$ if $\relrk_{(\hat{Y}_i, \hat{Z}_i)}(L_i|_{\rho}) \leq 1/\sqrt{2}$ and it is $0$ otherwise. By Lemma~\ref{lem:rank-loss}, we know that $\avg{}{V_i} \geq 1/1000$.  Let $V  = \sum_{i \in [t]} V_i$. Then by linearity of expectation, $\avg{}{V} \geq t/1000$. 

Recall that $|X_j| = 16$ for every $j \in [M]$. The multilinearity of the formula $T$ implies that for any $i \neq i'$, $\Vars(L_i) \cap \Vars(L_{i'}) = \emptyset$.  Therefore, for each segment $X_j$ the number of $L_i$ such that $\Vars(L_i)\cap X_j\neq \emptyset$ is at most $16$. This implies that the Bernoulli random variables $V_i$s are read-$16$, when viewed as functions of the independent random restrictions $\rho_1,\ldots,\rho_M.$ 
By the read-$k$ Chernoff bound (Theorem~\ref{thm:GLSS}),\footnote{The use of the read-$k$ Chernoff bound here can be easily circumvented by showing that in fact a constant-fraction of the $V_i$ ($i\in [t]$) are in fact completely independent and the standard Chernoff bound can be applied to their sum. We leave the details to the interested reader.} we thus obtain
\[
\prob{}{ V < t/2000 } \leq \exp(-\Omega(t)) \leq \exp(-\Omega(M))\leq  \exp(-2k^{0.1}).
\]
When $V \geq t/2000$, i.e. when at least  $t/2000$ many $L_i$ are such that $\relrk_{(\hat{Y}_i, \hat{Z}_i)}(L_i|_{\rho}) \leq 1/\sqrt{2}$, we have by (\ref{eq:relrkT}),
\begin{align*}
\relrk_{(Y'\cup \tilde{Y},Z'\cup \tilde{Z})}(T) &\leq \prod_{i\in [t]}\relrk_{(\hat{Y}_i,\hat{Z}_i)}(L_i|_{\rho}) \leq \frac{1}{(\sqrt{2})^{ t/2000}} \leq \exp(-2k^{0.1}).
\end{align*}
This proves (\ref{eq:prodLi}) in the case where $t > M/500.$ This completes the proof of the base case.

\subsection{The Inductive Case: Proof of (\ref{eq:relrkPigate}) in Cases 1-3}
\label{sec:cases1-4}
The following piece of notation will be useful for the inductive case. 

For each $i\in [t]$ and $j\in [M]$, let $\hat{Y}_{i,j}$ and $\hat{Z}_{i,j}$ denote the random sets $\rho(\Vars(F'_i)\cap X'_j)\cap Y$ and $\rho(\Vars(F'_i)\cap X'_j)\cap Z$ respectively. We also denote by $\hat{Y}_{i,0}$ and $\hat{Z}_{i,0}$ the (non-random) sets $\Vars(F'_i)\cap Y'$ and $\Vars(F'_i)\cap Z'.$  Let us define $\hat{Y}_{i} = \bigcup_{j\in [M]\cup \{0\}}\hat{Y}_{i,j}$ and $\hat{Z}_{i} = \bigcup_{j\in [M]\cup \{0\}}\hat{Z}_{i,j}$.

Note that the sets $\tilde{Y}'\cup Y'$ can be partitioned as $Y'\cup \bigcup_{j\in [M]} \tilde{Y}'_j$ and also as $\bigcup_{i\in [t]}\hat{Y}_i = \bigcup_{i\in [t], j\in  [M]\cup \{0\}} \hat{Y}_{i,j}.$ A similar statement is true for $\tilde{Z'}\cup Z'$ as well.

\subsubsection{Case 1}
\label{sec:case1}

By renaming the segments if necessary, we assume that $X_1,\ldots,X_{M/2}$ are $(1/4)$-shattered. For each $j \leq M/2$, the restriction $\rho_j$ is sampled according to the algorithm $\mc{A}$ described in Section~\ref{sec:restriction}. We say that a $\rho_j$ is \emph{good} if the algorithm $\mc{A}$ decides on option $E_3$ in sampling $\rho_j.$ For each $j$, $\rho_j$ is good with probability $1/3$. Let $\mc{E}_1$ denote the event that the number of good $\rho_j$ is at most $M/8$. Since different $\rho_j$'s are sampled independently, a Chernoff bound (Theorem~\ref{thm:Chernoff} Item 1) tells us that 
\begin{equation}
\label{eq:case1chernoff}
\prob{\rho}{\mc{E}_1}\leq \exp(-\Omega(M)).
\end{equation}

For each $j\in [M/2],$ we condition on whether or not $\rho_j$ is good. By (\ref{eq:case1chernoff}), the probability that $\mc{E}_1$ occurs is $\exp(-\Omega(M)).$

Assume that the event $\mc{E}_1$ does not occur. By renaming segments once more, we may assume that $\rho_j$ is good for each $j\in [M/8]$. We condition on any choice of the restrictions $\rho_j$ for $j > M/8$; this fixes the sets $\tilde{Y}'_j$ and $\tilde{Z}'_j$ defined above for $j > M/8$. 

We now observe the following from the description of the sampling algorithm $\mc{A}$, specifically option $E_3$ of $\mc{A}$. Conditioned on our choices so far, each $\rho_j$ for $j\in [M/8]$ is now a random $(X_j,\{y_j\},\{z_j\})$-restriction. For clarity, we call this conditioned random restriction $\rho_j'$.

We now proceed to analyzing $\relrk_{(\tilde{Y}'\cup Y',\tilde{Z}'\cup Z')}(F'|_\rho)$. Since $\tilde{Y}'\cup Y'$ (resp.\ $\tilde{Z}'\cup Z'$) can be partitioned as $\bigcup_{i\in [t]} \hat{Y}_i$ (resp.\ $\bigcup_{i\in [t]}\hat{Z}_i$), we know (Proposition~\ref{prop:relrk} Item 3) that for any choice of $\rho,$
\begin{equation}
\label{eq:case1relrkF'}
\relrk_{(\tilde{Y}'\cup Y',\tilde{Z}'\cup Z')}(F'|_\rho) = \prod_{i\in [t]}\relrk_{(\hat{Y}_i,\hat{Z}_i)}(F'_i|_\rho).
\end{equation}

So we analyze $\relrk_{(\hat{Y}_i,\hat{Z}_i)}(F'_i|_\rho)$ for each $i\in [t]$. To bound this quantity, we recall (Proposition~\ref{prop:relrk} Item 1) that $\relrk_{(\hat{Y}_i,\hat{Z}_i)}(F'_i|_\rho)$ is always at most $1$. Further, if $|\hat{Y}_i\cup \hat{Z}_i|$ is odd, then $\relrk_{(\hat{Y}_i,\hat{Z}_i)}(F'_i|_\rho)\leq 1/\sqrt{2}.$ This motivates what follows.

For any $i\in [t]$ and $j\in \{0,\ldots,M\}$, let $\alpha_{i,j}\in \{0,1\}$ be the random variable defined by $\alpha_{i,j} \equiv |\hat{Y}_{i,j}\cup \hat{Z}_{i,j}| \pmod{2}.$ Note that the variables $\alpha_{i,0}$ and $\alpha_{i,j}$ for $j > M/8$ are actually fixed. Define $\tilde{\alpha}_i\in \{0,1\}$ by $\tilde{\alpha}_i = \bigoplus_{j=0}^M \alpha_{i,j}.$ Observe that $\tilde{\alpha}_i = 1$ if and only if $|\hat{Y}_i\cup \hat{Z}_i|$ is odd. In particular, by Proposition~\ref{prop:relrk} (item 3), we have $\relrk_{(\hat{Y}_i,\hat{Z}_i)}(F'_i|_\rho) \leq 1/2^{\tilde{\alpha}_i/2}.$

Let $\mathrm{Odd}_\rho$ be the integer random variable defined by 
\[
\mathrm{Odd}_\rho = \sum_{i\in [t]} \tilde{\alpha}_i
\]
where the sum is defined over $\mathbb{R}.$ By the discussion above and (\ref{eq:case1relrkF'}), we know that 
\begin{equation}
\label{eq:case1oddrelrk}
\relrk_{(\tilde{Y}'\cup Y',\tilde{Z}'\cup Z')}(F'|_\rho) \leq \frac{1}{2^{\mathrm{Odd}_\rho/2}}.
\end{equation}

We will show below that 
\begin{equation}
\label{eq:case1oddclm}
\prob{\rho'_1,\ldots,\rho'_{M/8}}{\mathrm{Odd}_{\rho} \leq 10 k^{0.1}}\leq \exp(-k^{0.1}).
\end{equation}

The above, along with (\ref{eq:case1oddrelrk}) implies that whenever the event $\mc{E}_1$ does not occur, we have the inequality
\[
\prob{\rho}{\relrk_{(\tilde{Y}'\cup Y',\tilde{Z}'\cup Z')}(F'|_\rho) \geq \exp(-k^{0.1})} \leq \exp(-k^{0.1}).
\]

In particular, using (\ref{eq:case1chernoff}), we see that
\[
\prob{\rho}{\relrk_{(\tilde{Y}'\cup Y',\tilde{Z}'\cup Z')}(F'|_\rho) \geq \exp(-k^{0.1})} \leq \exp(-k^{0.1}) + \exp(-\Omega(M)) \leq 2\exp(-k^{0.1})
\]
which implies (\ref{eq:relrkPigate}) and hence finishes the analysis of Case $1$. For the last inequality above, we have used the fact $M \geq m/10\geq k$ for our choice of parameters.

We now prove (\ref{eq:case1oddclm}). To do this, we will actually prove a slightly different statement. For any $i\in [t]$, define $\alpha_i = \bigoplus_{j\in [M/8]}\alpha_{i,j}$. Note that $\alpha_i = \tilde{\alpha}_i \oplus \bigoplus_{j\not\in [M/8]}\alpha_{i,j}.$ For any $\beta_1,\ldots,\beta_t\in \{0,1\}$, we will show that 
\begin{equation}
\label{eq:case1alphabeta}
\prob{\rho_1',\ldots,\rho'_{M/8}}{\forall i\in [t]\ \alpha_i = \beta_i} \leq \exp(-k^{0.15}).
\end{equation}

Assuming the above, we obtain a similar statement for $\tilde{\alpha}_1,\ldots,\tilde{\alpha}_t$, since for any $\beta_1,\ldots,\beta_t\in \{0,1\}$,
\begin{align*}
\prob{\rho_1',\ldots,\rho'_{M/8}}{\forall i\in [t]\ \tilde{\alpha}_i = \beta_i} &= \prob{\rho_1',\ldots,\rho'_{M/8}}{\forall i\in [t]\ \alpha_i = \beta_i\oplus \bigoplus_{j\not\in [M/8]}\alpha_{i,j}} \leq \exp(-k^{0.15}).
\end{align*}
(We have used above the fact that $\alpha_{i,j}$ is fixed for any $j\in \{0,\ldots,M\}\setminus [M/8]$.)

From the above inequality and using Proposition~\ref{prop:ANDthr}, we get
\begin{align*}
\prob{\rho'_1,\ldots,\rho'_{M/8}}{\mathrm{Odd}_{\rho} \leq 10 k^{0.1}} &\leq \exp(10k^{0.1}\log t - k^{0.15})\\
&\leq \exp(10k^{0.1} k^{0.02+o(1)} - k^{0.15}) \leq \exp(-k^{0.1})
\end{align*}
where for the second inequality, we used the fact that $t\leq n_{\Delta'} = O(m)^{\Delta'}$ and hence $\log t = O(\Delta' \log m )\leq m^{0.001 + o(1)} = k^{0.02 + o(1)}$ as $\Delta' \leq m^{0.001}$ and $m = k^{20}.$ This finishes the proof of (\ref{eq:case1oddclm}) modulo (\ref{eq:case1alphabeta}).

The proof of (\ref{eq:case1alphabeta}) is the main technical statement of this section. From now on, fix some $\beta_1,\ldots,\beta_t\in \{0,1\}.$   

\paragraph{Proof outline of (\ref{eq:case1alphabeta}).}
From the description of the sampling algorithm $\mc{A}$ (specifically option $E_3$ of $\mc{A}$), we observe the following. To sample each $\rho_j'$, we choose independent random variables $x_{j,1}$ and $x_{j,2}$ uniformly from $X_{j,1}$ and $X_{j,2}$ respectively, and set $x_{j,1}$ to $y_j$ and $x_{j,2}$ to $z_j$; all other variables in $X_{j}$ are set deterministically to $0$ or $1$ according to some rule $R$ (the exact rule $R$ will not be relevant in this argument). We will view this sampling process iteratively via an algorithm $\mc{A}'$ described formally below. 

Informally, at each step, $\mc{A}'$ chooses a suitable $I\subseteq [t]$ and asks for each $j\in [M/8]$ if either of the variables $x_{j,1}$ or $x_{j,2}$ belongs to the set of variables $\bigcup_{i\in I}\Vars(F'_i)\cap X'_j$. If this event does occur for any $j\in [M/8]$, then $\mc{A}'$ reveals the restriction $\rho'_j$ entirely, and otherwise, it does not reveal anything else about $\rho'_j$ in this step. In either case, we are able to entirely deduce the values of $\alpha_{i,j}$ for $i\in I$ from the above information about $\rho_j'$ and hence we can also deduce $\alpha_i = \bigoplus_{j\in [M/8]} \alpha_{i,j}$ for each $i\in I$. We show that with reasonable probability, it holds that $\bigoplus_{i\in I} \alpha_i \neq \bigoplus_{i\in I} \beta_i,$ which in particular implies that there must be an $i$ such that $\alpha_i\neq \beta_i.$ In this case, $\mc{A}'$ outputs SUCCESS. If not, the algorithm continues with another iteration of the same procedure. We will show that with high probability, $\mc{A}'$ can carry out many such iterations. If either $\mc{A}'$ does not find an $i$ such that $\alpha_i \neq \beta_i$ after many iterations, or it cannot carry out too many iterations, then $\mc{A}'$ outputs FAILURE. We show that $\mc{A}'$ outputs SUCCESS with high probability, which will finish the proof of (\ref{eq:case1alphabeta}).

\paragraph{The algorithm $\mc{A}'$.} The algorithm $\mc{A}'$ is a general sampling procedure that has the following input-output behaviour.
\begin{itemize}
\item {\bf Input}: Sets $S\subseteq [M/8], T\subseteq [t]$, and fixed (i.e. not random) $(X_j,\{y_j\},\{z_j\})$-restrictions $(\rho'_j)_{j\in [M/8]\setminus S}.$ For each $j\in S$, we define $\bar{X}_j = (\bigcup_{i\in T}\Vars(F'_i))\cap X_j.$
\item {\bf Desired Output}: The algorithm samples an independent random restriction $\rho'_j$ for each $j\in S$ as follows. Independent and uniformly random variables $x_{j,1}$ and $x_{j,2}$ are chosen from sets $\bar{X}_{j,1}:= \bar{X}_j\cap X_{j,1}$ and $\bar{X}_{j,2}:= \bar{X}_j \cap X_{j,2}$ respectively. The variables $x_{j,1}$ and $x_{j,2}$ are set to $y_j$ and $z_j$ respectively, and the remaining variables in $X_j$ are set deterministically to $0$ or $1$ in accordance with the rule $R$ referenced above. 

Further, the algorithm either outputs SUCCESS or FAILURE, with the guarantee that when it outputs SUCCESS, then we must have $\alpha_i\neq \beta_i$ for some $i\in [t]$. 
\end{itemize}

Note that the problem of sampling $\rho'_1,\ldots, \rho'_{M/8}$ is equivalent to running the algorithm $\mc{A}'$ with $S = [M/8]$ and $T = [t]$. In this case, $\bar{X}_j = X_j$ for each $j\in [M/8]$.

The formal description of the algorithm $\mc{A}'$ follows. 

\paragraph{Algorithm $\mc{A}'(S,T,(\rho'_j)_{j\in [M/8]\setminus S})$}
\begin{enumerate}
\item Find $I\subseteq T$ such that the following holds. Let $F'_I := \prod_{i\in I}F'_i$ and $\Vars(F'_I) := \bigcup_{i\in I}\Vars(F'_i).$ 
\begin{enumerate}
\item For all $j\in S$, $|\Vars(F'_I)\cap \bar{X}_j|\leq (1/3)\cdot |\bar{X}_j|.$ 
\item $|\hat{S}| \leq M/k$, where $\hat{S} := \{j \in S \mid |\Vars(F'_I)\cap \bar{X}_j|> (2\varepsilon/k)\cdot |\bar{X}_j|\}$. 
\item There is some $j\in S$ such that $|\Vars(F'_I)\cap \bar{X}_j|\geq (\varepsilon/2k)\cdot |\bar{X}_j|.$
\end{enumerate}
If there are more than one such $I$, choose the lexicographically least one. If there is no such $I$ or if $S = \emptyset$, output FAILURE. Further, complete the sampling process as follows. For each $j\in S$, the restriction $\rho'_j$ is sampled by choosing variables $x_{j,1}$ and $x_{j,2}$ independently and uniformly from $\bar{X}_{j,1}$ and $\bar{X}_{j,2}$ respectively and setting them to $y_j$ and $z_j$ respectively. Other variables in $X_j$ are set according to the restriction rule $R$.
\item Initialize the set $S' = \emptyset .$
\item For each $j\in S$, define $\delta_{j,1}\in [0,1]$ by
\[
\delta_{j,1} = \frac{1}{|\bar{X}_{j,1}|}\cdot |\Vars(F'_I)\cap \bar{X}_{j,1}|
\]
and $\delta_{j,2}$ similarly. Let $\delta_j = |\Vars(F'_I)\cap \bar{X}_{j}|/|\bar{X}_j| $ and define $\delta = \sum_{j\in S}\delta_j.$
\item If $\delta < 1$, define $S'' = \{j\in S\ |\ \delta_j \geq 1/M^{0.5}\}$. If $\delta \geq 1,$ define $S'' = \hat{S}$ where $\hat{S}$ is as defined in Step 1 above, i.e., $S'' = \{j\in S\ | \delta_j \geq (2\varepsilon/k)\}.$
\item For each $j\in S$, do the following independently.
\begin{enumerate}
\item Sample $b_{j,1}\in \{0,1\}$ so that $b_{j,1}=1$ w.p. $\delta_{j,1}$ and $0$ otherwise. Similarly, sample $b_{j,2}$ independent of $b_{j,1}$ so that $b_{j,2}=1$ w.p. $\delta_{j,2}.$
\item If either of $b_{j,1}$ or $b_{j,2}$ is $1$, or $j\in S''$, do the following.
\begin{enumerate}
\item Add $j$ to the set $S'$.
\item If $b_{j,1} = 1$, sample $x_{j,1}$ uniformly from $\Vars(F'_I)\cap \bar{X}_{j,1}$ and set $x_{j,1}$ to $y_j$. If $b_{j,1}=0,$ sample $x_{j,1}$ uniformly from $\bar{X}_{j,1}\setminus \Vars(F'_I)$ and set $x_{j,1}$ to $y_j$.
\item Sample $x_{j,2}$ similarly from $\bar{X}_{j,2}$ and set $x_{j,2}$ to $z_j$. 
\item Set all other variables in $X_j$ deterministically in accordance with the rule $R$. (This fixes the $(X_j,\{y_j\},\{z_j\})$-restriction $\rho'_j$.)
\end{enumerate}
\end{enumerate}
\item Compute the Boolean variables $\alpha_{i,j}'\in \{0,1\}$ for each $i\in I$ and $j\in [M/8]$ as follows. 
\begin{enumerate}
\item If $j\not\in S$ or $j\in S'$, compute $\alpha_{i,j}'$ from $\rho'_j$ using $\alpha'_{i,j} = |\hat{Y}_{i,j}\cup \hat{Z}_{i,j}| \pmod{2}.$
\item If $j\in S\setminus S'$, then set $\alpha'_{i,j} = 0.$
\end{enumerate}
\item For each $i\in I$, let $\alpha'_i = \bigoplus_{j\in [M/8]}\alpha'_{i,j}.$ If $\alpha'_i \neq \beta_i$ for any $i\in I$, sample the remaining $\rho'_j$ ($j\in S\setminus S'$) as in Steps 5(b)(ii) and 5(b)(iii), and output SUCCESS.
\item Otherwise, run the algorithm $\mc{A}'$ on inputs $(S\setminus S', T\setminus I, (\rho'_j)_{j\not \in (S\setminus S')}).$
\end{enumerate}

\paragraph{Correctness.} Here, we show that $\mc{A}'$, on any input $S,T,$ and $(\rho'_j)_{j\not\in S}$ samples $(\rho'_j)_{j\in S}$ according to the desired distribution. Moreover, we show that whenever $\mc{A}'$ outputs SUCCESS, it is indeed because the sampled restrictions imply that $\alpha_i \neq \beta_i$ for some $i\in [t]$.

We first argue that the sampled distribution is correct. Note that if the algorithm cannot find a suitable $I$ in Step 1, then the restrictions sampled trivially have the correct distribution. So we assume that $\mc{A}$ does not output FAILURE in Step 1. 

Now, note that to sample a uniformly random variable $a$ from a finite set $A$, we may first fix any set $A'\subseteq A$ and sample a random bit $\beta\in \{0,1\}$ that is $1$ with probability  $|A'|/|A|$ and depending on whether $\beta$ is $1$ or $0$, sample a random element of $A'$ or $A\setminus A'.$ This describes how the sampling algorithm $\mc{A}'$ samples a random $x_{j,1}$ from $\bar{X}_{j,1}$ (the case of $\bar{X}_{j,2}$ is similar), where the role of the set $A'$ is taken by $\bar{X}_{j,1}\cap\Vars(F'_I)$ and the bit $b_{j,1}$ plays the role of $\beta$. If $j\in S'$, then the subsequent sampling takes place in Step 5(b), and for $j\in S\setminus S'$, the subsequent sampling takes place in either Step 7 (in case the algorithm outputs SUCCESS) or in a later iteration.

We now argue that an output of SUCCESS means that the algorithm has found an $i$ such that $\alpha_i\neq \beta_i.$ This is obvious once we argue that for each $i\in I$, the quantity $\alpha'_i$ computed by the algorithm is equal to the random variable $\alpha_{i}$ defined above. To argue this, it suffices to show that $\alpha'_{i,j}$ (computed in Step 6) equals $\alpha_{i,j}$ for each $i\in I,j\in [M/8]$. This is obvious for $j\not\in S$ or $j\in S'$ from Step 6 and the definition of $\alpha_{i,j}$ above. For $j\in S\setminus S'$, we know that $b_{j,1} = b_{j,2} = 0$ and hence the sampled $\rho'_j$ does not choose any variable from $\Vars(F'_I)$ to be either $x_{j,1}$ or $x_{j,2}$. In particular, this implies that  $x_{j,1}\not\in \Vars(F'_i)$ and $x_{j,2}\not\in \Vars(F'_i)$ for any $i\in I$ and hence, $\alpha_{i,j}=0.$ Thus, even in the case that $j\in S\setminus S'$, we have $\alpha_{i,j} = \alpha'_{i,j}.$ This concludes the proof of correctness.

\paragraph{Probability of SUCCESS.} We now argue that the algorithm $\mc{A}'$ outputs SUCCESS with high probability if started with the initial input of $S = S_0 := [M/8]$ and $T = T_0 := [t]$ (in which case $\bar{X}_j = X_j$ for each $j$ and the input $(\rho'_j)_{j\not\in S}$ is trivial). This will prove (\ref{eq:case1alphabeta}). 

Consider a run of the algorithm $\mc{A}'$ on the above inputs. Each such run can produce many successive iterations on different inputs $(S,T,(\rho'_j)_{j\not\in S})$. We say that a single iteration of $\mc{A}'$ is of Type I if the corresponding value of $\delta$ (computed in Step 3) is at least $1$ and of Type II otherwise. Note that this is a \emph{deterministic} function of the current inputs $S$ and $T$.

To show that the algorithm $\mc{A}'$ outputs SUCCESS with high probability on $(S_0,T_0)$, we show a more general statement. Call an input $(S,T,(\rho'_j)_{j\not\in S})$ an \emph{$(a,b)$-good} input, where $a,b\in \mathbb{N}$, if the following holds.
\begin{itemize}
\item The input set $S$ satisfies
\begin{equation}
\label{eq:case1S}
|S| \geq \frac{M}{8} - \frac{10\cdot a\cdot M}{k}- 10\cdot b\cdot M^{0.5},
\end{equation}
\item and for each $j\in S$, we have 
\begin{equation}
\label{eq:case1Xjbar}
|\bar{X}_j| \geq |X_j|\cdot \left(1-\frac{2\varepsilon}{k}\cdot a - \frac{b}{M^{0.5}}\right).
\end{equation}
\end{itemize}

We say an iteration of $\mc{A}'$ is $(a,b)$-good, if its input is $(a,b)$-good. Informally, we expect an iteration of $\mc{A}'$ to be $(a,b)$-good after at most $a$ iterations of Type I and at most $b$ iterations of Type II (see Claim~\ref{clm:case1gooditerns} below). 

The first observation is that an algorithm cannot output FAILURE on an $(a,b)$-good iteration for small $a,b$. 

Define $a_0 = k/100$ and $b_0 = M^{0.45}.$

\begin{claim}
\label{clm:case1failureiter}
The algorithm $\mc{A}'$ does not output FAILURE in any $(a,b)$-good iteration where $a < a_0$ and $b  < b_0.$
\end{claim}
\begin{proof}
We only need to show that $S\neq \emptyset$ and that $\mc{A}'$ is able to find an $I\subseteq T$ with the required properties in Step 1. 

Observe that in an $(a,b)$-good iteration for $a < a_0$ and $b < b_0,$ we have by (\ref{eq:case1S})
\begin{equation}
\label{eq:case1lbdS}
|S| \geq \frac{M}{8} - \frac{10\cdot (k/100) \cdot M}{k} - 10\cdot M^{0.45}\cdot M^{0.5} \geq \frac{M}{50}
\end{equation}
and also for each $j\in S$, by (\ref{eq:case1Xjbar})
\begin{equation}
\label{eq:case1lbdXj}
|\bar{X}_j|\geq |X_j|\cdot \left(1-\frac{2\varepsilon}{k}\cdot \frac{k}{100} - \frac{M^{0.45}}{M^{0.5}}\right) \geq  (1-o(1))|X_j|.
\end{equation}

First consider the case that there is some $i_0\in T$ and some $j_0\in S$ such that $|\Vars(F'_{i_0})\cap X_{j_0}|\geq (\varepsilon/2k)|X_{j_0}|.$ In this case, we claim that the singleton set $I = \{{i_0}\}$ has the properties required in Step 1 of $\mc{A}'.$ Property (a) follows from the fact that each $X_j$ ($j\in [M/8]$) is $(1/4)$-shattered and hence for each $j\in [M/8]$
\[
|\Vars(F'_{i_0})\cap \bar{X}_j| = |\Vars(F'_{i_0})\cap X_j|\leq |X_j|/4 \leq |\bar{X}_j|/3
\]
 where the last inequality uses (\ref{eq:case1lbdXj}). Property (c) follows from the choice of $i_0$, which implies 
 \[
 |\Vars(F'_{i_0})\cap \bar{X}_{j_0}| = |\Vars(F'_{i_0})\cap X_{j_0}|\geq (\varepsilon/2k)\cdot|X_{j_0}| \geq (\varepsilon/2k)\cdot |\bar{X}_{j_0}|.
 \]
 Finally, to see Property (b), we note that if $|\Vars(F'_{i_0})\cap \bar{X}_j| > (2\varepsilon/k)|\bar{X}_j|$ for some $j\in S$, then by (\ref{eq:case1lbdXj}), we also have 
 \[
|\Vars(F'_{i_0})\cap \bar{X}_j| = (2\varepsilon/k)|\bar{X}_j|  > (\varepsilon/k)|X_j|.
 \]
 In other words, $F'_{i_0}$ is $(\varepsilon/k)$-heavy in $X_j$. However, the assumption in this section (corresponding to Case 1) is that each $F'_i$ is $(\varepsilon/k)$-heavy in at most $M/k$ many $X_j$. In particular, this shows that the  number of $j$'s such that $|\Vars(F'_{i_0})\cap \bar{X}_j| > (2\varepsilon/k)|\bar{X}_j|$ is at most $M/k$ and hence Property (b) is indeed satisfied. Thus, the singleton set $I = \{i_0\}$ has all the required properties.
 
 From now onwards, we assume that there is no $i_0\in T$ and $j_0\in S$ satisfying the above. Thus, we have for each $i\in T$ and $j\in S$, 
 \begin{equation}
 \label{eq:case1Xjshatt}
 |\Vars(F'_i)\cap X_j| < (\varepsilon/2k)|X_j|.
 \end{equation}

From (\ref{eq:case1Xjshatt}), we know in particular that for each $i\in T$, we have $|\Vars(F'_i)\cap \bar{X}_j| = |\Vars(F'_i)\cap X_j|\leq (\varepsilon/2k)|X_j| \leq (2\varepsilon/k)|\bar{X}_j|.$ Thus, the set $\bar{X}_j$ is partitioned by the sets $(\Vars(F'_i)\cap \bar{X}_j)_{i\in T},$ each of relative size at most $(2\varepsilon/k)$.

Now, consider all sets $I\subseteq T$ such that for \emph{some} $j\in S$, we have 
\begin{equation}
\label{eq:case1I0def}
|\Vars(F'_I)\cap \bar{X}_j|\geq (\varepsilon/2k)\cdot |\bar{X}_j|.
\end{equation}
Clearly, since each $\bar{X}_j$ is partitioned by the sets $(\Vars(F'_i)\cap \bar{X}_j)_{i\in T},$ there do exist sets $I$ satisfying (\ref{eq:case1I0def}). That is, any such set $I$ satisfies property (c) mentioned in Step 1 of the algorithm $\mc{A}'$. 

We fix such an $I$ of the smallest possible size and claim that $I$ also satisfies 
\begin{equation}
\label{eq:case1I0prop}
|\Vars(F'_{I})\cap \bar{X}_j|\leq (2\varepsilon/k)\cdot |\bar{X}_j|
\end{equation}
for each $j\in S$. This will therefore prove that this $I$ satisfies properties (a) and (b) in Step 1 of the algorithm $\mc{A}'$. (In fact for proving that (a) holds, it suffices to prove that $|\Vars(F'_{I})\cap \bar{X}_j|\leq 1/3 \cdot |\bar{X}_j|$. Similarly, for proving (b), it suffices to prove (\ref{eq:case1I0prop}) for all but $M/k$ many $j\in S$. But we get better bounds in this case.)

To see this, we argue as follows. If $|I| = 1$, then $F'_I = F'_i$ for some $i\in T$. Then (\ref{eq:case1I0prop}) follows from (\ref{eq:case1Xjshatt}) and (\ref{eq:case1lbdXj}). So we may assume $|I| > 1.$ In this case, as $I$ is the smallest possible set satisfying (\ref{eq:case1I0def}), we must have $|\Vars(F'_{I'})\cap \bar{X}_j|<(\varepsilon/2k)\cdot |\bar{X}_j|$ for each $I'\subsetneq I$ and \emph{each} $j\in S$. Thus, we have for any fixed partition of $I$ into two disjoint sets $I'$ and $I''$ and any $j\in S$, 
\[ 
|\Vars(F'_{I})\cap \bar{X}_j|\leq |\Vars(F'_{I'})\cap \bar{X}_j| + |\Vars(F'_{I''})\cap \bar{X}_j| \leq (\varepsilon/2k + \varepsilon/2k)\cdot |\bar{X}_j| \leq (\varepsilon/k)|\bar{X}_j| 
\]
proving (\ref{eq:case1I0prop}).

This shows that there are suitable sets $I$ satisfying the properties required by Step 1 of the algorithm $\mc{A}'$ and hence $\mc{A}'$ does not output FAILURE on this input.
\end{proof}

Informally, as mentioned earlier, we expect an iteration of $\mc{A}'$ to be $(a,b)$-good if there have been at most $a$ iterations of Type I and at most $b$ iterations of Type II before this one. Further, in each good iteration, we have a reasonable probability of success. These points are made precise in the following claim.

\begin{claim}
\label{clm:case1gooditerns}
Consider an iteration of $\mc{A}'$ on an $(a,b)$-good input where $a < a_0$ and $b < b_0.$ 
\begin{itemize}
\item If the iteration is of Type I, then the probability that the next iteration is not $(a+1,b)$-good is at most $\exp(-\Omega(k))$. Further, the probability of SUCCESS in the current iteration is at least $\eta_1 = 1/100.$
\item If the iteration is of Type II, then the probability that the input produced for the next iteration is not $(a,b+1)$-good is at most $\exp(-\Omega(k)).$ Further, the probability of SUCCESS in the current iteration is at least $\eta_2 = \varepsilon/100k.$
\end{itemize}
\end{claim}

\begin{proof}
Let $(S,T,(\rho'_j)_{j\not\in S})$ be the input to this iteration of $\mc{A}'$ and $I\subseteq T$ the set chosen in Step 1 of $\mc{A}'.$ The inputs to the next iteration are $(S_1,T_1,(\rho'_j)_{j\not\in S_1})$ where $S_1 = S \setminus S'$, $T_1 = T\setminus I$ and $\rho'_j$ for $j\in S'$ are sampled during the current iteration. For each $j\in S_1$, the set corresponding to $\bar{X}_j$ in the next iteration is denoted by $\bar{X}_j'$; formally, $\bar{X}_j' := \bar{X}_j \setminus \Vars(F'_I).$

We start with some preliminary observations. 
Note that by (\ref{eq:case1lbdXj}), we have $|X_j\setminus \bar{X}_j| = o(|\bar{X}_j|)$ and also as a consequence $|X_{j,\gamma}\setminus \bar{X}_{j,\gamma}| = o (|\bar{X}_j|) =  o(|\bar{X}_{j,\gamma}|)$ for any $\gamma\in \{1,2\}.$ In particular, this implies that for any $j\in S$, 
\begin{align}
\delta_{j,1} + \delta_{j,2}  &= \frac{|\Vars(F'_I)\cap \bar{X}_{j,1}|}{|\bar{X}_{j,1}|} + \frac{|\Vars(F'_I)\cap \bar{X}_{j,2}|}{|\bar{X}_{j,2}|}\notag\\ 
&= \frac{|\Vars(F'_I)\cap \bar{X}_{j,1}|}{((1/2)\pm o(1))|\bar{X}_{j}|} + \frac{|\Vars(F'_I)\cap \bar{X}_{j,2}|}{((1/2)\pm o(1))|\bar{X}_{j}|}\notag\\
&= (2\pm o(1)) \frac{|\Vars(F'_I)\cap \bar{X}_{j}|}{|\bar{X}_{j}|} = (2\pm o(1))\cdot \delta_j.\label{eq:case1deltaj}
\end{align}

Since each variable $b_{j,\gamma}$ ($j\in S$, $\gamma\in \{1,2\}$) is independently set to $1$ with probability $\delta_{j,\gamma}$, we see (using (\ref{eq:case1deltaj})) that for each $j\in S$
\begin{equation}
\label{eq:case1jinS'}
\prob{}{b_{j,1} = 1 \text{ or } b_{j,2} = 1} \leq \delta_{j,1} + \delta_{j,2} \leq 3 \delta_j.
\end{equation}

We now proceed to the proof of the statement of the claim.

\paragraph{Type I iteration.} To show that $(S_1,T_1,(\rho'_j)_{j\not\in S_1})$ is $(a+1,b)$-good, we need to show the corresponding versions\footnote{I.e., the version with $a$ replaced by $a+1$ and $\bar{X}_j$ replaced by $\bar{X}'_j$.} of (\ref{eq:case1S}) and (\ref{eq:case1Xjbar}).

\subparagraph{Proof of (\ref{eq:case1S}).} To show (\ref{eq:case1S}), note that it suffices to show
\begin{equation}
\label{eq:probS'largeI}
\prob{}{|S'|\geq \frac{10M}{k}} \leq \exp(-\Omega(k)).
\end{equation}
Note also that in the case of a Type I iteration, $S' =  S''' \cup \hat{S}$ where $S''' =  \{j\in S\ |\ b_{j,1} = 1 \text{ or } b_{j,2} = 1\}.$ By our choice of $I$, we know that $|\hat{S}|\leq M/k.$ Below, we will bound $|S'''\setminus \hat{S}|$. 

Recall that each $b_{j,\gamma}$ ($\gamma\in [2]$) is set to $1$ with probability $\delta_{j,\gamma}.$ Hence, for each $j\in S$, the probability that $j\in S'''$ is at most $\delta_{j,1} + \delta_{j,2},$ which is at most $3\delta_j$ by (\ref{eq:case1jinS'}). Note that for each $j\in S\setminus \hat{S}$, we have $\delta_j \leq (2\varepsilon/k).$ Thus, the expected value of $|S'''\setminus \hat{S}|$ is at most $3\sum_{j\in S\setminus \hat{S}} \delta_j \leq 3\cdot (2\varepsilon/k)\cdot |S| \leq M/k,$ as $|S|\leq M/8$ and $\varepsilon \leq 1.$ Thus, by using our bound on $|\hat{S}|$ and a Chernoff bound (Theorem~\ref{thm:Chernoff}), we obtain
\[
\prob{}{|S'| \geq 10 M/k}\leq \prob{}{|S'''\setminus \hat{S}|\geq 9M/k}\leq \exp(-\Omega(M/k)) \leq \exp(-k)
\]
proving (\ref{eq:probS'largeI}). For the final inequality, we have used the fact that $M \geq m/10 = \Omega(k^{20}).$

\subparagraph{Proof of (\ref{eq:case1Xjbar}).} Fix any $j\in S_1$. Since $\hat{S}\cap S_1=\emptyset$, we have $|\Vars(F'_I)\cap \bar{X}_j|\leq (2\varepsilon/k)|\bar{X}_j|\leq (2\varepsilon/k)\cdot |X_j|$. Thus, the set $\bar{X}'_j = \bar{X}_j \setminus \Vars(F'_I)$ satisfies
\[
|\bar{X}'_j| \geq |\bar{X}_j| - (2\varepsilon/k)|X_j| \geq |X_j|\cdot \left(1-\frac{2\varepsilon}{k}\cdot (a+1)-\frac{b}{M^{0.5}}\right)
\]
thus proving (\ref{eq:case1Xjbar}). This shows that the next iteration is $(a+1,b)$-good except with probability $\exp(-\Omega(k)).$

\subparagraph{Probability of SUCCESS.} The algorithm outputs SUCCESS whenever it finds an $i\in I$ such that $\alpha'_i\neq \beta_i$. Since, as argued in the correctness proof above, $\alpha'_i$ equals the random variable $\alpha_i,$ the algorithm outputs SUCCESS whenever there is an $i$ such that $\alpha_i\neq \beta_i$ for some $i\in I$.

Define the Boolean random variable $\alpha_I$ to be $\bigoplus_{i\in I} \alpha_i$; similarly, define $\beta_I$ to be $\bigoplus_{i\in I}\beta_i.$ For there to be some $i\in I$ such that $\alpha_i \neq \beta_i$, it suffices to have $\alpha_I \neq \beta_I$, or equivalently, whenever $\alpha_I \oplus \beta_I = 1.$  We show that this occurs with probability at least $\eta_1$.

Each $\alpha_I = \bigoplus_{j\in S} \alpha_{I,j}\oplus \bigoplus_{j\not\in S}\alpha_{I,j}$ where $\alpha_{I,j}$ is defined to be $\bigoplus_{i\in I}\alpha_{i,j}.$ Note that for $j\in S$, $\alpha_{I,j}$ is $1$ precisely when exactly one among the randomly chosen variables $x_{j,1}$ and $x_{j,2}$ lands up in $\Vars(F'_I);$ this happens precisely when $b_{j,1}\oplus b_{j,2} = 1.$ Recall that $b_{j,1}$ and $b_{j,2}$ are independent random variables that are $1$ with probability $\delta_{j,1}$ and $\delta_{j,2}$ respectively.

We know that $\delta_{j,\gamma} = |\Vars(F'_I)\cap \bar{X}_{j,\gamma}|/|\bar{X}_{j,\gamma}|$ for each $\gamma\in [2]$. Since $|\Vars(F'_I)\cap \bar{X}_{j,1}| \leq (1/3)|\bar{X}_j|$ (by the choice of $I$ in Step 1 of $\mc{A}'$) and $|\bar{X}_{j,\gamma}|\geq ((1/2)-o(1))\cdot |\bar{X}_j|$ as noted above, we have $\delta_{j,\gamma}\leq (2/3+o(1)).$ In particular, this implies that $\mathrm{Unbias}(b_{j,\gamma}) \geq \min\{\delta_{j,\gamma}, 1-\delta_{j, \gamma}\} \geq \delta_{j,\gamma}/3$\footnote{Recall (Section~\ref{sec:misc}) that $\mathrm{Unbias}(b_{j,1}) = \min\{\prob{}{b_{j,1} = 1}, \prob{}{b_{j,1}=0}\}$.}. 

Thus, Proposition~\ref{prop:xorunbias} and (\ref{eq:case1deltaj}) imply that for each $j\in S$,
\begin{align*}
  \mathrm{Unbias}(\alpha_{I,j}) = \mathrm{Unbias}(b_{j,1}\oplus b_{j,2}) &\geq \min\{\frac{1}{6}(\delta_{j,1} + \delta_{j,2}),\frac{1}{10}\}  \geq\min\{\frac{1}{3}(1-o(1))\delta_j,\frac{1}{10}\}.
\end{align*}

Since the random variables $\alpha_{I,j}$ are independent and this is a Type I iteration (meaning that $\sum_{j\in S} \delta_j \geq 1$), we see from Proposition~\ref{prop:xorunbias} that $\mathrm{Unbias}(\alpha_I) \geq \min\{(1/6)\sum_{j\in S} (1-o(1))\cdot \delta_j, 1/10\} \geq 1/100.$ This implies that the probability that $\alpha_I \oplus \beta_I = 1$ is at least $\eta_1=1/100$ as claimed.

\paragraph{Type II iteration.} The proof is, to a large extent, similar to the proof in the Type I case. So we merely point out the main differences.

\subparagraph{Proof of (\ref{eq:case1S}).} To show (the corresponding version of) (\ref{eq:case1S}), note that it suffices to show
\begin{equation}
\label{eq:probS'largeII}
\prob{}{|S'|\geq 10\cdot M^{0.5}} \leq \exp(-\Omega(k)).
\end{equation}
Note that in this case, $S' = S'' \cup S''' $ where $S''$ is as defined in $\mc{A}'$ and	 $S''' = \{j\in S\ |\ b_{j,1} = 1 \text{ or } b_{j,2} = 1\}.$ 

To bound $S''$, we note that since we are dealing with a Type II iteration, we have $\sum_{j\in S} \delta_j < 1.$ In particular, the number of $j\in S$ such that $\delta_j \geq 1/M^{0.5}$ is at most $M^{0.5}$. Hence, we have shown that $|S''|\leq M^{0.5}.$

To bound $S''',$ we note as in the Type I case that the expected size of $|S'''|$ is at most $3\sum_j\delta_j\leq 3.$ In particular, by Theorem~\ref{thm:Chernoff} (Item 2), we have
\[
\prob{}{|S'''| \geq M^{0.5}}\leq \exp(-\Omega(M^{0.5})) \leq \exp(-\Omega(k)).
\]
Thus, we see that with probability at least $1-\exp(-\Omega(k)),$ both $S''$ and $S'''$ have size at most $M^{0.5}$ and in this case, $|S'|\leq |S''| + |S'''|  \leq 2 M^{0.5}.$ This proves (\ref{eq:probS'largeII}).

\subparagraph{Proof of (\ref{eq:case1Xjbar}).} By our choice of $S''$, the set $S_1 = S\setminus S'$ does not contain any $j$ such that $\delta_j \geq 1/M^{0.5}.$ Hence, we see that for all $j\in S_1$, $|\Vars(F'_I)\cap \bar{X}_j|\leq |\bar{X}_j|/M^{0.5}$. Hence, we have
\[
|\bar{X}'_j| \geq |\bar{X}_j| - \frac{|\bar{X}_j|}{M^{0.5}} \geq  |X_j|\cdot \left(1-\frac{2\varepsilon}{k}\cdot a - \frac{b+1}{M^{0.5}}\right)
\]
as desired, for each $j\in S_{1}$. 

\subparagraph{Probability of SUCCESS.} We argue as in the Type I case that $\mathrm{Unbias}(\alpha_I)$ is large. 

We know by our choice of $I$ that there is at least one $j$ such that $\delta_j\in [\varepsilon/2k,2\varepsilon/k].$ Hence, exactly as in the Type I case, we see that
$\mathrm{Unbias}(\alpha_I) \geq \min\{(1/6)\sum_j (1-o(1))\cdot \delta_j, 1/10\} \geq \varepsilon/12k.$ In particular, this implies that the probability that $\alpha_I \oplus \beta_I = 1$ is at least $\eta_2 = \varepsilon/100k$, which yields the statement of the claim. 
\end{proof}

We are now ready to upper bound the probability that the algorithm outputs FAILURE on input $(S_0,T_0)$. We prove a more general statement. Assume that $(S,T,(\rho'_j)_{j\not\in S})$ is any $(a,b)$-good input to $\mc{A}'$ where $a \leq a_0$ and $b \leq b_0.$ Then, we claim that
\begin{equation}
\label{eq:case1indn}
\prob{}{\mc{A}'(S,T,(\rho'_j)_{j\not\in S}) \text{ outputs FAILURE}} \leq ((a_0-a) + (b_0-b))\cdot \exp(-\Omega(k)) + (1-\eta_1)^{a_0-a} + (1-\eta_2)^{b_0-b}
\end{equation}
where $\eta_1,\eta_2$ are as defined in the statement of Claim~\ref{clm:case1gooditerns}. We prove the above by downward induction on $a+b$ where $a\leq a_0$ and $b\leq b_0$. 

The base case of the induction is defined to be the case when either $a = a_0$ or $b = b_0.$ In this case, the claim is trivial since the right hand side of (\ref{eq:case1indn}) is more than $1$. 

Now consider the case when $a < a_0$ and $b < b_0.$ Let us analyze the behaviour of the algorithm $\mc{A}'$ on an $(a,b)$-good input $(S,T,(\rho'_j)_{j\not\in S}).$ Assume that the current iteration is of Type I (the other case is similar). Claim~\ref{clm:case1failureiter} tells us that the algorithm does not output FAILURE in this iteration. By Claim~\ref{clm:case1gooditerns}, the probability that the algorithm produces an input (for the next iteration) that is \emph{not} $(a+1,b)$-good is at most $\exp(-\Omega(k))$; in this case, we give up and assume that the algorithm subsequently will output FAILURE. Finally, by Claim~\ref{clm:case1gooditerns}, the probability that the algorithm does not output SUCCESS in this round is at most $(1-\eta_1)$. In particular, this implies that the probability that the algorithm does not output SUCCESS \emph{and} produces an $(a+1,b)$-good input for the next iteration is also at most $(1-\eta_1)$. However, conditioned on this event, we can use the inductive hypothesis to bound the probability of FAILURE. By the union bound and the induction hypothesis, we get
\begin{align*}
&\prob{}{\mc{A}'(S,T,(\rho'_j)_{j\not\in S}) \text{ outputs FAILURE}}\\ 
&\leq\exp(-\Omega(k)) + (1-\eta_1)\cdot (((a_0-(a+1)) + (b_0-b))\cdot \exp(-\Omega(k)) + (1-\eta_1)^{a_0-(a+1)} + (1-\eta_2)^{b_0-b})\\
&\leq \exp(-\Omega(k)) + ((a_0-(a+1)) + (b_0-b))\cdot \exp(-\Omega(k)) + (1-\eta_1)\cdot (1-\eta_1)^{a_0-(a+1)} + (1-\eta_2)^{b_0-b}\\
&\leq ((a_0-a) + (b_0-b))\cdot \exp(-\Omega(k)) + (1-\eta_1)^{a_0-a} + (1-\eta_2)^{b_0-b}
\end{align*}
which completes the induction. 

In particular, since the initial input $(S_0,T_0)$ is $(0,0)$-good, we see that 
\begin{align*}
\prob{}{\mc{A}' \text{ outputs FAILURE on } (S_0,T_0)} &\leq (a_0+b_0)\cdot \exp(-\Omega(k)) + (1-\eta_1)^{a_0} + (1-\eta_2)^{b_0}\\
&\leq (a_0+b_0)\cdot \exp(-\Omega(k)) + \exp(-\Omega(a_0)) + \exp(-\Omega(b_0\varepsilon/k))\\
&\leq M\cdot \exp(-\Omega(k)) + \exp(-\Omega(k)) + \exp(-\Omega(M^{0.45}\varepsilon/k))\\
&\leq \exp(-\Omega(k))
\end{align*}
where the second inequality follows from the definition of $\eta_1$ and $\eta_2$, the third from the definition of $a_0$ and $b_0$, and the last follows from the fact that $M \leq m^{\Delta'} \leq  \exp(k^{0.02+o(1)}),$\footnote{Note that there are at most $m^{\Delta'}$ many $X^{(\Delta)}$-segments in $G^{(\Delta')}$. } $\varepsilon \geq 1/M^{0.25},$ and $M = \Omega(m) = \Omega(k^{20}).$ This implies (\ref{eq:case1alphabeta}) and hence completes the proof of Case 1.

\subsubsection{Case 2}
\label{sec:case2}
By renaming our segments if necessary, we assume that $X_1, \dots, X_{M/2}$ are the segments that are not $(1/4)$-shattered. That is, for each of these segments $X_j$ ($j\in[M/2]$), there exists an $F_{i_j}'$ ($i_j\in[t]$) such that $\abs{\Vars(F_{i_j}')\cap X_j} \geq (1/4)\abs{X_j}$. By an averaging argument, we can say that there exists $b_j\in\{1,2\}$ such that $\abs{\Vars(F_{i_{j}}')\cap X_{j,b_j}} \geq (1/4)\abs{X_{j,b_j}}$. 

For each $j \leq M/2$, the restriction $\rho_j$ is sampled according to the algorithm $\mc{A}$ described in Section~\ref{sec:restriction}. We first condition on the choice (among the options $E_1,E_2,E_3$) made by the algorithm $\mc{A}$ in sampling each $\rho_j.$ We say that a $\rho_j$ is \emph{good} if the algorithm $\mc{A}$ decides on option $E_{b_j}\in \{E_1,E_2\}$ in sampling $\rho_j.$ For each $j$, $\rho_j$ is good with probability $1/3$. Let $\mc{E}_2$ denote the event that no $\rho_j$ is good for $j\in [M/2]$ and let $\bar{\mc{E}_2}$ be the complement of $\mc{E}_2$. Clearly, we have
\begin{equation}
\label{eq:case2chernoff}
\prob{\rho}{\mc{E}_2}\leq (2/3)^{M/2} \leq \exp(-\Omega(M)).
\end{equation}

Assume that $\mc{E}_2$ does not occur. Then, there is a $j\in[M/2]$ such that $\rho_j$ is good. By renaming segments once more, we can assume that $j=1$. Further, by renaming the $F_{i}'$s, we assume that $i_1 =1$; in particular, we have $\abs{\Vars(F_1')\cap X_1}\geq \abs{X_1}/4$. Finally, we can similarly assume that $b_1 = 1$, which implies that $\abs{\Vars(F_1')\cap X_{1,1}}\geq \abs{X_{1,1}}/4.$

We condition on any choice of the restrictions $\rho_j$ for $j > 1$; this fixes the sets $\tilde{Y}'_j$ and $\tilde{Z}'_j$ defined above for $j > 1$. Conditioned on all these events, we see that the random restriction $\rho_1$ sets all the variables in $X_{1,2}$ to $0$ deterministically, and hence can now be considered an $(X_{1,1},Y_{1,1},Z_{1,1})$-restriction, which we will call $\rho_1'$ for clarity. Note that $X_{1,1}, Y_{1,1}$ and $Z_{1,1}$ are $X^{(\Delta-1)}$-, $Y^{(\Delta-1)}$- and $Z^{(\Delta-1)}$-clones respectively.

Recall that the $X^{(\Delta-1)}$-clone $X_{1,1}$ is made up of $m$ $X^{(\Delta-1)}$-segments $X_{1,1,1},\ldots,X_{1,1,m}$. We similarly have $Y_{1,1} = \bigcup_{j\in [m]} Y_{1,1,j}$ and $Z_{1,1} =  \bigcup_{j\in [m]} Z_{1,1,j}$. 

The following claim is proved by a standard averaging argument. 

\begin{claim}\label{clm:popularity}
  There exist at least $m/8$ many $j\in [m]$ such that $F'_1$ is $1/8$-heavy in $X_{1,1,j}$ i.e., $|\Vars(F'_1)\cap X_{1,1,j}|\geq (1/8)\cdot |X_{1,1,j}|.$ 
\end{claim}
Assuming Claim~\ref{clm:popularity}, we will first show how we can invoke the induction hypothesis to prove the desired result.

By renaming if needed, let the $m/8$ segments from Claim~\ref{clm:popularity} be $X_{1,1, 1}, \dots, X_{1,1, m/8}$. The restriction $\rho_1'=\rho_{1,1}\circ\cdots\circ\rho_{1,m}$ where $\rho_{1,\ell}$ for $\ell\in[m]$ is an $(X_{1,1,\ell},Y_{1,1,\ell},Z_{1,1,\ell})$-restriction sampled using the algorithm $\mc{A}$ in Section~\ref{sec:restriction}. We further condition on any choice of restrictions $\rho_{1,\ell}$ for all $\ell > m/8$. This fixes the sets $\tilde{Y}_{1,1,\ell}' := \rho_{1,\ell}(\Vars(F'_1)\cap X_{1,1,\ell})\cap Y_{1,1,\ell}$ and $\tilde{Z}_{1,1,\ell}'$ (defined similarly) for all $\ell > m/8$. Let $F''$ denote the restricted formula thus obtained and $\rho_1''$ the random restriction after this conditioning (i.e. $\rho_1'' = \rho_{1,1}\circ \cdots \circ \rho_{1,m/8}).$

We proceed to analyze $\relrk_{(\tilde{Y}'\cup Y',\tilde{Z}'\cup Z')}(F'|_\rho)$. Since $\tilde{Y}'\cup Y'$ (resp.\ $\tilde{Z}'\cup Z'$) can be partitioned as $\bigcup_{i\in [t]} \hat{Y}_i$ (resp.\ $\bigcup_{i\in [t]}\hat{Z}_i$), we know (using Proposition~\ref{prop:relrk} Item 3) that for any choice of $\rho$ consistent with our choices so far,
\begin{equation}
\label{eq:case2relrkF'}
\relrk_{(\tilde{Y}'\cup Y',\tilde{Z}'\cup Z')}(F'|_\rho) = \prod_{i\in [t]}\relrk_{(\hat{Y}_i,\hat{Z}_i)}(F'_i|_\rho) \leq \relrk_{(\hat{Y}_1,\hat{Z}_1)}(F'_1|_\rho) = \relrk_{(\hat{Y}_1,\hat{Z}_1)}(F''|_{\rho_{1}''}).
\end{equation}
To bound the latter term, we invoke the induction hypothesis on $F''$ with $M $ replaced by $M'= m/8$ and $\varepsilon$ replaced by $\varepsilon' = 1/8$. Here $\rho''$ is a $(X'', Y'', Z'')$-restriction where
\begin{align*}
X''&=\cup_{j\in[m/8]}X_{1,1,j}\\ Y''&=\cup_{j\in[m/8]}Y_{1,1,j}\\ Z''&=\cup_{j\in[m/8]}Z_{1,1,j}
\end{align*} 
and $\Vars(F'')=\bar{X}\cup\bar{Y}\cup\bar{Z}$ where
\begin{align*}
  \bar{X} &= \Vars(F_1')\cap X''\\
  \bar{Y} &= \left(\cup_{j>1}\tilde{Y}_j'\right)\cup\left(\cup_{\ell>m/8}\tilde{Y}'_{1,1,j}\right)\cup \hat{Y}_{1,0}\\
  \bar{Z} &= \left(\cup_{j>1}\tilde{Z}_j'\right)\cup\left(\cup_{\ell>m/8}\tilde{Z}'_{1,1,j}\right)\cup \hat{Z}_{1,0}.
\end{align*}
Note that by Claim~\ref{clm:popularity}, it follows that $|\Vars(F'')\cap X_{1,1,j}|\geq \varepsilon'\cdot |X_{1,1,j}|$ for each $j\in [M']$ and hence the induction hypothesis for product-depth $(\Delta -1)$ is applicable.

From the induction hypothesis, we see that 
\begin{align*}
  \prob{}{\relrk_{(\hat{Y}_1,\hat{Z}_1)}(F''|_{\rho''}) \geq \exp(-k^{0.1}/(\Delta-1))} &\leq (\Delta-1) \exp(-k^{0.1}/(\Delta-1))
  \end{align*}
Using (\ref{eq:case2relrkF'}), we get
    \begin{align*}
\prob{\rho}{\relrk_{(\tilde{Y}'\cup Y',\tilde{Z}'\cup Z')}(F'|_{\rho}) \geq  \exp(-k^{0.1}/(\Delta-1))\mid \bar{\mc{E}}_2} &\leq (\Delta-1) \exp(-k^{0.1}/(\Delta-1)).
\end{align*}

By (\ref{eq:case2chernoff}), we thus have
\begin{align*}
  \prob{}{\relrk_{(\tilde{Y}'\cup Y',\tilde{Z}'\cup Z')}(F'|_{\rho}) \geq  \exp(-k^{0.1}/(\Delta-1))}&\leq (\Delta-1) \exp(-k^{0.1}/(\Delta-1)) + \exp(-\Omega(M))\\
    &\leq \Delta \exp(-k^{0.1}/(\Delta-1)).
\end{align*}
The last inequality holds since $M \geq m/10 > k$. This finishes the proof of Case 2 modulo the (standard) proof of Claim~\ref{clm:popularity}, which we give below for completeness.
\begin{proof}[Proof of Claim~\ref{clm:popularity}]
  For $j\in [m]$, let $p_j$ denote $\abs{\Vars(F_{1}')\cap X_{1,1,j}}/|X_{1,1,j}| = (m\abs{\Vars(F_{1}')\cap X_{1,1,j}})/|X_{1,1}|\,$. Since $\abs{\Vars(F_{1}')\cap X_{1}}\geq (1/4)\abs{X_{1}}$, we get that $\sum_{j=1}^mp_j \geq m/4$. Assume that the number of $j$ such that $p_j \geq 1/8$ is strictly less than $m/8$. Then, we have
  \begin{align*}
  \sum_{j\in [m]} p_j \leq \sum_{j: p_j \geq 1/8} 1 + \sum_{j: p_j \leq 1/8} (1/8)  < (m/8) + 1/8 \sum_{j: p_j \leq 1/8} 1 \leq m/8 + (1/8)\cdot m = m/4,
  \end{align*}
  yielding $\sum_{j\in [m]} p_j < m/4,$ contradicting our assumption on $\sum_{j\in [m/8]} p_j.$ This proves the claim.
 \end{proof} 

\subsubsection{Case 3}
\label{sec:case3}
Recall that in Case 3, we have that there is some $F'_i$ ($i\in [t]$) that is $(\varepsilon/k)$-heavy in at least $M/k$ many segments $X_j.$ We will prove (\ref{eq:relrkPigate}) in this case using induction. W.l.o.g., assume that $i=1$.

Since $\tilde{Y}'\cup Y'$ (resp.\ $\tilde{Z}'\cup Z'$) can be partitioned as $\bigcup_{i\in [t]} \hat{Y}_i$ (resp.\ $\bigcup_{i\in [t]}\hat{Z}_i$), we know (Proposition~\ref{prop:relrk} Item 3) that for any choice of $\rho,$
\begin{equation}
\label{eq:case4relrkF'}
\relrk_{(\tilde{Y}'\cup Y',\tilde{Z}'\cup Z')}(F'|_\rho) = \prod_{i\in [t]}\relrk_{(\hat{Y}_i,\hat{Z}_i)}(F'_i|_\rho) \leq \relrk_{(\hat{Y}_1,\hat{Z}_1)}(F'_1|_\rho).
\end{equation}
It is therefore sufficient to bound $\relrk_{(\hat{Y}_1,\hat{Z}_1)}(F'_1|_\rho)$.

W.l.o.g., we assume that $F'_1$ is $(\varepsilon/k)$-heavy in segments $X_1,\ldots,X_{M/k}$. Each $X_j$ ($j\in [M/k]$) consists of two half-segments, i.e., $X_j = X_{j,1} \cup X_{j,2}$.  By averaging, we know that there is some $\gamma\in [2]$ such that $F'_1$ is $(\varepsilon/k)$-heavy in at least half of $X_{1,\gamma},\ldots,X_{M/k,\gamma}$. W.l.o.g.\ we assume that $\gamma = 1$. By renaming the segments, let us assume that $F'_1$ is $(\varepsilon/k)$-heavy in $X_{1,1},\ldots,X_{M/2k,1}$. 

For each $j \leq M/2k$, the restriction $\rho_j$ is sampled according to the algorithm $\mc{A}$ described in Section~\ref{sec:restriction}. Let us say that $\rho_j$ is \emph{good} if the algorithm $\mc{A}$ decides on option $E_1$ in sampling $\rho_j.$ For each $j$, $\rho_j$ is good with probability $1/3$. For each $j\in [M/2k],$ we condition on whether or not $\rho_j$ is good. Let $\mc{E}_3$ denote the event that the number of good $\rho_j$ is at most $M/8k$. Since the different $\rho_j$ are sampled independently, a Chernoff bound (Theorem~\ref{thm:Chernoff} Item 1) tells us that 
\begin{equation}
\label{eq:case4chernoff}
\prob{\rho}{\mc{E}_3}\leq \exp(-\Omega(M/k)).
\end{equation}

Assume that the event $\mc{E}_3$ does not occur. By renaming segments once more, we may assume that $\rho_j$ is good for each $j\in [M/8k]$. We condition on any choice of the restrictions $\rho_j$ for $j > M/8k$; this fixes the sets $\tilde{Y}'_j$ and $\tilde{Z}'_j$ for $j > M/8k$. Also, note that since we have conditioned on choosing the option $E_1$ for sampling restrictions $\rho_1,\ldots,\rho_{M/8k}$, all variables in $X_{j,2}$ ($j\in [M/8k]$) are set to $0$ with probability $1$. We will therefore think of $\rho_j$ as an $(X_{j,1},Y_{j,1},Z_{j,1})$-restriction for the rest of the proof.

We will now prove our result by considering two subcases. Recall that each half-segment $X_{j,1}$ is an $X^{(\Delta-1)}$-clone and hence a union of $m$ $X^{(\Delta-1)}$-segments $X_{j,1,1},\ldots,X_{j,1,m}.$ We will refer to these as \emph{sub-segments} of $X_{j,1}.$

\paragraph*{Case (3a): There are at least $M/16k$ many $j\in [M/8k]$ such that $F'_1$ is at least $(\varepsilon/2k)$-heavy in at least $m^{\what}$ many sub-segments of $X_{j,1}$.} 

Without loss of generality, let us assume that the hypothesis of Case (3a) holds for the half-segments $X_{1,1},\ldots,X_{M/16k,1}.$ That is, $F'_1$ is $(\varepsilon/2k)$-heavy in at least  $m^{\what}$ many sub-segments of $X_{j,1}$ for each $j\in [M/16k]$. Here, we will be able to apply induction on $F'_1$ (which is a $\Sigma(\Pi\Sigma)^{(\Delta-1)}$ formula) to obtain an upper bound on its relative rank under the random restriction.

We will apply the induction hypothesis with $M$ replaced by $M' = M/16k \times m^{\what}$ and $\varepsilon$ replaced by $\varepsilon' = \varepsilon/2k$. We claim that we have
$\varepsilon' \geq (1/M')^{0.25}$ and $M'\geq m/10$ (which are needed to apply induction). The latter inequality is trivial as $M'\geq M \geq m/10$. For the former, note that we have
\[
\frac{1}{(\varepsilon')^4} = \frac{(2k)^4}{\varepsilon^4} \leq (2k)^4 M \leq \frac{M \sqrt{m}}{16k}  = M'
\]
where the first inequality uses the fact that $\varepsilon \geq 1/M^{0.25}$ and the last inequality uses the fact that $k \leq m^{0.05}.$ This shows that $\varepsilon'\geq (1/M')^{0.25}$ as required. 

We apply induction after some further processing. By renaming if necessary, let us assume that $F_1'$ is heavy in the first $\sqrt{m}$ sub-segments $X_{j,1,1}, \cdots, X_{j,1,m^{0.5}}$ of $X_{j,1}$ for each $j\in [M/16k]$. Conditioned on what is known about $\rho$ so far, for each $j\in [M/8k]$, the restriction $\rho_j$ can be written as $\rho_j=\rho_{j,1}\circ\cdots\circ\rho_{j,m}$ where each $\rho_{j,\ell}$ for $\ell \in[m]$ is an $(X_{j,1,\ell},Y_{j,1,\ell},Z_{j,1,\ell})$-restriction sampled as per the sampling algorithm $\mc{A}$. We further condition on any choice of the restrictions $\rho_{j_1}$ for all $j_1 > M/16k$ and $\rho_{j_2, \ell}$  for all $j_2\in[M/16k]$ and $\ell > m^{0.5}$. This fixes the sets $\tilde{Y}_{j_1}'$, $\tilde{Z}_{j_1}'$ ($M/8k < j_1 \leq M/16k$) and $\tilde{Y}'_{j_2,1, \ell} := \rho(\Vars(F'_1)\cap X_{j_2,1,\ell})\cap Y_{j_2,1,\ell}$ and $\tilde{Z}'_{j_2,1, \ell} := \rho(\Vars(F'_1)\cap X_{j_2,1,\ell})\cap Z_{j_2,1,\ell}$ ($j_2\in [M/16k], \ell > m^{0.5}$). 

Conditioned on our choices so far, the restriction $\rho$ can now be identified an $(X'', Y'', Z'')$-restriction $\rho''$ where
\begin{align*}
X''&=\cup_{j\in[M/8k]}\cup_{\ell\in[m^{0.5}]}X_{j,1,\ell}\\ Y''&=\cup_{j\in[M/8k]}\cup_{\ell\in[m^{0.5}]}Y_{j,1,\ell}\\ Z''&=\cup_{j\in[M/8k]}\cup_{\ell\in[m^{0.5}]}Z_{j,1,\ell}.
\end{align*}
and the restricted formula, which we denote by $F_1''$, satisfies $\Vars(F''_1)=\bar{X}\cup\bar{Y}\cup\bar{Z}$ where
\begin{align*}
  \bar{X} &= \Vars(F_1')\cap X''\\
  \bar{Y} &= \left(\cup_{j> M/16k}\tilde{Y}_j'\right)\cup\left(\cup_{j\in [M/16k]}\cup_{\ell>m^{0.5}}\tilde{Y}'_{j,1,\ell}\right)\cup \hat{Y}_{1,0}\\
  \bar{Z} &= \left(\cup_{j>M/16k}\tilde{Z}_j'\right)\cup\left(\cup_{j \in [M/16k]}\cup_{\ell > m^{0.5}}\tilde{Z}'_{j,1,\ell}\right)\cup \hat{Z}_{1,0}.
\end{align*}

Based on the discussion above, we can now invoke the induction hypothesis with $M' = (M/16k)m^{0.5}$ and $\varepsilon' = \varepsilon/2k$ on $F''_1$ to obtain
\[
\prob{\rho''}{\relrk_{(\hat{Y}_1,\hat{Z}_1)}(F''_1|_\rho) \geq \exp(-k^{0.1}/(\Delta-1))}\leq (\Delta-1)\exp(-k^{0.1}/(\Delta-1)).
\]
Since the above holds for an arbitrary sequence of fixings after it was determined that $\mc{E}_3$ did not hold, we have thus shown the following.

\begin{equation}
\label{eq:case4aresult}
\prob{\rho}{\relrk_{(\hat{Y}_1,\hat{Z}_1)}(F'_1|_\rho) \geq \exp(-k^{0.1}/(\Delta-1))\ |\ \bar{\mc{E}}_3}\leq (\Delta-1)\exp(-k^{0.1}/(\Delta-1)).
\end{equation}

We will use this below, after the proof of Case 3(b), to finish the proof in this case.

\paragraph*{Case (3b): There are at least $M/16k$ many $j\in [M/8k]$ such that $F'_1$ is at least $(\varepsilon/2k)$-heavy in less than $m^{\what}$ many sub-segments of $X_{j,1}$.} 
W.l.o.g., assume that for each $j\in [M/16k]$, there is a $h_j < \sqrt{m}$ such that $F_1'$ is $(\varepsilon/2k)$-heavy only in the sub-segments $X_{j,1,1},\ldots,X_{j,1,h_j}.$

For each $j \in [M/16k]$, let $W_j$ be $X_{j,1} \cap \Vars(F_1')$. By our assumption in Case 3, we know that

\begin{equation}
\label{eq:lower}
|W_j|  \geq \frac{\varepsilon}{k} \cdot |X_{j,1}|~~~~\forall j \in [M/16k]
\end{equation}
Also, we know that for any $\ell \in \{h_j+1, \ldots, m\}$, $|X_{j,1,\ell} \cap W_j| < \frac{\varepsilon}{2k}\cdot |X_{j,1,\ell}|$ for each $j \in [M/16k]$. Therefore, we get that

\begin{equation}
\label{eq:upper}
\sum_{\ell > h_j} |X_{j,1,\ell} \cap W_j| \leq \frac{\varepsilon}{2k}\cdot |X_{j,1}| ~~~~\forall j \in [M/16k]
\end{equation}

Using Equations~(\ref{eq:lower}) and (\ref{eq:upper}) we get that
\begin{align*}
\sum_{\ell \in [h_j]} |X_{(j,1),\ell} \cap W_j| \geq \frac{\varepsilon}{2k}\cdot |X_{j,1}|~\mbox{ for each }  j \in [M/16k].
\end{align*}
Thus by averaging and by using the fact that $h_j \leq m^{\what}$, we get that for every $j \in [M/16k]$ there exists an $\ell_j \in [h_j]$ such that $|X_{j,1,\ell_j} \cap W_j| \geq \frac{\varepsilon}{2k}\cdot \frac{|X_{j,1}|}{m^{\what}}$. By renaming if necessary, let $\ell_j = 1$ for all $j\in[M/16k]$.

Notice that in fact $|X_{j,1}|  = m\cdot |X_{j,1,\ell}|$ for any $\ell \in [m]$. Therefore, we get 

\begin{align}
\label{eq:new-e}
|X_{j,1,1} \cap \Vars(F'_{1})| = |X_{j,1,1} \cap W_j| & \geq \frac{\varepsilon}{2k}\cdot \frac{|X_{j,1,1}| \cdot m}{m^{\what}} = \frac{\varepsilon}{2k} \cdot m^{\notwhat} \cdot |X_{j,1,1}| & \forall j \in [M/16k]
\end{align}

We will apply induction on the formula $F'_1$ and the sub-segments $X_{j,1,1}$ ($j\in [M/16k]$). For the induction, the parameter $M$ will thus be replaced by $M' = M/16k$ and $\varepsilon$ by $\varepsilon' = \frac{\varepsilon}{2k} \cdot m^{\notwhat}$ (we can take this $\varepsilon'$ by (\ref{eq:new-e}) above). To check that induction is possible with these parameters, we need to ensure  that $\varepsilon' \geq 1/M'^{0.25}$ and $M'\geq m/10$. The first inequality follows as in Case (3a). Formally, we note that
\[
\frac{1}{(\varepsilon')^4} = \frac{(2k)^4}{\varepsilon^4 m^2} \leq \frac{M}{16k}  = M'
\]
where the inequality follows from the fact that $\varepsilon \geq 1/M^{0.25}$ and $k = m^{0.05}.$ This shows that $\varepsilon' \geq (1/M')^{0.25}.$

Now we will show that $M'\geq m/10$. We know that $M\geq m/10$, but in this case we will be able to get a better lower bound on the value of $M$, which will then give us the intended lower bound on $M'$. For this, first observe that $\varepsilon' = (\varepsilon/2k)\cdot m^{\what}$. Note that by (\ref{eq:new-e}) and the fact that $|X_{j,1,1} \cap \Vars(F'_{1})|\leq |X_{j,1,1}|$, we get 
$$(\varepsilon/2k)\cdot m^{\what}\leq 1 \Leftrightarrow \frac{1}{\varepsilon} \geq \frac{\sqrt{m}}{2k}.$$
 We also know that $\varepsilon \geq 1/M^{0.25}$ and therefore, $M \geq 1/\varepsilon^4 \geq m^2/(2k)^4$. As $M' = M/16k$, we get that $M' \geq \frac{m^2}{8k \cdot (2k)^4 } \geq m$. This gives us (a bound that is slightly better than) the desired bound on $M'$. 

We now apply induction. As in Case (3a), some processing is needed. Note that conditioned on what is known about $\rho$ so far, for each $j\in [M/8k]$, the restriction $\rho_j$ sets all  $\rho_j=\rho_{j,1}\circ\cdots\circ\rho_{j,m}$ where each $\rho_{j,\ell}$ for $\ell \in[m]$ is an $(X_{j,1,\ell},Y_{j,1,\ell},Z_{j,1,\ell})$-restriction sampled as per the sampling algorithm $\mc{A}$. We further condition on any choice the restrictions $\rho_{j_1}$ for all $j_1 > M/16k$ and $\rho_{j_2, \ell}$  for all $j_2\in[M/16k]$ and $\ell > 1$. This fixes the sets $\tilde{Y}_{j_1}'$, $\tilde{Z}_{j_1}'$ ($M/8k < j_1 \leq M/16k$) and $\tilde{Y}'_{j_2,1, \ell} := \rho(\Vars(F'_1)\cap X_{j_2,1,\ell})\cap Y_{j_2,1,\ell}$ and $\tilde{Z}'_{j_2,1, \ell} := \rho(\Vars(F'_1)\cap X_{j_2,1,\ell})\cap Z_{j_2,1,\ell}$ ($j_2\in [M/16k], \ell > 1$). 

Conditioned on our choices so far, the restriction $\rho$ can be identified with an $(X'', Y'', Z'')$-restriction $\rho''$ where
\begin{align*}
X''&=\cup_{j\in[M/16k]}X_{j,1,1}\\ Y''&=\cup_{j\in[M/16k]}Y_{j,1,1}\\ Z''&=\cup_{j\in[M/16k]}Z_{j,1,1}.
\end{align*}
and the restricted formula $F_1''$ satisfies $\Vars(F''_1)=\bar{X}\cup\bar{Y}\cup\bar{Z}$ where
\begin{align*}
  \bar{X} &= \Vars(F_1')\cap X''\\
  \bar{Y} &= \left(\cup_{j> M/16k}\tilde{Y}_j'\right)\cup\left(\cup_{j\in [M/16k]}\cup_{\ell>1}\tilde{Y}'_{j,1,\ell}\right)\cup \hat{Y}_{1,0}\\
  \bar{Z} &= \left(\cup_{j>M/16k}\tilde{Z}_j'\right)\cup\left(\cup_{j \in [M/16k]}\cup_{\ell > 1}\tilde{Z}'_{j,1,\ell}\right)\cup \hat{Z}_{1,0}.
\end{align*}

Based on the discussion above, we can now invoke the induction hypothesis with $M'$ and $\varepsilon'$ (as defined above) on $F''_1$ to obtain
\[
\prob{\rho}{\relrk_{(\hat{Y}_1,\hat{Z}_1)}(F''_1|_{\rho''}) \geq \exp(-k^{0.1}/(\Delta-1))}\leq (\Delta-1)\exp(-k^{0.1}/(\Delta-1)).
\]
Hence, we obtain as in Case (3a),

\begin{equation}
\label{eq:case4bresult}
\prob{\rho}{\relrk_{(\hat{Y}_1,\hat{Z}_1)}(F'_1|_\rho) \geq \exp(-k^{0.1}/(\Delta-1))\ |\ \bar{\mc{E}}_3}\leq (\Delta-1)\exp(-k^{0.1}/(\Delta-1)).
\end{equation}

We now see how to finish the proof in both Cases (3a) and (3b). Using (\ref{eq:case4aresult}) and (\ref{eq:case4bresult}), we see that in each of the cases (3a) and (3b),

\begin{align*}
 \prob{\rho}{\relrk_{(\hat{Y}_1,\hat{Z}_1)}(F'_1|_{\rho}) \geq \exp(-k^{0.1}/(\Delta - 1))}&\leq (\Delta-1)\exp(-k^{0.1}/(\Delta - 1)) + \prob{\rho}{\mc{E}_{3}} \\
 &\leq (\Delta-1)\exp(-k^{0.1}/(\Delta - 1)) + \exp(-\Omega(M/k))\\
 &\leq \Delta \exp(-k^{0.1}/(\Delta - 1))
\end{align*}
where the second inequality follows from (\ref{eq:case4chernoff}) and the last inequality from the fact that $M = \Omega(m) = \Omega(k^{20}).$

By (\ref{eq:case4relrkF'}), we thus get 
\[
\prob{\rho}{\relrk_{(\tilde{Y}'\cup Y',\tilde{Z}'\cup Z')}(F'|_{\rho}) \geq \exp(-k^{0.1}/(\Delta - 1))} \leq \Delta \exp(-k^{0.1}/(\Delta - 1))
\]
proving inequality (\ref{eq:relrkPigate}) in this case. This completes the proof in Case 4.

\bibliography{localref}

\end{document}